\documentclass[
superscriptaddress,
 amsmath,amssymb,
 aps,
pra,
onecolumn,
]{revtex4-2}
\usepackage{graphicx} % Required for inserting images
\usepackage{amsmath}
\usepackage{amsfonts}
\usepackage{amssymb}
\usepackage{amsthm}
\usepackage{braket}
\usepackage{subcaption} 
\DeclareMathOperator{\tr}{tr}
\usepackage{bm}
\usepackage{bbm}
\usepackage{color}
\usepackage{algpseudocode}
\usepackage[ruled,vlined]{algorithm2e}
\usepackage{algorithm2e}

\usepackage{xcolor}
\usepackage[colorlinks=true,
            citecolor=NatureTeal,
            linkcolor=NatureRed,
            urlcolor=NatureRed]{hyperref}
 
\definecolor{NatureTeal}{HTML}{0F6B68}
\definecolor{NatureBlue}{HTML}{1F4E79}
\definecolor{NatureRed}{HTML}{A23B3B}

\newcommand{\dataset}{\left\{\left(\alpha_i,y_i\right)\right\}_{i=1}^\mathcal{N}}

\newcommand{\bomega}{\bm{\omega}}

\newtheorem{theorem}{Theorem}
\newtheorem{lemma}{Lemma}
\newtheorem{proposition}{Proposition}
\newtheorem{definition}{Definition}
\newtheorem{corollary}{Corollary}

\usepackage{subcaption}
 
\makeatletter
\AtBeginDocument{%
  \long\def\@makecaption#1#2{%
    \par\vskip\abovecaptionskip
    \begingroup\small\rmfamily
      \begingroup\samepage\flushing
        \let\footnote\@footnotemark@gobble
        \@make@capt@title{#1}{#2}\par
      \endgroup
    \endgroup
    \vskip\belowcaptionskip}%
}
\makeatother

\begin{document}

\title{Learning to reconstruct Wigner functions in phase space }
\author{Xinyu Tang}
\thanks{These authors contributed equally to this work.}
\affiliation{John Hopcroft Center for Computer Science, Shanghai Jiao Tong University, Shanghai 200240, China}
\author{Yi-Hsin Lin}
\thanks{These authors contributed equally to this work.}
\affiliation{John Hopcroft Center for Computer Science, Shanghai Jiao Tong University, Shanghai 200240, China}
\author{Yan Zhu}
\email{vinzhu2@hku.hk}
\affiliation{QICI Quantum Information and Computation Initiative, Department of Computer Science,
The University of Hong Kong, Pokfulam Road, Hong Kong}
\author{Tailong Xiao}
\affiliation{State Key Laboratory of Photonics and Communications, Institute for Quantum Sensing and Information Processing, Shanghai Jiao Tong University, Shanghai 200240, China}%
 \affiliation{Hefei National Laboratory, Hefei 230088, China}
\affiliation{Shanghai Research Center for Quantum Sciences, Shanghai, 201315, P.R. China}
\author{Yuxuan Du}
\email{yuxuan.du@ntu.edu.sg}
\affiliation{College of Computing and Data Science, Nanyang Technological University, Singapore 639798, Singapore}
\affiliation{School of Physical and Mathematical Sciences, Nanyang Technological University, Singapore 639798, Singapore}
\author{Giulio Chiribella}
\email{giulio@hku.hk}
\affiliation{QICI Quantum Information and Computation Initiative, Department of Computer Science,
The University of Hong Kong, Pokfulam Road, Hong Kong}
\affiliation{Department of Computer Science, Parks Road, Oxford, OX1 3QD, United Kingdom}
\affiliation{Perimeter Institute for Theoretical Physics, Waterloo, Ontario N2L 2Y5, Canada}
\author{Qiongyi He}
\affiliation{State Key Laboratory for Mesoscopic Physics, School of Physics, Peking University, Beijing 100871, China}
\affiliation{Collaborative Innovation Center of Extreme Optics, Shanxi University, Taiyuan, Shanxi 030006, China}
\affiliation{Hefei National Laboratory, Hefei 230088, China}
\author{Ya-Dong Wu}
\email{wuyadong301@sjtu.edu.cn}
\affiliation{John Hopcroft Center for Computer Science, Shanghai Jiao Tong University, Shanghai 200240, China}

\begin{abstract}
Wigner function tomography is a central tool for characterizing continuous variable quantum systems subject to energy bounds that  guarantee a cutoff on their state space dimension. However,  it becomes increasingly demanding as the effective dimension  increases:  as experiments can only access the Wigner function at a finite number of points,  accurate characterizations require  large amounts of experimental data, rapidly becoming intractable as we try to explore larger portions of the system's state space.   
To address this bottleneck, here we  develop a machine learning framework for reconstructing Wigner functions directly from sparse phase-space data. For quantum states with sparse Fock-space or coherent-state representations, such as  binomial code states and cat states, we devise provably efficient machine learning algorithms whose measurement complexity scales only logarithmically with the effective Hilbert-space dimension. For more general states beyond this regime, such as Gottesman–Kitaev–Preskill (GKP) states, we design a deep learning model that reconstructs the Wigner function from sparse measurements and generalizes to arbitrary phase-space resolution. We demonstrate the broad applicability of our framework on both simulated data and experimental data from a circuit quantum electrodynamic system. On simulated data, we find that our model reconstructs  GKP states using substantially fewer measurements than required for informational completeness. On experimental data, we find that it reconstructs Wigner functions of GKP code states across multiple rounds of quantum error correction and identifies the dominant error process using significantly fewer measurements than conventional estimation techniques.
\end{abstract}

\maketitle

%\tableofcontents

\section{Introduction}
Continuous-variable (CV) quantum information encoded in bosonic harmonic oscillators has emerged as a promising platform for quantum computation. A leading implementation is provided by hybrid oscillator–qubit architectures, where bosonic modes are coupled to auxiliary qubits for control and measurement~\cite{4rf7-9tfx}. This approach has driven significant progress in bosonic quantum error correction~\cite{campagne2020quantum,ma2020error,sivak2023real,ni2023beating,brock2025quantum}. In these systems, CV quantum states are naturally characterized by quasi-probability distributions in phase space, most notably the \textit{Wigner function}, which provides a complete description equivalent to the density matrix~\cite{serafini2017}. This phase-space perspective has motivated extensive experimental efforts to measure Wigner functions and related quasi-probability distributions directly, enabling pointwise estimation of their values across phase space. These methods are now routinely performed in superconducting circuit quantum electrodynamics (QED)~\cite{vlastakis2013,eickbusch2022,sivak2023real,ni2023beating,brock2025quantum,cai2024protecting,t4cv-y398}, trapped-ion systems~\cite{PhysRevLett.125.043602,wsqr-j9f4}, and quantum acoustic systems~\cite{von2022parity}.

A fundamental challenge in the characterization of CV systems is the infinite dimensional nature of their state space. In practice, realistic characterization techniques rely on dimensional truncations, {\em e.g.} motivated by upper bounds on the average energy of the system. As the truncation dimension increases, however, the characterization becomes increasingly demanding. The core difficulty is that the quantity of interest is a continuous phase-space function, whereas experiments only access the Wigner function  through finitely many pointwise samples~\cite{PhysRevLett.78.2547}. To this end, existing experimental protocols~\cite{sivak2022,brock2025quantum,cai2024protecting} often rely on brute-force sampling over a dense phase-space grid. This strategy incurs a substantial computational bottleneck. For a truncated Hilbert space of dimension $d$, resolving the Wigner function requires at least $d^2$ distinct phase-space points, each of which must be measured and processed separately. As the photon number and effective dimension increase, this overhead rapidly becomes prohibitive, even for single-mode states. This scaling bottleneck raises a fundamental question: \textit{Can we characterize the entire Wigner function without a dense grid of  measurement data?}

\begin{figure*}
    \centering
    \includegraphics[width=0.9\columnwidth]{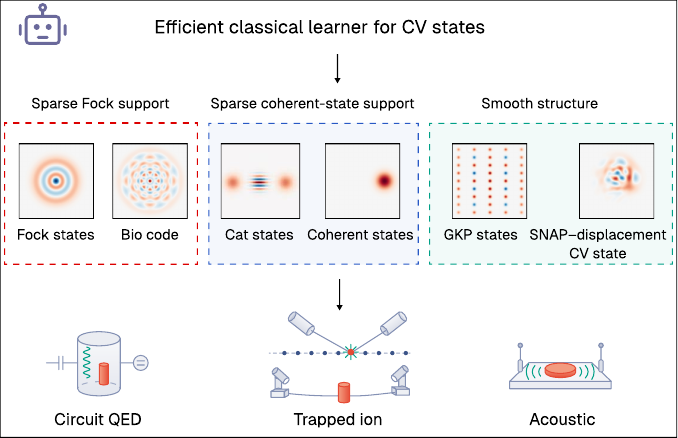}
    \caption{Illustration of various types of CV quantum states whose Wigner functions can be efficiently learned by our learning models. For each type of state, typical examples are provided. The Wigner function learning protocols can be widely performed on various physical platforms, including circuit-QED, trapped ions and acoustic systems.
    }
    \label{fig:discussion}
\end{figure*}

Artificial intelligence (AI) provides a data-driven route to representing and characterizing large-scale quantum systems~\cite{gebhart2023learning,du2025artificial}, with approaches ranging from statistical learning methods to neural-network ansatzes. Recent studies have established rigorous guarantees for learning CV quantum states in quantum state tomography~\cite{ohliger2011,gandhari2024,mele2024,zhao2025complexity,chen2026towards} and shadow tomography~\cite{becker2024,PhysRevResearch.6.033280}. However, these results have mainly been developed for optical settings, where the available data are typically quadrature or coherent-state samples obtained from homodyne or heterodyne measurements. They are therefore not directly applicable to characterizing an entire Wigner function in hybrid oscillator–qubit systems. A theory for learning phase-space representations directly from pointwise measurements remains largely unexplored. 

In parallel, deep neural network (DNN) approaches have shown promising empirical performance across CV quantum-state tasks, including tomography~\cite{tiunov2020,Ahmed2021PRL,Ahmed2021PRR}, entanglement detection~\cite{PhysRevLett.132.220202,gao2024classifying,gao2025foundation}, and similarity testing~\cite{wu2023}.
Yet existing neural-network approaches have largely been developed in optical measurement settings~\cite{tiunov2020,PhysRevLett.132.220202,gao2024classifying,gao2025foundation}, where data are obtained from homodyne or heterodyne measurements. They therefore do not directly address the hybrid oscillator–qubit systems considered here, where Wigner functions are characterized from sparse pointwise phase-space measurements. 
On the other hand, the neural-network models directly producing density matrices from phase-space measurements~\cite{Ahmed2021PRL,Ahmed2021PRR} are not readily extendable to scalable bosonic quantum systems. 
Taken together, these limitations raise the question of whether AI can enable efficient characterization of an entire Wigner function from sparse pointwise measurements.

%This work is motivated by this gap between existing theories and algorithms for learning CV states and the experimental progress in Wigner tomography.

Here we develop a learning-based framework that reconstructs Wigner functions from sparse pointwise measurements, addressing the central challenge posed above. We first study a wide class of CV states that obey a sparsity condition, namely states with sparse support in either the Fock basis or the coherent-state basis. As shown in Fig.~\ref{fig:discussion}, this class contains many experimentally relevant states, such as Fock states, binomial code states, coherent states, and cat states. In this regime, we cast Wigner-function reconstruction as a probably approximately correct (\textsf{PAC}) learning problem~\cite{aaronson2007} and develop an efficient machine learning model with a provable measurement-complexity guarantee. This model serves as a surrogate of a Wigner function enabling its precise prediction from only $\mathcal O(\log d)$ pointwise measurements, where $d$ is the effective Hilbert-space truncation dimension. These results extend quantum learning theory~\cite{anshu2024survey} to CV phase-space representations and provide a sample-efficient route to Wigner-function characterization.

Second, for quantum states beyond the sparse support regime, we develop a DNN model that represents the Wigner function as a neural surrogate, enabling evaluation at arbitrary phase-space coordinates. Inspired by high-resolution image reconstruction, the model takes sparse measurement data as input and learns an implicit representation of the underlying Wigner function, allowing reconstruction at arbitrarily high phase-space resolution.
This DNN model can be applied to reconstruct the Wigner functions of practical GKP states and randomly prepared states with certain low-depth quantum circuits, as shown in Fig.~\ref{fig:discussion}.
Using experimental data of GKP code states~\cite{sivak2023real}, we demonstrate that the proposed DNN model can accurately reconstruct Wigner functions across multiple rounds of quantum error correction~\cite{terhal2015quantum}, highlighting its robustness to realistic experimental noise. Furthermore, the reconstructed Wigner functions can be used to uncover the structure of the code space and the dominant error subspace in the quantum error-correction experiment. Our results show that the model successfully identifies the dominant quantum error process from sparse and noisy measurements, a task that conventional estimation approaches struggle to achieve.

This paper is structured as follows. Sec.~\ref{sec:introduction} provides a brief introduction to the Wigner functions of CV quantum states and their measurements. Sec.~\ref{sec:protocol} provides the formalism of Wigner function reconstruction as a learning problem. Sec.~\ref{sec:learningModels} presents two learning models: a provably efficient regression model for learning the Wigner functions of CV states with sparse Fock-state or coherent-state support, and a DNN model for learning the Wigner functions of CV states beyond the sparse regime. Sec.~\ref{sec:numericalResults} presents numerical results of these learning models on simulated data of quantum states, including binomial code states, cat states, GKP states, and randomly prepared CV states.
Sec.~\ref{sec:experiment} focuses on the application of this DNN model to experimental data. Sec.~\ref{sec:discussion} discusses a comparison between the two learning models. Sec.~\ref{sec:conclusion} concludes the paper by summarizing its key contributions and outlining future research directions.

\begin{figure*}
    \centering
    \includegraphics[width=0.80\textwidth]{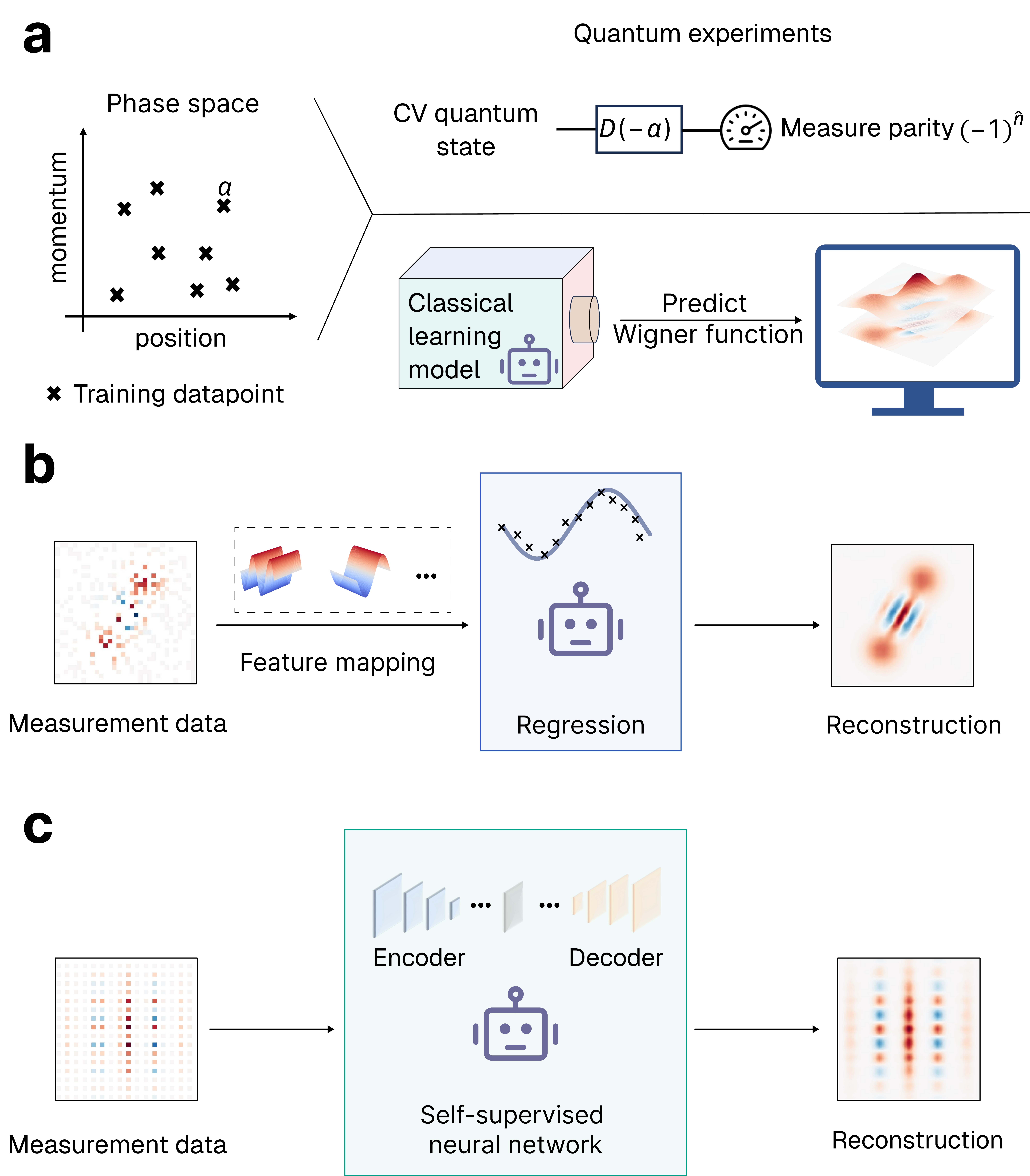}
    \caption{Schematic diagram of the learning model used to reconstruct Wigner functions of quantum states in phase space.
Subfigure~\textbf{a} illustrates the learning protocol. Using Wigner function measurement data collected at a sparse set of phase-space points (e.g., randomly sampled from a probability distribution over phase space or a predetermined grid), we train a model to learn the quantum state’s Wigner function. After training, the model can predict the Wigner function value at any arbitrary point in phase space. Subfigure~\textbf{b} presents a regression model that fits the measurement data by a linear combination of independent features. Subfigure~\textbf{c} presents a neural network with an encoder-decoder architecture for reconstructing a high-resolution Wigner function from its sparse measurement data. }
    \label{fig:task_diagram}
\end{figure*}

\section{Wigner functions and their measurements}
\label{sec:introduction}
Phase-space representations provide a powerful and intuitive framework for describing quantum states, particularly for CV systems~\cite{serafini2017}. Instead of representing a CV quantum state by a density operator acting on an infinite-dimensional Hilbert space $\mathcal H:=\text{span}_{n\in \mathbb N}\{\ket{n}\}$, phase-space methods map the density operator to a quasi-probability distribution over phase space $(x, p)$ with $x$ being position and $p$ being momentum, enabling a direct connection between quantum mechanics and classical statistical descriptions~\cite{curtright2013concise}.
A central object in phase-space quantum mechanics is the characteristic function, defined as the expectation value of the displacement operator 
\begin{equation}\label{eq:displacement}
D(\alpha):=\text{e}^{\alpha \hat{a}^\dagger-\alpha^* \hat{a}},
\end{equation}
where $\hat{a}$ and $\hat{a}^\dagger$ are the annihilation and creation operators, $\alpha=({x+\text{i}p})/{\sqrt{2}}$ is complex, and $\alpha^*$ is its complex conjugate. For a CV quantum state $\rho$ on $\mathcal H$, the characteristic function $C_\rho(\alpha)=\tr(\rho D(\alpha))$ contains full information about $\rho$ and serves as the Fourier transform of various phase-space quasi-probability distributions. 

Among the family of phase-space distributions, the Wigner function $W_\rho(\cdot)$ occupies a distinguished position. Define the Fourier transform of the characteristic function $C_\rho(\cdot)$ as
\begin{equation}\label{eq:WignerFunctionDef}
W_\rho(\alpha)=\frac{1}{\pi^2}\int_{\mathbb C} \text{d}^2\beta \,C_\rho(\beta) \text{e}^{\alpha\beta^*-\alpha^*\beta}, %{\color{red} Please check the above formula}, 
\end{equation}
The Wigner function offers a real-valued representation of a quantum state on phase space. 
Throughout this work, following the convention in Ref.~\cite{serafini2017}, we use the phase-space coordinates $\alpha$ and $(x,p)$ interchangeably satisfying the relations
\(
\alpha=(x+\text{i}p)/\sqrt{2}
\)
and $\text{d}^2\alpha=\frac{1}{2}\text{d} x\text{d} p$.
With this convention, the Wigner function satisfies the normalization
\[
\int_{\mathbb C}W_\rho(\alpha)\,\text{d}^2\alpha=1,
\]
or equivalently
\[
\int_{\mathbb R^2}W_\rho(x, p)\,\text{d}x \text{d}p=2.
\]
It preserves correct marginal distributions for position and momentum, i.e., $\int_\mathbb{R} \text{d} p W_\rho(x, p)=  2 \braket{x|\rho|x}$ and $\int_\mathbb{R} \text{d} x W_\rho(x, p)= 2\braket{p|\rho|p}$, where $x$ and $p$ denote position and momentum, respectively. Although it resembles a classical probability distribution, the Wigner function may take negative values, which are widely regarded as signatures of nonclassicality and known as Wigner negativity~\cite{kenfack2004negativity}. These features make the Wigner function a central tool for visualizing quantum interference effects and characterizing genuinely quantum resources. Note that Wigner functions can be defined for any linear operators on $\mathcal H$. An essential property of Wigner functions, which will be used by our learning models, is their linearity. That is, given any two linear operators $\mathcal A_1$ and $\mathcal A_2$, the Wigner function of their sum satisfies 
\begin{equation}\label{eq:linearity}
W_{\mathcal A_1+\mathcal A_2}(\cdot)=W_{\mathcal A_1}(\cdot)+W_{\mathcal A_2}(\cdot).
\end{equation}

The Wigner function can be directly measured in experiments~\cite{PhysRevLett.78.2547}. In particular, the Wigner function value at $\alpha$, i.e.,
\begin{equation}\label{eqn:wigner-fun}
W_\rho(\alpha) =\frac{2}{\pi} \tr\left(\rho D(\alpha)(-1)^{\hat{n}} D(-\alpha)\right),
\end{equation}
can be measured by performing a displacement operation $D(-\alpha)$ and then detecting the resulting photon parity $(-1)^{\hat{n}}$, where $\hat{n}=\hat{a}^\dagger \hat{a}$ is the photon number operator. Each shot of measurement yields a positive one for even parity and a negative one for odd parity. Averaging them over many shots asymptotically reaches the value of the Wigner function up to a constant $2/\pi$.

A direct measurement of a phase-space function can provide estimates only at a finite set of selected points. A straightforward strategy to reconstruct the Wigner function is to measure the function values on a sufficiently fine-grained grid within a supported region in phase space and use the resulting discretized values as the approximation to the continuous function. However, its cost is controlled by the number of phase-space points required to resolve the relevant structures of the Wigner function. For highly nonclassical states, especially states with large photon numbers, the required grid can become prohibitively dense~\cite{4rf7-9tfx} (see discussion in Sec.~\ref{app:complexity_of_conventional_approaches}).  

To recover its value at an arbitrary intermediate point, interpolation techniques, such as Lagrange interpolation, can be applied~\cite{PhysRevLett.120.090501}.  Local interpolation methods, such as bilinear interpolation, are typically more stable,  but they rely only on nearby measured values and therefore require dense sampling to recover fine-grained structures and high-frequency features of the original Wigner distribution. %These limitations motivate a different approach: instead of treating phase-space tomography as pointwise measurements followed by interpolation, we formulate it as a function-learning problem, where a classical learning model is trained to infer continuous Wigner-function values.

\section{An overview of learning models} 
\label{sec:protocol}
Here, we formulate the reconstruction of Wigner functions as a learning protocol. 
Given an unknown CV state $\rho$, this learning protocol aims to construct a surrogate $\widehat{W}_{\bm \eta}$ of its Wigner function $W_\rho$ in Eq.~(\ref{eq:WignerFunctionDef}) such that for any phase-space point $\alpha\in\mathbb C$, one can accurately predict the value of $W_\rho(\alpha)$ by $\widehat{W}_{\bm \eta}(\alpha)$.
The protocol consists of three main stages, which are data collection, model implementation and training, and model prediction. We briefly describe each stage below.

Data collection involves measuring the CV quantum states. To be concrete, one measures the displaced parity observable $D(\alpha_i)(-1)^{\hat n}D(-\alpha_i)$ in Eq.~(\ref{eqn:wigner-fun}) at a finite set of phase-space points $\{\alpha_i\}_{i=1}^\mathcal N$. The constructed dataset
\begin{equation}\label{dataset}
\mathcal{T}:=\dataset
\end{equation}
consists of statistical estimates $y_i$ of the Wigner function values $W(\alpha_i)$ at $\mathcal N$ distinct phase-space points. This dataset $\mathcal T$ serves either directly as the training set or as the raw data from which a training set is constructed.

Given the training dataset $\mathcal{T}$, the implementation of a learning model $\widehat{W}_{\bm \eta}$ depends on the learning paradigm under consideration. In a regression model, phase-space coordinates $\{\alpha_i\}$ are mapped into a higher-dimensional feature space, such that a linear combination of the basis functions in the feature space can approximate the target Wigner function. These basis functions must be properly tailored based on the sparsity structure of the quantum states.  In contrast, DNNs automatically learn suitable feature representations through multiple layers of linear transformations and nonlinear activation functions, with the neural architecture properly designed to exploit the phase-space pattern features of Wigner functions.

Regardless of the chosen paradigm, the learning model  $\widehat{W}_{\bm \eta}$ is parameterized by a set of trainable parameters $\bm \eta$, which are optimized to minimize the average prediction error over the training dataset $\mathcal{T}$, i.e.,
\begin{equation}\label{eq:trainErr}
\widehat{\mathsf{R}}_{\mathcal{T}}(\widehat{W}_{\bm \eta}) := \frac{1}{\mathcal N} \sum_{i=1}^{\mathcal N}\left|\widehat{W}_{\bm \eta}(\alpha_i) - y_i\right|^2.
\end{equation}
The optimized learning model can be used to predict the Wigner function at arbitrary phase-space points that were not measured experimentally, enabling the reconstruction of the entire phase-space distribution from a limited number of measurements.

A standard metric in statistical learning to quantify the performance is the expected risk, i.e.,
\begin{equation}\label{def:generalizationError}
\mathsf{R}(\widehat{W}) := \mathbb{E}_{\alpha} \left| \widehat{W}_{\bm \eta}(\alpha) - W(\alpha) \right|^2,
\end{equation}
where $\alpha$ is sampled from the same probability distribution that generates the phase-space points in the training dataset $\mathcal{T}$. The generalization ability of $\widehat{W}_{\bm \eta}$ is quantified by $|\mathsf{R}(\widehat{W})- \widehat{\mathsf{R}}_{\mathcal{T}}(\widehat{W}_{\bm \eta})|$.

\section{Implementation of learning models for Wigner functions}
\label{sec:learningModels}
Given the learning protocol outlined above, this section introduces learning models designed to address the problem of Wigner function learning. As shown in Fig.~\ref{fig:discussion}, these learning models are tailored to different classes of quantum states relevant to CV quantum information processing and quantum computing. For states with sparse support, we develop a sparse regression model with theoretical guarantees to fit the target Wigner function, as discussed in Sec.~\ref{sec:regression_model}. For states beyond the sparse regime, we introduce a DNN model to effectively represent Wigner functions with the smooth property, as indicated in Sec.~\ref{sec:nn_cv_states}. %These two types of learning models are both inspired by recent progress in AI for characterizing quantum systems~\cite{du2025artificial}.

\subsection{Learning Wigner function with sparse property}
\label{sec:regression_model}

Many physically relevant CV quantum states admit sparse representations in a suitable basis, which enables us to develop efficient learning algorithms to reconstruct their Wigner functions. For instance, in the Fock basis, Fock states and superpositions involving only a small number of Fock states are naturally sparse. Likewise, in the overcomplete coherent-state basis, coherent states and superpositions of a few distinct coherent states, such as cat states, can also be described by sparse expansions. Both types of sparsely supported states have been extensively studied and experimentally implemented with wide applications, including quantum error correction and quantum metrology~\cite{vlastakis2013,deng2024quantum,ni2023beating,PRXQuantum.4.030336}.

Motivated by their broad applicability, we design a provably efficient classical surrogate that reconstructs the Wigner functions of these sparsely supported CV states with guaranteed predictive accuracy. To this end, we develop a sparse regression framework for learning the Wigner function $W_\rho$ in Eq.~(\ref{eq:WignerFunctionDef}) of a single-mode CV quantum state from a limited number of measurements at different phase-space points. We focus on two practically important classes of states: those that are sparse in the Fock basis and those that are sparse in the coherent-state basis. For each class, we construct an explicit feature map and derive upper bounds on the required sample complexity $\mathcal{N}$ in Eq.~(\ref{dataset}), as well as the computational complexity to make the prediction.

The implementation of the regression model follows the learning protocol in Sec.~\ref{sec:protocol}. In particular, each $\alpha_i\in \mathcal{T}$ is sampled from a distribution $\mathbb{D}$, e.g., a uniform distribution over a specified phase-space region. Given the dataset $\mathcal T$, as illustrated in Fig.~\ref{fig:task_diagram}\textbf{b}, the next step is to construct a suitable feature map based on the sparsity assumptions on the target state. Mathematically, these states admit a sparse representation in a known basis, or more generally a frame, $\{\ket{\nu}\}$ of $\mathcal H$, i.e., $\rho=\sum_{\nu,\nu'}\eta_{\nu,\nu'}\ket{\nu}\bra{\nu'}$ with only few nonzero coefficients. 
Due to the linearity condition in Eq.~(\ref{eq:linearity}), $W_\rho$ admits a sparse linear combination of Wigner functions $W_{\ket{\nu}\bra{\nu'}}$, i.e., 
\begin{equation}\label{eq:WignerSurrogateLinearCombination}
W_\rho(\cdot)=\sum_{\nu,\nu'} \eta_{\nu,\nu'} W_{\ket{\nu}\bra{\nu'}}(\cdot).
\end{equation}

In this regard, we implement the surrogate $\widehat{W}_{\bm{\eta}}$ as a linear combination of elementary Wigner functions in Eq.~(\ref{eq:WignerSurrogateLinearCombination}), where the coefficients $\eta_{\nu,\nu'}$ in Eq.~(\ref{eq:WignerSurrogateLinearCombination}) are determined by fitting $\mathcal T$. Specifically, the feature map in $\widehat{W}_{\bm{\eta}}$ is specified by a set of linearly independent functions $\{\Phi_\omega(\cdot): \alpha \in \mathbb C \mapsto \Phi_\omega(\alpha) \in \mathbb{R}\}_{\omega=1}^{d_f}$, where each function $\Phi_\omega(\cdot)$ corresponds to an elementary Wigner function $W_{\ket{\nu}\bra{\nu'}}(\cdot)$.
The corresponding feature space is the linear span of these functions, with dimension $d_f$. The feature map $\bm{\Phi}$ is therefore defined as 
\begin{equation}
\alpha \mapsto \bm{\Phi}(\alpha) := \bigl[\Phi_1(\alpha), \Phi_2(\alpha), \dots, \Phi_{d_f}(\alpha)\bigr]^{\top} \in \mathbb{R}^{d_f}.
\end{equation}
Associated with $\bm{\Phi}$, the surrogate takes the form of 
\begin{equation}\label{def:WignerSurrogate}
\widehat{W}_{\bm{\eta}}(\alpha) := \langle \bm{\eta}, \bm{\Phi}(\alpha) \rangle,
\end{equation}
where $\bm{\eta} := (\eta_1, \eta_2, \dots, \eta_{d_f})$ denotes the trainable parameters of the regression model $\widehat{W}_{\bm \eta}$, and $\langle \cdot, \cdot \rangle$ denotes the inner product. 

The training of the surrogate in Eq.~(\ref{def:WignerSurrogate}) amounts to finding the optimal coefficients in Eq.~(\ref{eq:WignerSurrogateLinearCombination}) by using the dataset $\mathcal T$. Specifically, the coefficient vector $\bm \eta$ is obtained as the solution to a Lasso regression problem~\cite{hazan2012linearregressionlimitedobservation}, i.e.,
\begin{equation} 
    \min_{\bm{\eta}} \widehat{\mathsf{R}}_{\mathcal{T}}(\widehat{W}_{\bm \eta}) \quad \text{subject to} \quad \|\bm{\eta}\|_1 \leq t.\,
    \label{eq:Lasso}
\end{equation}
where $\| \bm \eta \|_1 = \sum_{i=1}^{d_f} |\eta_i|$.
The performance of the model is evaluated by $\mathsf{R}(\widehat{W})$ in Eq.~(\ref{def:generalizationError}),
where $\alpha\sim \mathbb D$ and we write $\widehat W$ instead of $\widehat W_{\bm\eta}$ when no confusion occurs. Our analysis below shows that the expected risk $\mathsf{R}(\widehat{W})$ depends on the underlying sparsity structure of the CV state, including sparse support in the Fock and coherent-state bases.

\subsubsection{States with sparse Fock support}
Let us begin to explore CV states that are sparsely supported in the Fock basis.
Specifically, we consider a truncated Hilbert space $\operatorname{span}\,\{\ket{n}\}_{n=0}^{d-1}$, while ignoring the truncation error from  $\mathcal H$.
A natural representation of any CV state $\rho$ on this truncated Hilbert space is given by its density matrix
\begin{equation}\label{eq:densityMatrix}
\rho = \sum_{n,m=0}^{d-1} \rho_{mn} |m\rangle \langle n|.
\end{equation}
We say that $\rho$, pure or mixed, is $s^2$-sparse in the Fock basis if its density matrix has at most $s^2$ nonzero entries, i.e.,  
\begin{equation}\label{sparse}
\sum_{n,m=0}^{d-1} \mathbbm{1}\{\rho_{mn} \neq 0\} \leq s^2.
\end{equation}

Due to linearity of Wigner functions, we have $W_\rho(\cdot)=\sum_{n,m=0}^{d-1} \rho_{mn} W_{\ket{m}\bra{n}}(\cdot)$, where $W_{\ket{m}\bra{n}}(\cdot)$ in Eq.~(\ref{eqn:wigner-fun}) refers to the Wigner function of Fock basis operator $|m\rangle \langle n|$. According to Eqs.~(\ref{eq:WignerSurrogateLinearCombination}) and (\ref{def:WignerSurrogate}), the feature map of the surrogate $\widehat{W}_{\bm{\eta}}(\alpha)$ takes the form as 
\begin{equation}
\mathbf \Phi = \left\{ \Phi_{m,n}(\alpha) =  W_{|m\rangle \langle n|} (\alpha) : 0\le n,m \le d-1 \right\}.
\end{equation}
The resulting feature space has dimension $d_f = d^2$.
The explicit form of the Wigner function of operator $\ket{m}\bra{n}$ with $n\ge m$ is 
\begin{equation}\label{WignerFockBasis}
W_{|m\rangle \langle n|} (\alpha) \!=\! \frac{2}{\pi} (-1)^m \!\sqrt{\frac{m!}{n!}} (2\alpha)^{n-m} L_m^{(n-m)}\bigl(4|\alpha|^2\bigr) e^{-2|\alpha|^2}, \nonumber
\end{equation}
where 
$L_m^{(n-m)}(x)$
is the generalized (associated) Laguerre polynomial of degree $m$~\cite{curtright2013concise,vanherstraeten2026extremenonnegativewignerfunctions}. As the degree of the polynomial $m$ rises, the function $L_m^{(n-m)}(x)$ becomes increasingly oscillatory along the real axis $x\in\mathbb R$. For $n < m$, we have $W_{|m\rangle \langle n|} = \overline{W_{|n\rangle \langle m|}}$, where $\overline{\phantom{xxx}}$ denotes complex conjugation.
With this feature map, the regression model in Eq.~\eqref{def:WignerSurrogate} becomes 
\begin{equation}\label{def:WignerSurrogate-fock}
\widehat{W}_{\bm \eta}(\alpha) = \sum_{n,m=0}^{d-1} \eta_{m,n} \Phi_{m,n}(\alpha),
\end{equation}
that is, a linear combination of $W_{\ket{m}\bra{n}}(\cdot)$.

The optimization of the surrogate $\widehat{W}_{\bm \eta}$ follows Eq.~(\ref{eq:Lasso}). The following theorem quantifies the prediction error, where the proof is deferred to Sec.~\ref{app:Fock_ML_proof}.

\begin{theorem}
    For any $\epsilon>0$ and $1/2>\delta >0$, suppose that the state under consideration is $s^2$-sparse as defined in Eq.~\eqref{sparse}, and every training data pair in $\mathcal T$ satisfies $|W(\alpha_i) - y_i| \leq \sqrt{\epsilon/4}$ with probability at least $1-\delta'$ where $\delta' = \delta/\mathcal{N}$. Let $\widehat{W}$ be the regression model obtained by solving Eq.~\eqref{eq:Lasso} with $\widehat{\textsf R}_{\mathcal{T}}(\widehat{W}) \leq \epsilon/2$. Then a training data of size
    \begin{equation}
    \mathcal{N}= O\left(\frac{s^4}{\epsilon^2} \log  \frac{d}{\delta}  \right)
    \end{equation}
    suffices to achieve $\mathsf{R}(\widehat{W}) \leq \epsilon$ for any distribution $\mathbb{D}$ with probability at least $1-2\delta$.
    \label{theorem:fock_ML}
\end{theorem}
This result demonstrates a sharp separation between dense-grid Wigner tomography and structure-aware learning. Conventional grid-based protocols require at least $d^2$ phase-space points to resolve a Wigner function in a $d$-dimensional truncated Hilbert space. By contrast, Theorem~\ref{theorem:fock_ML} indicates that for a target state that is $s^2$-sparse in the Fock basis, the required number of training measurements scales only as
$\mathcal N=\mathcal O(s^4\log d)$. Thus, when sparse CV states with $s=\mathcal O(1)$, there is an exponential improvement in the $d$-dependent scaling.  

Beyond the scaling improvement, this result extends the idea of classical learning surrogates~\cite{du2025artificial} to CV phase space. Prior studies have mainly focused on qubit-based systems, where surrogates are trained to predict observables of parameterized quantum circuits and the sample complexity is governed by circuit-dependent quantities, e.g., tunable gates or relevant Fourier modes~\cite{PhysRevLett.131.100803,du2025efficient,liao2025demonstration}. By contrast, the surrogate $\widehat{W}_{\bm \eta}$ developed here learns a continuous Wigner function, and its sample complexity is governed by the sparsity $s$ of the underlying CV state. This suggests that efficient quantum surrogates are not limited to qubit circuits, but can arise more broadly whenever the target representation has exploitable low-complexity structure, motivating extensions to other platforms such as fermionic systems.

\subsubsection{States with sparse coherent-state support}

We next turn to states that are sparse in the coherent-state basis. A pure CV state $|\psi\rangle$ is $s$-sparse in coherent-state support if it can be written as a superposition of $s$ coherent states
\begin{align}\label{sparse_coherent}
    |\psi\rangle = \sum_{i=1}^s a_i |\alpha_i\rangle,
\end{align} 
where the coefficients $\alpha_i \in \mathbb C$ ensure the normalization $\langle \psi|\psi\rangle=1$, and each $\ket{\alpha_i}$ is a coherent state 
\begin{equation}\label{eqn:coh-state}
|\alpha_i \rangle = D(\alpha_i)|0\rangle  = e^{-\frac{|\alpha_i|^2}{2}} \sum_{n = 0}^\infty \frac{\alpha_i^n}{\sqrt{n!}} |n\rangle.
\end{equation}
In addition, we consider a bounded region in phase space, with $|\alpha_j|\le R$ for all $ j \in [s]$, and require that the amplitudes of these coherent states be sufficiently distinct from each other so that
\begin{align}\label{sparse_coherent_separation}
    |\alpha_j - \alpha_\ell|\ge \Delta_\alpha, \quad \forall j,\ell \in [s] \quad \text{with} \quad j\neq \ell.
\end{align}

To identify a suitable feature map for learning states $|\psi\rangle \langle \psi|=\sum_{i,j=1}^s a_ia_j^* |\alpha_i\rangle \langle \alpha_j|$ in Eq.~\eqref{sparse_coherent}, we examine its Wigner function. For any coherent states $\ket{\alpha_i}$ and $\ket{\alpha_j}$, the Wigner function of the operator $|\alpha_i\rangle\langle\alpha_j|$ is
    \begin{align}
        W_{|\alpha_i\rangle \langle \alpha_j|} (\alpha) = \frac{2}{\pi}  \exp\left( 2i \Im \left[  u_1(\alpha_i, \alpha_j) \right]\right) \exp\left(u_2(\alpha, \alpha_i, \alpha_j) \right),
        \label{eq:coherent_operator_Wigner_main}
    \end{align}
where $u_1(\alpha_i, \alpha_j) = (\alpha_j^* - \alpha_i^*)\alpha  + \frac{1}{2}\alpha_i^*\alpha_j$, $u_2 = - 2\big|  \alpha - \frac{\alpha_i+\alpha_j}{2} \big|^2$, and $\Im[\cdot]$ denotes the imaginary part of a complex number. Because $ |\alpha_i\rangle\langle\alpha_j|$ is not Hermitian when $i\neq j$, its Wigner representation $W_{|\alpha_i\rangle\langle\alpha_j|}(\alpha)$ in Eq.~\eqref{eq:coherent_operator_Wigner_main}
is generally complex-valued. It can be viewed as a Gaussian envelope in phase space, centered at
$(\alpha_i+\alpha_j)/2$, multiplied by an oscillatory phase factor
$\exp{(2i\Im[(\alpha_j^*-\alpha_i^*)\alpha])}$. Thus, an off-diagonal coherent-state operator contributes an interference-like phase-space component rather than a real probability-like peak.

Since the set of coherent states forms an overcomplete basis and the amplitude $\alpha$ of a coherent state $\ket{\alpha}$ takes values in an uncountable set, we cannot represent an unknown target Wigner function as a finite or countable linear combination of the elementary Wigner functions in Eq.~(\ref{eq:coherent_operator_Wigner_main}) and directly optimize the corresponding coefficients as shown in Eq.~(\ref{def:WignerSurrogate}).
To this end, we introduce an alternative to align with Eq.~(\ref{def:WignerSurrogate}). According to the explicit form of $|\alpha_i\rangle\langle\alpha_j|$ in Eq.~(\ref{eq:coherent_operator_Wigner_main}), we define a countable set of functions for positive real numbers $a$ and $b$:
\begin{equation}\label{eq:GaborFrame}
\mathcal G (a,b)=\left\{ \Phi_{n,k}(r)= v_1(r,n,b) v_2(r,k,a) \big|n,k \in \mathbb Z^2\right\},
\end{equation} 
where $r=(x/\sqrt{2}, p/\sqrt{2})\in \mathbb R^2$ is a phase-space point, $|r|^2=\langle r,r\rangle$, $n=(n_1,n_2)$, $k = (k_1,k_2)$, $v_1(r,n,b)=\exp\left(2\pi i \langle bn, r \rangle\right)$, and $v_2(r,k,a) = \exp\left(-\pi |r-ak|^2\right)$. More specifically,
$\mathcal G(a,b)$ is an infinite set of two-dimensional Gaussian functions located on a lattice $a\mathbb Z\times a\mathbb Z$ (indexed by $k$ with spacing $a$ in phase space), each modulated by a phase with frequency multiples of $b$ (indexed by $n$).
When $ab<1$, $\mathcal G(a, b)$ is overcomplete in the space of integrable functions $L^2(\mathbb R^2)$ \cite{grochenig_foundations_2001,heil_basis_2011}, and hence, for any CV state $\rho$, its Wigner function $W_\rho\in L^2(\mathbb R^2)$ can be written as a linear combination of functions in $\mathcal G(a,b)$ with complex coefficients.

This set $\mathcal G(a,b)$ of functions is known as a Gabor frame~\cite{grochenig_foundations_2001}, where frames generalize the notion of bases in a Hilbert space by allowing overcompleteness. Gabor frames have important applications in time-frequency analysis (see Secs.~\ref{sec:frame_theory} and \ref{sec:Gabor_STFT} for details).
To express the Gabor frame more concisely, it is customary to define the translation operator \((T_{ak}f)(r) = f(r-ak)\) and the modulation operator \((M_{bn} f)(r) = \exp(2\pi i \langle bn, r\rangle) f(r)\) for any function $f\in L^2(\mathbb R^2)$. Then, using $g(r) = e^{-\pi |r|^2}$ to denote the two-dimensional Gaussian function, the expression in Eq.~(\ref{eq:GaborFrame}) is reformulated as $\mathcal G(a,b):=\left\{ (M_{bn}T_{ak} g) (\cdot) : n,k \in \mathbb Z^2\right\}$.

Since any physical state has bounded energy, instead of considering an infinite lattice in both the phase space indexed by $k$ and the frequency space indexed by $n$, we only need to consider a finite region indexed by both truncated values of $k$ and $n$. This enables us to consider a finite lattice in both the phase space and the frequency space, determined by parameters $N,K \in \mathbb N$ as
\begin{equation}
\Lambda_{N,K} := \left\{ (n,k) \in \mathbb Z^2 \times \mathbb Z^2: \|n\|_\infty \le N, \|k\|_\infty \le K \right\},
\end{equation}
where $\|x\|_\infty = \max_i |x_i|$ is the maximum element of the vector $x$. The resulting feature map for learning sparse coherent states is 
\begin{equation}
\mathbf \Phi = \left\{ \Phi_{n,k}(r) = (M_{bn}T_{ak} g)(r): (n,k) \in  \Lambda_{N,K}\right\},
\end{equation}
with feature dimension $d_f = (2N+1)^2(2K+1)^2$, and the corresponding regression model becomes
\begin{equation}\label{eq:regression_model_gabor}
\widehat{W}_{\bm \eta}(r) = \sum_{(n,k)\in \Lambda_{N,K}} \eta_{n,k} \Phi_{n,k}(r).
\end{equation}
In practice, to obtain a real-valued Wigner function,  by converting the modulation $e^{2\pi i \langle bn,r\rangle}$ into $\cos(2\pi \langle bn,r\rangle)$ and $\sin(2\pi \langle bn,r\rangle)$, we split each function $\Phi_{n,k}$ into two: $\Phi_{n,k,1}=\cos(2\pi \langle bn,r\rangle) T_{ak} g$ and $\Phi_{n,k,-1}=\sin(2\pi \langle bn,r\rangle)T_{ak}g$. Then the regression model in Eq.~(\ref{eq:regression_model_gabor}) becomes 
\begin{equation}
\widehat{W}_{\bm \eta}(r) = \sum_{(n,k)\in \Lambda_{N,K}} \eta_{n,k,\pm 1} \Phi_{n,k, \pm 1}(r).
\end{equation}

With this feature map, an analogous sample complexity bound is obtained for learning states with sparse coherent-state bases. The theorem below quantifies the learnability of the proposed model, whose proof is deferred to Sec.~\ref{app:Gabor_ML_proof}.
\begin{theorem}
Let the target CV state be $s$-sparse in the coherent-state basis given by Eq.~\eqref{sparse_coherent} and satisfy the condition Eq.~\eqref{sparse_coherent_separation}. Set $\kappa = 1 - (s-1)e^{-\frac 1 2 \Delta_\alpha^2}$ and let $\widehat{W}$ be the regression model obtained by solving Eq.~\eqref{eq:Lasso}. For any $\epsilon,\delta>0$, when all training examples in \(\mathcal{T}\) satisfy \(|W(\alpha_i) - y_i| \leq \sqrt{\epsilon/4}\), with probability at least $1-\delta'$ and $\delta' = \delta/\mathcal{N}$, a training data of size
    \begin{equation}
    \mathcal N =  \tilde{\mathcal O}\left( \frac{s^4}{\epsilon^2 \kappa^4} \log \frac{R}{\delta} \right)
    \end{equation}
    suffices to achieve $\mathsf R(\widehat{W}) \le \epsilon $ for any distribution $\mathbb D$ with probability at least $1-2\delta$.
    \label{theorem:Gabor_ML}
\end{theorem}
Here, the big $\tilde{\mathcal O}$ notation hides $\log\log$ factors in $s,1/\kappa$, and $1/\epsilon$, while the dependence on $R$ and $\delta$ is displayed explicitly. Under conditions $s= \mathcal O(1)$ and $1/\kappa = \mathcal O(1)$, the sample complexity $\mathcal{N}$ scales logarithmically with the phase-space region $R$. These conditions are routinely satisfied in the preparation of practically relevant CV states, including two- and four-component cat states~\cite{PhysRevLett.111.120501,mirrahimi2014dynamically}.

For a coherent state $|\alpha\rangle$ in Eq.~(\ref{eqn:coh-state}), the mean photon number is $|\alpha|^2$. 
In the Fock basis, the Hilbert-space dimension 
$d$ therefore scales linearly with $R^2$, where $R$ is the radius bounding the amplitude $|\alpha|$. Consequently,  a phase space region with $|\alpha|\le R$ corresponds to a Hilbert space truncation of order $d = \mathcal O(R^2)$. 
This result applies to any superpositions of a few distinct coherent states, exemplified by cat states, i.e., balanced superpositions of coherent states with opposite amplitudes.  %Numerical demonstrations are presented in Sec.~\ref{sec:RM_numerics}.

Together, Theorems~\ref{theorem:fock_ML} and~\ref{theorem:Gabor_ML} reveal a clear separation between dense-grid tomography and structure-aware Wigner-function surrogates. Conventional Wigner tomography requires at least $d^2$ phase-space points to characterize a state in a $d$-dimensional truncated Hilbert space, whereas the proposed surrogate $\hat{W}_{\bm{\eta}}$ in Eqs.~(\ref{def:WignerSurrogate-fock}) and (\ref{eq:regression_model_gabor}) can exploit the sparse support of the target CV state and reduce the sample complexity to $\mathcal O(s^4\log d)$ for sparse Fock-support states and $\mathcal O(s^4\kappa^{-4}\log R)$ for sparse coherent-support states. 

This separation persists for the total measurement cost. Statistical estimation of the sampled Wigner values introduces only logarithmic overheads, so the total number of measurements remains $\widetilde{\mathcal O}(s^4\log d)$ for sparse Fock-support states and $\widetilde{\mathcal O}(s^4\kappa^{-4}\log R)$ for sparse coherent-support states. As a result, unlike conventional tomography, the proposed surrogate $\hat{W}_{\bm{\eta}}$   provides a sample-efficient route to Wigner-function reconstruction.

Finally, we analyze the computational cost of the two proposed surrogates. The feature dimension scales as $d_f^{\rm Fock}=(d+1)^2=\mathcal O(d^2)$ for sparse Fock support and as $d_f^{\rm coh}=(2N+1)^2(2K+1)^2=\mathcal O(R^4)=\mathcal O(d^2)$ for sparse coherent-state support. Because both surrogates are regression models, the training runtime required to achieve prediction error $\epsilon$ is upper bounded by $\mathcal O(d^2\log d/\epsilon^2)$, while inference requires $\mathcal O(d^2)$ operations per prediction. Hence, both the training and inference costs scale at most quadratically with the effective truncation dimension $d$.

\subsection{Learning Wigner functions with smooth structure}
\label{sec:nn_cv_states}

\begin{figure*}
    \centering
    \includegraphics[width=0.95\linewidth]{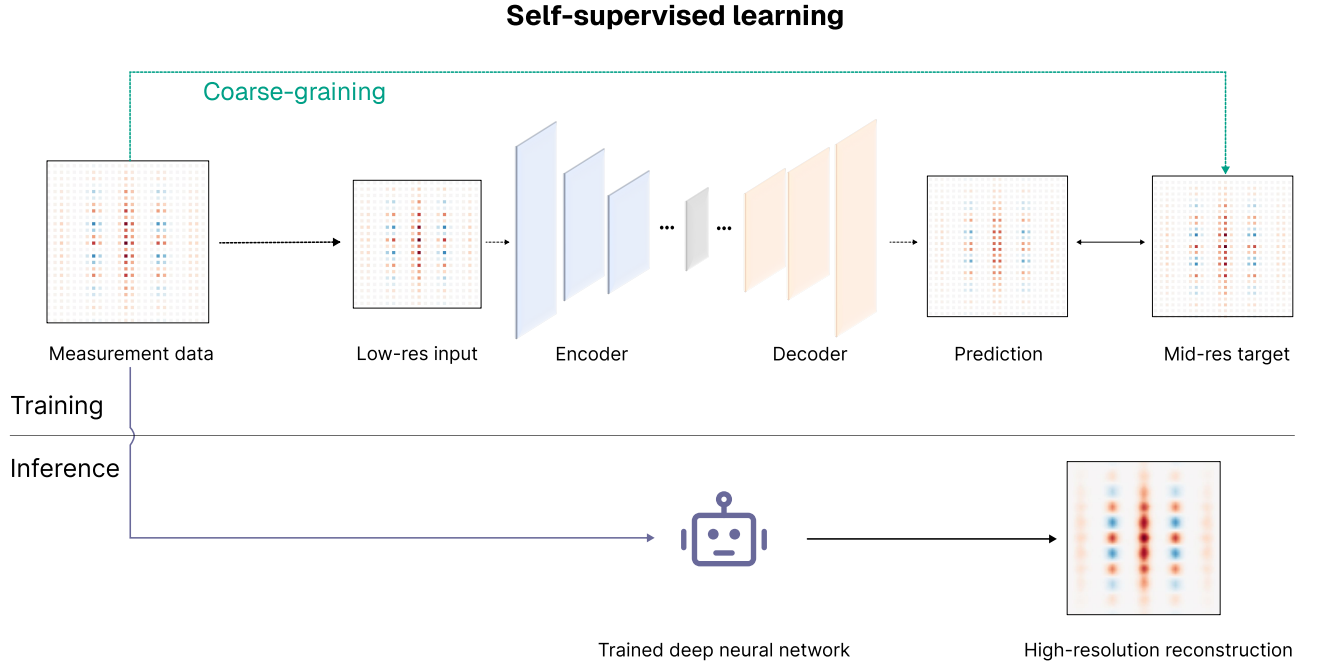}
    \caption{The learning scheme for our DNN model is inspired by high-resolution image reconstruction. During training, the model parameters are optimized in a self‑supervised manner to minimize estimation errors in the mid‑resolution Wigner function reconstructed from low‑resolution input. During inference, the trained model is applied to reconstruct high‑resolution Wigner functions.}
    \label{fig:diagram_nn}
\end{figure*}

%While efficient for states with sparse support, the regression model introduced above relies on prior knowledge of the target state. This requirement motivate the neural-network approach introduced next.
For many physically relevant CV states, the Wigner function exhibits structured patterns in phase space analogous to those encountered in natural images. In this regard, reconstructing a Wigner function from sparse pointwise measurements can be viewed as a super-resolution problem. Motivated by the success of neural implicit representations in computer vision~\cite{sitzmann2020implicit,chen2021learning,9156855}, we design a DNN to learn the continuous Wigner function directly from measurement data. Specifically, this approach encodes a continuous function as a neural network that maps phase-space coordinates directly to the corresponding function values. Rather than storing Wigner function values on a discrete grid, the continuous pattern of Wigner function is implicitly encoded into the weights of this DNN model, enabling the evaluation of Wigner function values at arbitrary phase-space points. 

More generally, neural networks have been widely applied for efficient implicit reconstruction of quantum systems~\cite{torlai2018,carrasquilla2019,wang2022}. Closely related to our setting, a neural network framework is developed to learn representations of quantum states from partial measurement data that enables predictions of outcome statistics for those measurements not performed yet~\cite{zhu2022}.  
Here, our DNN model learns an implicit representation of the Wigner function from sparse measurement data and enables the prediction of Wigner function values at arbitrary phase-space points.

We now describe the implementation and training of this DNN model. We consider a dataset $\mathcal{T}$, where each $\alpha_i$ lies on a sparse $m \times m$ phase-space grid $\Gamma$ covering $\Omega$, with each square cell of size $h \times h$. 
Rather than directly using the raw dataset $\mathcal{T}$ as training data, we preprocess $\mathcal{T}$ to obtain the training dataset for our DNN model. 
% In this preprocessing, we generate pairs of low- and medium-resolution datasets from the original $m \times m$ measurement grid, achieved by downsampling the measurement data using spline interpolation. 

In this preprocessing step, we construct paired low- and medium-resolution datasets from the original $m\times m$ measurement grid by interpolation-based resampling. Let $\Gamma_{\rm low}$ and $\Gamma_{\rm mid}$ denote the resulting coarse and finer phase-space grids, respectively. Denote the two resampled datasets as $\mathcal T_{\rm low}=\{(\alpha_i, y_i):\alpha_i\in\Gamma_{\rm low}\}$ and $\mathcal T_{\rm mid}=\{(\alpha_i,y_i):\alpha_i\in\Gamma_{\rm mid}\}$. As shown in Figs.~\ref{fig:task_diagram}\textbf{c} and \ref{fig:diagram_nn}, $\mathcal T_{\rm low}$ provides the low-resolution input and $\mathcal T_{\rm mid}$ provides the corresponding medium-resolution target for training. For ease of notation, let $W_{\rm low}\in
\mathbb R^{m_{\rm low}\times m_{\rm low}}$
denote the two-dimensional array of Wigner function values corresponding to
$\mathcal T_{\rm low}$ on the grid
$\Gamma_{\rm low}$.

Different from machine learning models in Sec.~\ref{sec:regression_model} that require a manually selected feature map from prior physical knowledge, the proposed DNN $\widehat W_{\bm\eta}$ implicitly learns the feature representation from data. The model adopts an encoder--decoder architecture, as illustrated in Fig.~\ref{fig:task_diagram}\textbf{c}. In particular, the encoder $E_{\bm\theta}$, implemented as a convolutional neural network with residual connections, maps the input array $W_{\rm low}$ to a latent feature representation
$E_{\bm\theta}(W_{\rm low})$,
where $\bm\theta$ denotes the encoder parameters. The latent feature representation $E_{\bm\theta}(W_{\rm low})$ consists of learned feature vectors distributed over the phase-space grid and captures both local and global structures of the underlying Wigner function. Unlike the regression models introduced in Eq.~(\ref{def:WignerSurrogate}), where the feature map is fixed in advance, the representation $E_{\bm\theta}(W_{\rm low})$ is learned automatically from the training data and optimized jointly with the prediction model.

Given an arbitrary phase-space point $\alpha\in\Omega$, the model first extracts a local feature vector by interpolating the latent representation,
\begin{equation}
    \bm f_{\bm\theta}(\alpha, W_{\rm low})
    =
    \mathcal I(E_{\bm\theta}(W_{\rm low}),\alpha),
\end{equation}
where $\mathcal I$ denotes bilinear interpolation on the latent feature map. For a query point $\alpha$, the feature vector $\bm f_{\bm\theta}(\alpha,W_{\rm low})$ is computed from the four nearest latent feature vectors surrounding $\alpha$, weighted according to their relative distances from the query point. Since the Wigner function is inherently a continuous phase-space function, this interpolation step converts the discrete latent representation generated by the encoder into a continuous feature field from which information can be queried at arbitrary phase-space coordinates. The resulting local feature vector captures the neighborhood structure around the query point and provides the decoder with spatially localized information. Following the local implicit image function framework of Refs.~\cite{chen2021learning, 9156855}, this design enables resolution-independent reconstruction and allows a single trained model to predict Wigner function values on grids finer than those encountered during training.
% {\color{red} [add intuition, add citation, explain $\mathcal I$]}  

The decoder $D_{\bm\phi}$, implemented as a multilayer perceptron, then predicts the corresponding Wigner function value according to
\begin{equation}
    \widehat W_{\bm\eta}(\alpha,W_{\rm low})
    =
    D_{\bm\phi}
    \bigl(
        \bm f_{\bm\theta}(\alpha,W_{\rm low}),
        \alpha,
        \bm c
    \bigr),
\end{equation}
where $\bm c$ denotes the normalized cell size of the target grid and $\bm\eta=(\bm\theta,\bm\phi)$. In this way, the neural network defines a surrogate of the continuous Wigner function over phase space. Since the phase-space coordinate $\alpha$ enters explicitly as an input, the trained model can evaluate $\widehat W_{\bm\eta}(\alpha)$ at arbitrary phase-space points and thereby reconstruct the Wigner function at arbitrary resolutions.

The DNN model is trained in a self-supervised manner, as illustrated in Fig.~\ref{fig:diagram_nn}, implying that the supervision labels are generated from the measurement raw data $\mathcal T$ rather than from externally labeled high-resolution ground truth.
The goal of training is to enable the neural network to recover a medium-resolution Wigner function from a low-resolution input. To this end, we construct multiple paired low- and medium-resolution datasets $\mathcal{T}_{\text{low}}$ and $\mathcal{T}_{\text{mid}}$ from the original dataset $\mathcal{T}$. During training, a paired dataset is sampled, and the DNN $\widehat W_{\boldsymbol{\eta}}$ is optimized by minimizing the mean squared prediction error~(Eq.~(\ref{eq:trainErr})) on the corresponding medium-resolution grid, i.e.,
\begin{equation}
\widehat{\mathsf R}_{\mathcal T_{\rm mid}}(\widehat{W}_{\eta})
=
\frac{1}{|\Gamma_{\rm mid}|}
\sum_{\alpha_i\in\Gamma_{\rm mid}}
\left|
\widehat W_{\boldsymbol{\eta}}(\alpha_i,W_{\rm low})-y_i
\right|^2,
\end{equation}
where $|\Gamma_{\rm mid}|$ is the number of phase-space points in the sampled medium-resolution grid. 

Once trained, the model $\widehat W_{\boldsymbol{\eta}}$ can take measurement data sampled at only \(m \times m\) points—potentially far fewer than required for information-complete tomography—and produce a continuous approximation of the Wigner function that can be evaluated at phase-space points on a finer grid with arbitrarily high resolution.
%This test assesses the model's generalization performance on out-of-distribution data, since the network has not been exposed to data at such high resolutions during training. 

We remark that the neural implicit representation of CV states does not rely on any prior assumptions about the underlying quantum states. This generality provides two key advantages. First, it enables the model to learn a wide variety of physical CV states, including those states beyond sparse supports, as long as the state exhibits a certain learnable structure in its phase-space function. Second, this approach enhances robustness to practical noise. When measurement data contain statistical error and systematic noise, the DNN model can still reconstruct the experimental Wigner function faithfully.

\section{Numerical simulations}
\label{sec:numericalResults}
The numerical study considers several representative bosonic error-correcting code states, including binomial code states, cat states, and GKP states. For these states, the two learning models introduced in Sec.~\ref{sec:learningModels} are applied to Wigner-function reconstruction from sparse phase-space measurements. The simulations examine how the required number of measurement points scales with the effective dimension of the truncated Hilbert space.

In the following numerical experiments, we use the phase-space coordinates  $(x,p)=(\sqrt{2}\Re(\alpha), \sqrt{2}\Im{\alpha})$  instead of $\alpha$ for convenience. 
The reconstruction quality is quantified by the overlap between the reconstructed Wigner function and the ground truth, i.e., 
\begin{equation}
\mathcal{F} := \frac{\pi}{2} \int_{\mathbb{R}^2} \text{d}x \text{d}p \, \widehat{W}(x, p) W(x, p),
\end{equation}
which equals the quantum fidelity for a target pure state~\cite{serafini2017}.  Without loss of generality, this integral is approximated by a sum over a finite grid $\Gamma$ of points,
\begin{equation}\label{eq:overlap}
\mathcal{F} \approx \frac{\pi}{2} \Delta\sum_{(x,p) \in \Gamma} \widehat{W}(x,p) W(x,p),
\end{equation}
where $\Delta$ is the area of each cell in the grid. 
As $\mathcal{F}$ approaches unity, the reconstructed state converges toward the ground truth.

\subsection{Numerical results with machine learning models}\label{sec:RM_numerics}
The regression models introduced in Sec.~\ref{sec:regression_model} are first applied to simulated CV states with sparse supports. The two representative examples are binomial code states with sparse support in the Fock basis and cat states with sparse support in the coherent-state basis. Both classes are highly non-Gaussian and play central roles in bosonic quantum error correction~\cite{michael2016,PhysRevX.10.011058,ni2023beating}.

\subsubsection{Learning binomial code states}
 The binomial code states~\cite{michael2016} are superpositions of Fock states with amplitudes involving binomial coefficients. Mathematically, they take the form as
\begin{equation}\label{binomialCodeState}
\ket{\bar{0}/\bar{1}}_{\text{b}} := \frac{1}{\sqrt{2^N}} \sum_{p = \text{even/odd}}^{[0, N+1]} \sqrt{\binom{N+1}{p}} \ket{p(S+1)},
\end{equation}
where $\bar{0}$ and $\bar{1}$ denotes the encoded logical qubit, and $(N, S) \in \mathbb{N}^+ \times \mathbb{N}$ specify different classes of binomial code states.
The states in Eq.~(\ref{binomialCodeState}) have support only on Fock states $\ket{p(S+1)}$, where $p\in [0, N+1]$ is restricted to even values for $\ket{\bar{0}}_{\text{b}}$ and odd values for $\ket{\bar{1}}_{\text{b}}$.
For example, when $N=S=1$, Eq.~(\ref{binomialCodeState}) yields the encoding $\ket{\bar{0}}=\frac{1}{\sqrt{2}}(\ket{0}+\ket{4})$ and $\ket{\bar{1}}=\ket{2}$.

Specifically, we consider the reconstruction of the state $\ket{\bar{1}}_{\text{b}}$ for various combinations of parameters $(N, S)$.
Here we choose the training data $\mathcal{T}$ as the true Wigner function values on a set of phase-space points randomly chosen from a $147 \times 147$ grid. By applying the regression model in Eq.~(\ref{def:WignerSurrogate-fock}), we can predict the Wigner function values over all other phase space points on the grid. 

\begin{figure*}               
      \centering
      \includegraphics[width=1.0\linewidth]{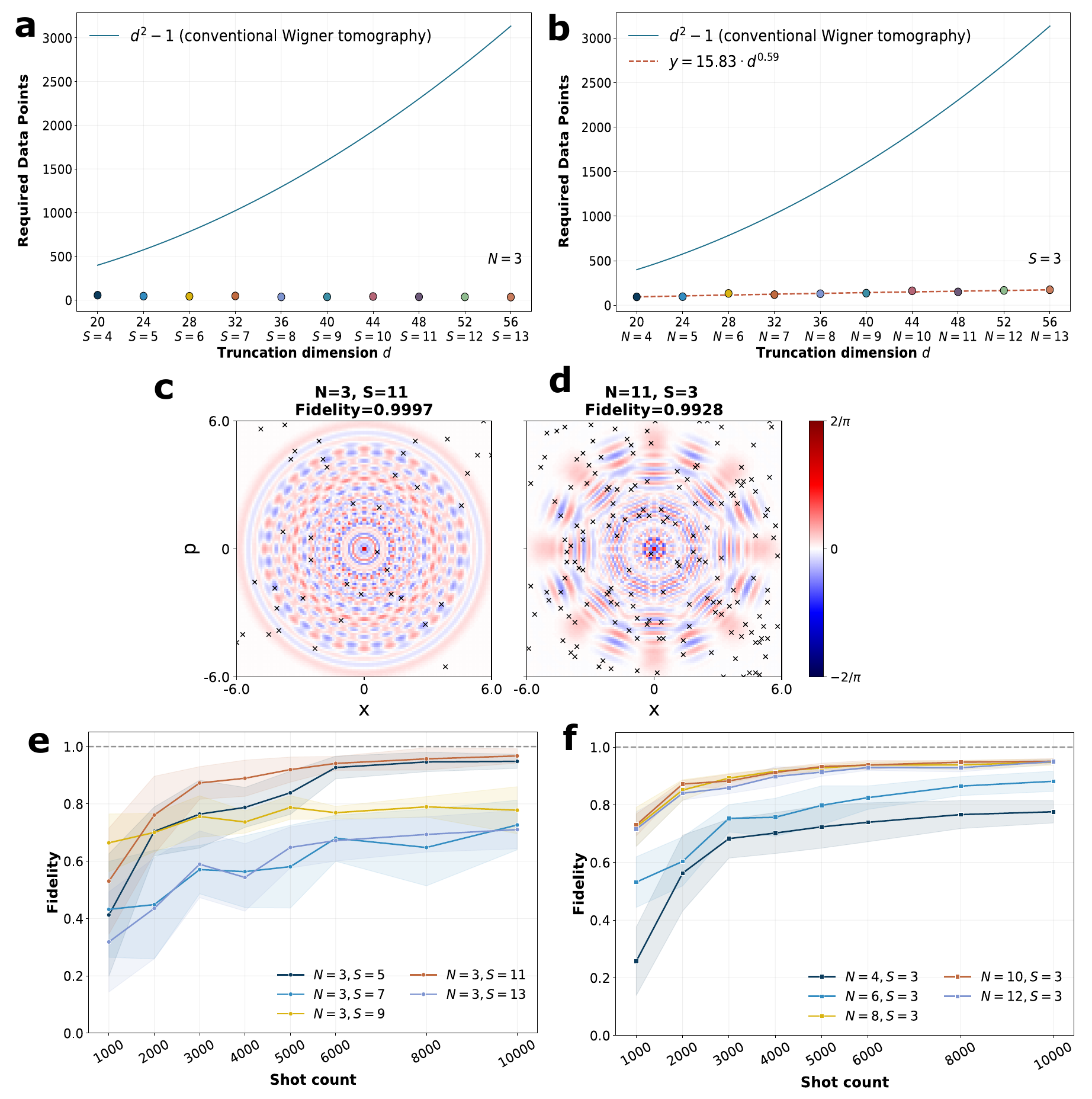}
     \caption{Reconstruction of Wigner functions for simulated binomial code states.
In Subfig.~\textbf{a}, the number of phase-space points required to achieve a fidelity of $0.99$ is shown with respect to the truncation dimension for binomial code states with fixed $N=3$ and increasing values of $S$.
In Subfig.~\textbf{b}, the number of phase-space points required to achieve a fidelity of $0.99$ is shown with respect to the truncation dimension for binomial code states with fixed $S=3$ and increasing values of $N$. In both Subfigs.~\textbf{a} and \textbf{b}, the scaling required for information completeness in conventional Wigner tomography ($d^2-1$) is shown as a reference. Each point represents the mean over five independent random realizations of the phase-space sampling pattern, and the error bars indicate one standard deviation. The error bars may appear visually small because the vertical axis spans a relatively large range.
In Subfig.~\textbf{c}, an example reconstruction of the Wigner function is shown for a binomial code state with $N=3$ and $S=11$. Crosses indicate the 41 randomly selected training points on a $128\times128$ grid. The fidelity between the reconstruction and the ideal Wigner function is 0.9997.
In Subfig.~\textbf{d}, the same analysis as in Subfig.~\textbf{c} is shown for $N=11$ and $S=3$, using 166 training points. The fidelity is 0.9928.
In Subfig.~\textbf{e}, the fidelity of the reconstructed Wigner function is shown as a function of the number of measurements per phase-space point for binomial code states with fixed $N=3$ and increasing values of $S$.
In Subfig.~\textbf{f}, the fidelity of the reconstructed Wigner function is shown as a function of the number of measurements per phase-space point for binomial code states with fixed $S=3$ and increasing values of $N$. In both Subfigs.~\textbf{e} and \textbf{f}, each point represents the mean fidelity over 10 independent experiments, and the shaded region indicates one standard deviation.}
  % {\color{red} [Add explanations on shaded areas]}   }                                                                                                                
      \label{fig:binomial}                                                                                                                                                      
  \end{figure*}

Figures~\ref{fig:binomial}\textbf{a} and~\ref{fig:binomial}\textbf{b} plot the minimum number of training data points required by this regression model to reconstruct the Wigner function with fidelity exceeding $0.99$. The truncation dimension here refers to the maximum photon number of the binomial code states considered, which ranges from $20$ to $56$.

In Fig.~\ref{fig:binomial}\textbf{a}, we fix $N$ while increasing the value of $S$ in Eq.~\eqref{binomialCodeState}. In this case, although the truncation dimension $d$ grows, the number of Fock basis supports $\lceil (N+1)/2 \rceil$ remains constant. The results show no obvious change in the number of phase-space points required as $d$ increases. This observation aligns with Theorem~\ref{theorem:fock_ML}. That is, when the CV states are supported on a constant number of Fock basis elements, the required number of phase-space points scales at most as $\mathcal O(\log d)$.

In contrast, in Fig.~\ref{fig:binomial}\textbf{b}, we fix $S$ and increase the value of $N$ in Eq.~\eqref{binomialCodeState}. Here, the number of Fock basis supports increases with $d$. 
The number of required Fock basis support increases with $d$ as $s^2 = C d^2$ for some constant fraction $C < 1$. Nevertheless, the sample complexity remains substantially lower than $d^2$, highlighting the potential of this regression model for efficient reconstruction of binomial code states.

In Figs.~\ref{fig:binomial}\textbf{c} and~\ref{fig:binomial}\textbf{d}, we present examples of reconstructed Wigner functions of two binomial code states using our regression model, together with their associated training data points, indicated by crosses.

The numerical results above are obtained using shot-noise-free training data $\mathcal T$. To assess robustness to finite sampling, Figs.~\ref{fig:binomial}\textbf{e} and~\ref{fig:binomial}\textbf{f} plot the reconstruction fidelity as a function of the number of measurement shots per phase-space point, using the settings of Figs.~\ref{fig:binomial}\textbf{a} and~\ref{fig:binomial}\textbf{b}, respectively. The results show that the reconstruction fidelity saturates once the shot count exceeds a modest threshold, indicating robustness to finite-shot statistical noise.

%Compared with DNN models, as discussed later, the regression model exhibits less robustness to statistical errors. 

\subsubsection{Learning cat states}

The canonical examples of states that are sparse in the coherent state basis are the Schrödinger cat states. These are defined by
\begin{equation}\label{catState}
\ket{\mathcal C_\alpha^\pm}
= \mathcal N_\alpha^\pm
\left(\ket{\alpha}\pm\ket{-\alpha}\right),
\end{equation}
where $\mathcal N_\alpha^\pm=[2(1\pm e^{-2|\alpha|^2})]^{-1/2}$ is the normalization factor. In the following numerical experiments, we focus on even cat states $\ket{\mathcal C_\alpha^+}$ with different amplitudes $\alpha$. Since each cat state is supported on two coherent states centered at $\alpha$ and $-\alpha$, it provides a natural test case for the regression model designed for CV states sparse in coherent-state support. We choose the training data as the true Wigner-function values on a set of phase-space points randomly selected within a circular region covered by a $142\times142$ grid. For each state, the sampled phase-space region and the parameters of the Gabor feature map are chosen according to the coherent amplitude $\alpha$ to cover the region where the Wigner function has non-negligible values. The regression model is then trained with Gabor features and used to predict the Wigner function values on the full grid. 

Figure~\ref{fig:cat}\textbf{a} plots the minimum number of randomly sampled phase-space points required to reconstruct the Wigner function with fidelity exceeding $0.99$ for cat states of varying amplitudes. We distinguish between the total number of randomly sampled phase-space points and the number of effective non-zero training points, namely, the sampled points located in the phase-space region where the Wigner function has non-negligible values. 
The latter provides a more accurate measure of the number of data points needed with varying truncation dimension, since the area of the nontrivial region in phase space remains constant regardless of the amplitude of $\alpha$. 
As the truncation dimension $d$ increases, the two quantities are empirically found to be $\mathcal O(d^{1.4})$ and $\mathcal O(d^{0.6})$, respectively.
Although the numerical scaling does not exactly follow the logarithmic scaling provided in Theorem~\ref{theorem:Gabor_ML}, our numerical scaling still grows much more slowly than the $d^2-1$ scaling required for informationally complete Wigner tomography.

In Fig.~\ref{fig:cat}\textbf{b}, we present representative reconstructions for four rotated cat states marked in Fig.~\ref{fig:cat}\textbf{a}. For each state, the reconstructed Wigner function is compared to the corresponding ground truth. The regression model accurately captures both the two coherent state peaks and the interference fringes between them, achieving fidelities above $0.99$ using only a small subset of the full phase-space grid.

Additionally, we apply the two regression models to states beyond the regime of sparse support, including an experimental GKP state shown in Fig.~\ref{fig:gkp_samples} and an experimental state prepared by applying SNAP gate to a coherent state, shown in Fig.~\ref{fig:experimental_data} of the Appendix. Unsurprisingly, the reconstruction quality is limited. Specific details regarding the numerical implementations of the regression model and other results on simulated states can be found in Appendix Subsec.~\ref{app:regression_numexp}.

\begin{figure*}               
      \centering
      \includegraphics[width=0.8\linewidth]{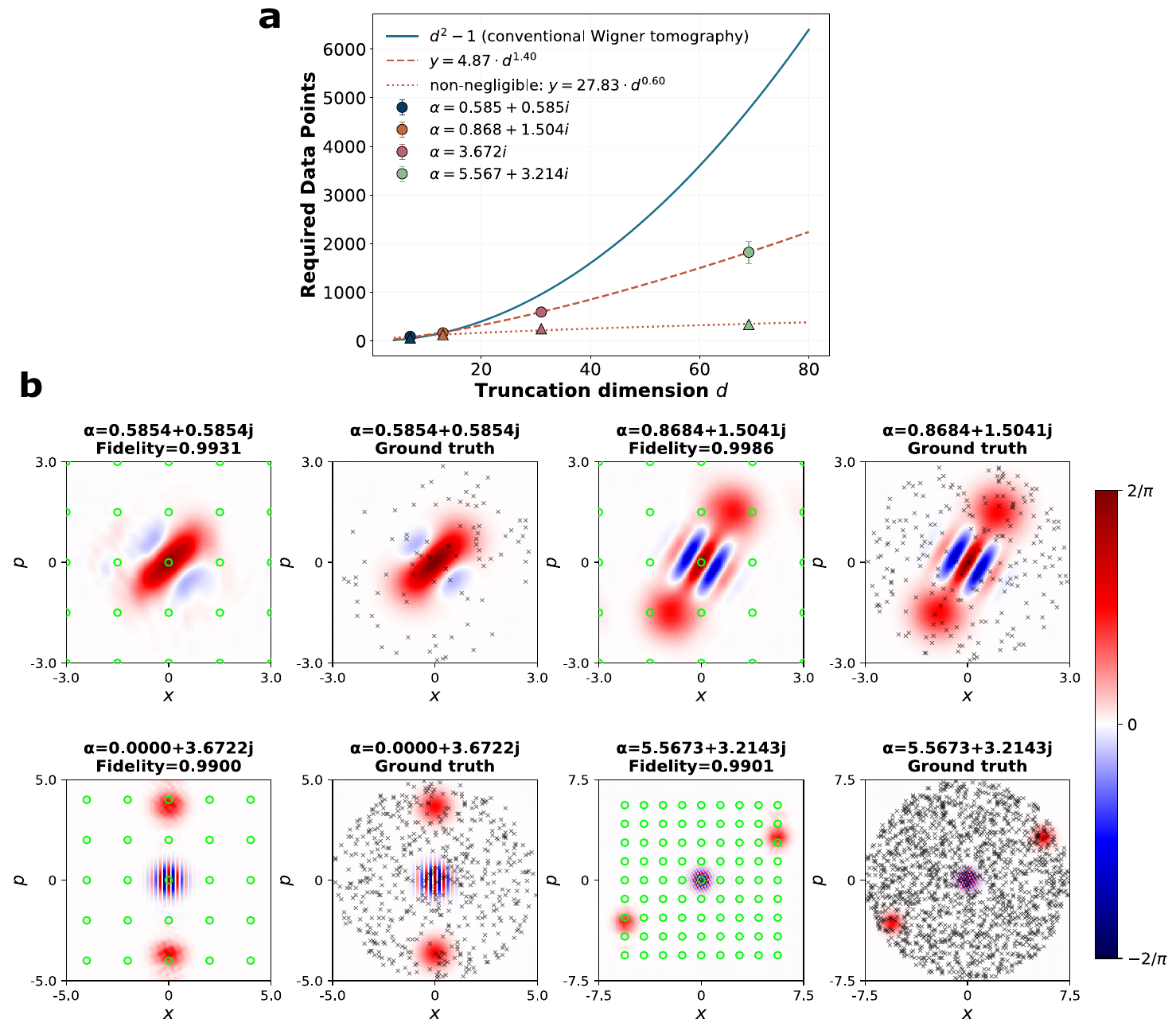}
      \caption{Reconstruction of Wigner functions for simulated cat states.
In Subfig.~\textbf{a}, the number of phase-space points required to achieve a fidelity $>0.99$ with respect to the truncation dimension $d$ is shown for cat states with varying complex amplitudes $\alpha$. The scaling required for information completeness in conventional Wigner tomography ($d^2-1$) is shown as a reference. The dashed and dotted lines indicate the empirical scaling of the total and non-zero required points, respectively. Each point represents the mean over five independent random realizations of the phase-space sampling pattern, and the error bars indicate one standard deviation.
In Subfig.~\textbf{b}, example reconstructions of the Wigner functions for the four corresponding cat states marked in Subfig.~\textbf{a} are shown. For each state, the reconstructed Wigner function (left) is compared with the ideal ground truth (right). The green circles on the left panels correspond to the centers $ak$ of the Gaussians in the Gabor frame~(Eq.~(\ref{eq:GaborFrame})). The fidelities between the reconstructions and the ideal Wigner functions are $0.9931$, $0.9986$, $0.9900$, and $0.9901$, respectively.}                                                                                                    
      \label{fig:cat}                                                                                   
  \end{figure*}

\subsection{Numerical results with DNN models}
%Rather than evaluating the performance of the DNN on states that are sparse in the Fock or coherent-state bases, we benchmark its performance on states that lie beyond these sparse regimes.
Motivated by the prominent role of GKP states in quantum error correction, we first investigate the performance of the proposed DNN model when applied to reconstruct their Wigner functions. After that, we consider a broader family of non-Gaussian states generated by sequences of displacement operators and selective-number-dependent arbitrary phase (SNAP) gates. This gate set has become a standard toolbox for preparing a wide variety of important single-mode non-Gaussian states in circuit QED~\cite{kudra2022}. Consequently, the output states of such circuits provide a natural testbed for evaluating the performance of our DNN model.

\subsubsection{Learning to reconstruct GKP states}

 An ideal GKP state is a grid state in phase space composed of delta functions~\cite{gottesman2001}. Since ideal GKP states are not normalizable, we must consider practical GKP states consisting of superpositions of  squeezed Gaussian peaks with a Gaussian envelope,
\begin{align}
    &\ket{\bar{0}/\bar{1}(\sigma)}_{\text{GKP}}=\int_{-\infty}^{\infty} \text{d} x \psi_{0/1}(x;\sigma) \ket{x} ,\\ \label{eq:theoreticalGKP}
    &\psi_k(x;\sigma)=\sum_{s\in \mathbb N} \exp(-2\pi \sigma^2 s^2)\exp\left[ -\frac{(x-\sqrt{\pi}(2s+k))^2}{2\sigma^2}\right],
\end{align}
where $k=0$ or $1$, $\ket{x}$ is a position eigenstate with eigenvalue $x\in\mathbb R$, $\sigma^2$ is the variance of each individual peak, and $1/\sigma^2$ characterizes the variance of the global envelope. 
For our convenience, we omit the squeezing parameter $\sigma$ in the expression of GKP code states.
The Gaussian peaks of $\ket{\bar{0}}_{\text{GKP}}$ and $\ket{\bar{1}}_{\text{GKP}}$ have the same spacing of $2\sqrt{\pi}$ but is offset from each other by $\sqrt{\pi}$.
Unlike the states considered in Subsec.~\ref{sec:regression_model}, GKP states are not sparse in either the Fock basis or the coherent-state basis, placing them beyond the sparse support regime.

In experiments, GKP states are usually superpositions of a finite number of Gaussian peaks. 
Specifically, the wavefunctions of realistic GKP states can be written as
\begin{align}\notag
&\psi_k(x;\sigma,s_{\text{max}})\approx\\
&\sum_{s=-s_{\text{max}}}^{s_{\text{max}}} 
\exp(-2\pi \sigma^2 s^2)
\exp\left[-\frac{(x-\sqrt{\pi}(2s+k))^2}{2\sigma^2}\right], \label{state:GKP} 
\end{align}
where $2s_{\text{max}}+1$ denotes the number of Gaussian peaks retained in the truncated expansion. Larger $s_{\text{max}}$ and smaller $\sigma$ correspond to an increase in the effective Hilbert space dimension. We proceed to examine how the measurement complexity of our proposed methods for learning the state $\ket{\bar{0}}_{\text{GKP}}$ with wavefunction in (\ref{state:GKP}) scales with this truncation dimension.
\begin{figure*}
    \centering 
   \includegraphics[width=0.85\linewidth]{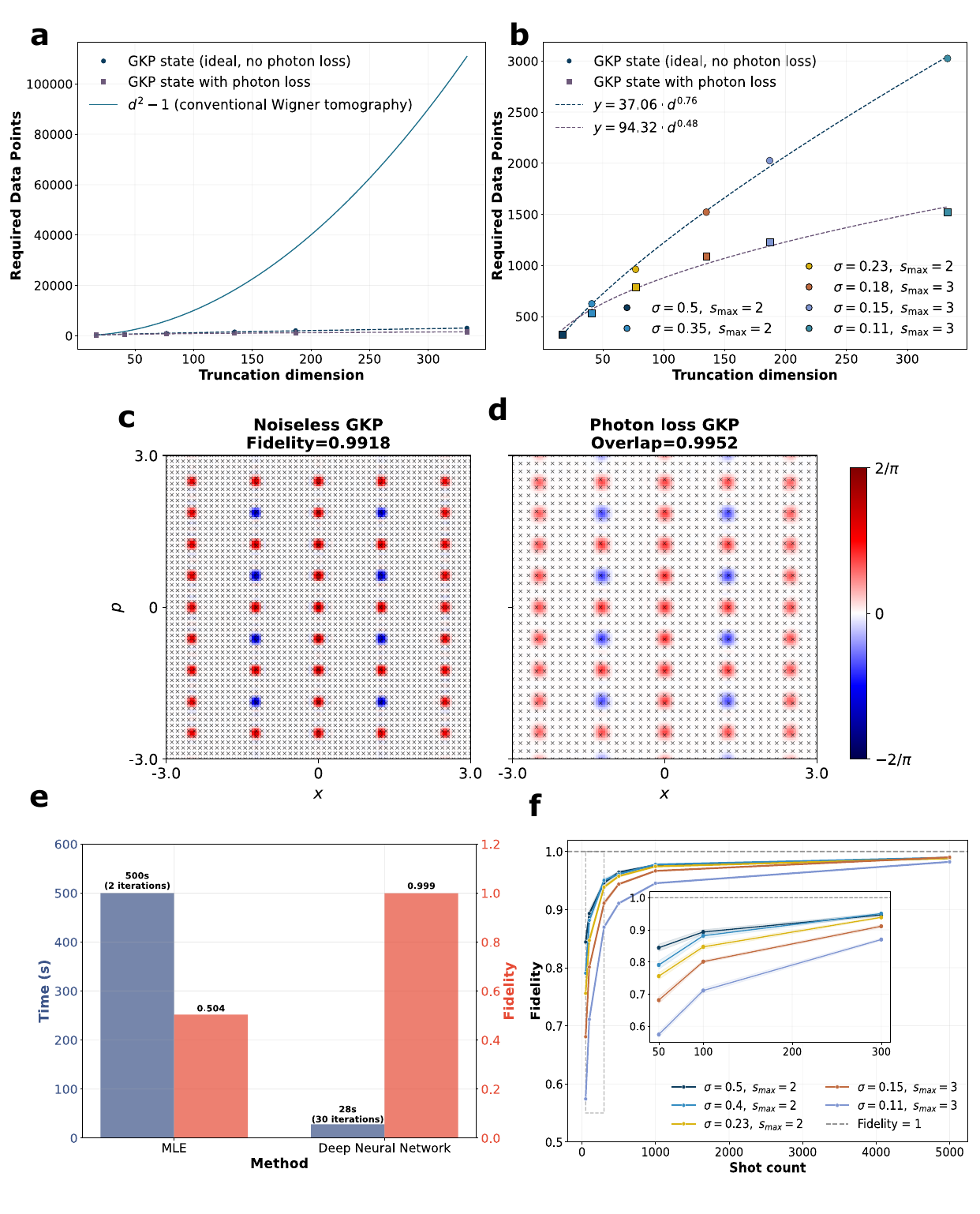}     
   \caption{Reconstruction of Wigner functions for simulated GKP states.
In Subfig.~\textbf{a}, the number of phase-space points required to achieve a fidelity of $0.99$ is compared for a noiseless GKP state and a GKP state subject to photon loss as a function of the truncated Hilbert-space dimension. The scaling required for information completeness in conventional Wigner tomography ($d^2-1$) is shown as a reference.
In Subfig.~\textbf{b}, a detailed view of the data requirements presented in Subfig.~\textbf{a} is shown, illustrating the number of phase-space points needed for reconstructing ideal and lossy GKP states as a function of truncation dimension.
In Subfig.~\textbf{c}, the reconstructed Wigner function is shown for a noiseless GKP state ($\sigma=0.11$, $s_{\max}=3$) using 3025 phase-space points, achieving a fidelity of $0.99$. Black crosses mark the measurement grid.
In Subfig.~\textbf{d}, the reconstructed Wigner function is shown for the same GKP state subject to photon loss, where only 1521 phase-space points are required to reach a fidelity of $0.99$.
In Subfig.~\textbf{e}, the classical computational time of density matrix reconstruction via MLE is compared with that of Wigner function prediction using our neural network approach for a GKP state ($\sigma=0.4$, $s_{\max}=2$), together with the corresponding fidelities achieved. All computational runtimes were evaluated on an x86\_64 system running Ubuntu 22.04.5 LTS.
In Subfig.~\textbf{f}, the fidelity of the reconstructed Wigner function is shown as a function of the number of measurement shots per phase-space point for GKP states with varying squeezing parameters and numbers of Gaussian peaks. Each point represents the mean fidelity over 10 independent experiments, and the shaded region corresponds to one standard deviation across these experiments.}
    \label{fig:GKP}
\end{figure*}

We plot the sample complexity of our DNN model in reconstructing the Wigner function of simulated GKP states in Fig.~\ref{fig:GKP}\textbf{a} and Fig.~\ref{fig:GKP}\textbf{b}. 
As illustrated in Fig.~\ref{fig:GKP}\textbf{b}, the number of data points required by our algorithm to attain a fidelity above $0.99$ scales as $d^{0.76}$ for the noiseless GKP state $\ket{\bar 0}_{\text{GKP}}$ — a significant improvement over the $d^2$ scaling. Each data point in Fig.~\ref{fig:GKP}\textbf{b} corresponds to a combination of $(\sigma, s_{max})$, where smaller values of $\sigma$ and larger values of $s_{max}$ correspond to higher truncation dimension. 

Photon loss is a common quantum noise in bosonic quantum system to corrupt the quantum state. When photon loss errors are introduced, the resulting state becomes a mixed state. In this noisy case, the data requirement for achieving fidelity $0.99$ further decreases, scaling approximately as $d^{0.48}$. This reduction of measurement complexity matches our intuition that photon loss error typically reduces the non-Gaussianity of the state and hence reduces the difficulty of learning them. Figures~\ref{fig:GKP}\textbf{c} and~\ref{fig:GKP}\textbf{d} illustrate examples of reconstructed Wigner functions of GKP states without noise and with photon loss, respectively.

\begin{figure*}[htbp]
    \centering
    \includegraphics[width=0.9\linewidth]{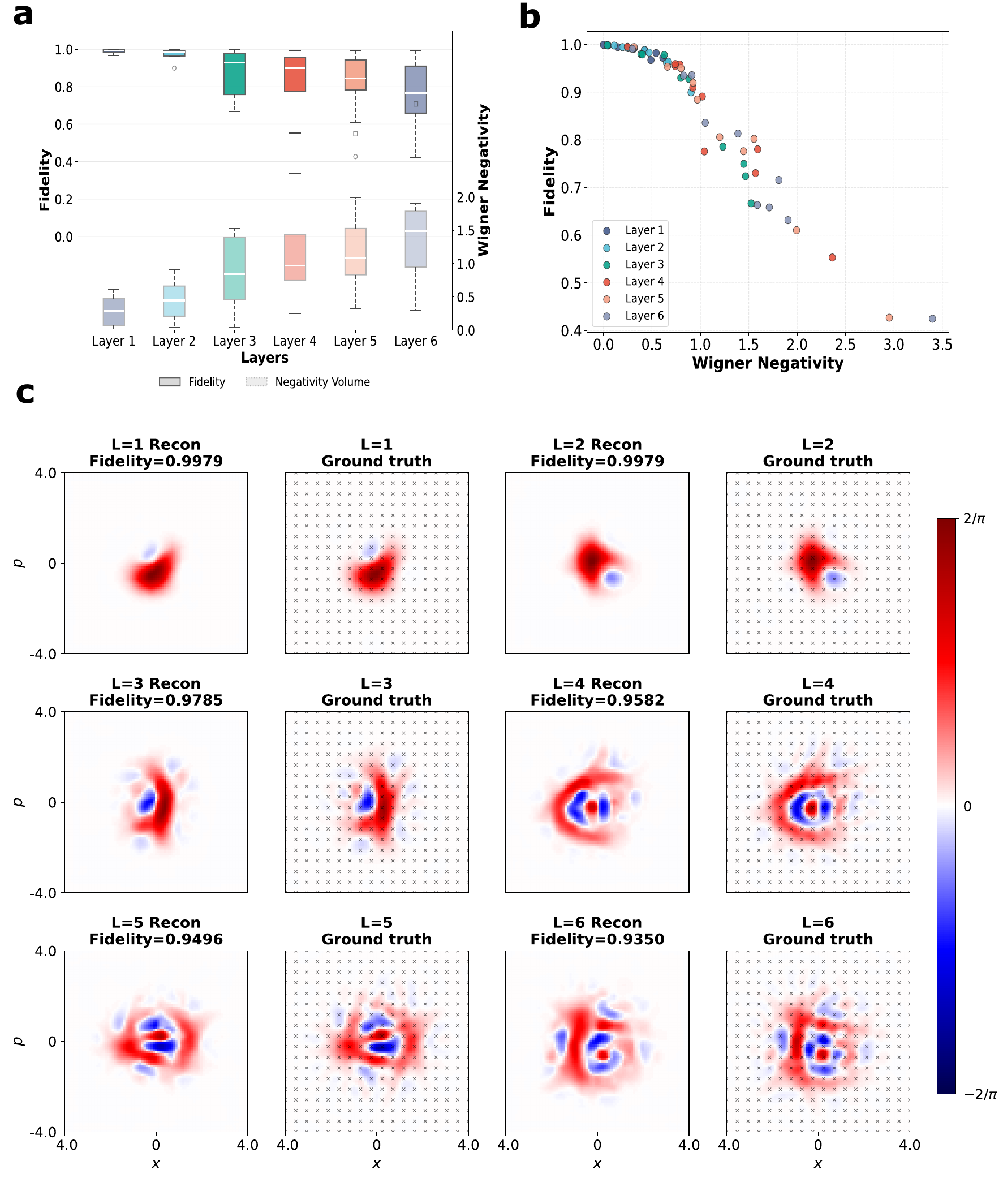}

    \caption{Performance of the DNN model for output states prepared using random displacement and SNAP circuits.
In Subfig.~\textbf{a}, the reconstruction fidelity is shown as a function of the number of circuit layers. The box plots summarize 10 independently generated output states, with the center line indicating the median and the box indicating the interquartile range.
In Subfig.~\textbf{b}, the reconstruction fidelity is shown as a function of the Wigner negativity of the corresponding output states.
In Subfig.~\textbf{c}, examples of Wigner function reconstructions are shown for states with circuit depths $L$ ranging from one to six using the DNN model. Each pair of images shows the reconstructed Wigner function alongside the ground truth. The measurement data grid is also overlaid on the ground truth plots.}
    \label{fig:layervsFidelity}
\end{figure*}

In quantum state reconstruction, not only the measurement complexity but also the classical computational complexity affects the practical efficiency of the learning algorithm.
We also compare the computational time required to obtain the density matrix using the conventional maximum likelihood estimation (MLE) approach against our proposed method based on Wigner function prediction. In both cases, the input consists of true values of the Wigner function evaluated on a $33 \times 33$ grid for a GKP state with parameters $\sigma = 0.4$ and $s_{\text{max}} = 2$. For MLE, the objective is to find a density matrix in a truncated Hilbert space of dimension $d = 29$ that maximizes the likelihood of observing the given Wigner function values. In contrast, our neural network approach aims to predict the Wigner function on a finer $81 \times 81$ grid. Reconstruction quality is evaluated in both cases by computing the fidelity between the reconstructed state and the ground truth.
The results in Fig.~\ref{fig:GKP}\textbf{e} demonstrate that our neural network method is significantly faster, requiring only about $28$ seconds to perform $30$ optimization iterations, while producing predictions with fidelity close to one. In comparison, each iteration of MLE is substantially more time-consuming, and the method struggles to converge to the ground truth when using the same $33 \times 33$ data grid. This failure of MLE can be explained by the fact that the observation data is too sparse and hence its optimization landscape is extremely flat.

In the above discussion, each input data point corresponds to a true value of the Wigner function at a phase-space point, but in a practical experiment, a finite number of measurements introduces statistical errors into the input data. To demonstrate how the number of measurement shots per data point influences reconstruction fidelity, we apply our trained model to simulated data of a finite number of measurement shots.  Figure~\ref{fig:GKP}\textbf{f} shows the fidelity of the Wigner function reconstructed by our algorithm as a function of the number of shots per measurement.
In this simulation, the number of input phase-space points is fixed to the minimum number required to reach fidelity $0.99$ in the noiseless setting shown in Fig.~\ref{fig:GKP}\textbf{b}.
Five different combinations of $(\sigma, s_{max})$ are studied, indicating that typically GKP states with smaller values of $\sigma$ and higher values of $s_{max}$ require a larger number of measurement shots per phase-space point to reach high reconstruction fidelity.

\subsubsection{Learning randomly prepared CV states}
\label{subsec:arbitraryState}

We next apply our DNN model to learn the Wigner functions of a broader class of states, i.e., CV states prepared by displacement and SNAP gates~\cite{heeres2015}. 
For a single-mode CV state, the set of displacement operators $D(\alpha)$ and SNAP gates $\text{SNAP}(\bm\theta)$ form a universal gate set~\cite{4rf7-9tfx}, which has been widely applied for preparing important classes of bosonic code states and Wigner-negative states~\cite{sivak2022, kudra2022}. The displacement operator $D(\alpha)$ is defined in Eq.~(\ref{eq:displacement}).
A SNAP gate is
\begin{equation}
    \text{SNAP}( \bm\theta):= \sum_{n\in \mathbb N} \exp(\text{i}\theta_n)\ket{n}\bra{n},
\end{equation}
where $\bm\theta =\{\theta_n\}_{n\in\mathbb N}$ denote the phases on different Fock basis.

Here we investigate learning the Wigner functions of quantum states prepared by a sequence of quantum gates consisting of displacement gates and SNAP gates. Each layer of gates consists of a displacement gate $D(\alpha)$ followed by a SNAP gate, $\text{SNAP}(\bm \theta)$. Here the displacement amount $\alpha$ in each displacement gate is randomly chosen within $|\alpha|\le 1.5$, and each $\theta_i$ in $\bm \theta$ are randomly selected within $[0, 2\pi)$. For each state, we truncate the Hilbert space to $d=30$. Given Wigner function values at points at $18\times 18$ phase-space grid points, the neural network reconstructs the Wigner function at a finer grid of $81\times 81$ within the same region in phase space. 

Figure~\ref{fig:layervsFidelity}\textbf{a} shows the reconstruction fidelity as a function of the number of gate layers. As the circuit depth increases, the average reconstruction fidelity decreases from nearly unity to approximately $0.8$ when the same amount of training data is used. Intuitively, applying additional gate layers generates increasingly intricate phase-space structures, making the corresponding Wigner functions more difficult to learn accurately.
One possible measure of this complexity is the Wigner negativity~\cite{mattia2021}, defined as
$\int_{\mathbb{R}^2}\mathrm{d}x \mathrm{d}p |W_\rho(x,p)|-1$.
To further examine this connection, Fig.~\ref{fig:layervsFidelity}\textbf{b} plots the reconstruction fidelity versus the Wigner negativity for all generated states. The results show that small Wigner negativity has little observable effect on the reconstruction fidelity. However, beyond a certain regime, the reconstruction fidelity decreases approximately linearly as the Wigner negativity increases.

Fig.~\ref{fig:layervsFidelity}\textbf{c} illustrates several examples of reconstructed Wigner functions together with the ground truth for different numbers of layers of quantum circuits. Instead of exhibiting a structural pattern, the patterns of the Wigner functions here are more random, but do not contain any high-frequency oscillations. This feature makes our DNN model better fitted for learning these Wigner functions. Compared with the regression model using Gabor-frame feature map (see Fig.~\ref{fig:layervsFidelityGabor} in Appendix), our DNN model produces more accurate predictions typically. Figure~\ref{fig:experimental_data} presents a comparison of our DNN model with regression models on reconstructing an experimental state prepared by applying a SNAP gate to a coherent state~\cite{Ahmed2021PRL}.

\begin{figure*}[htbp]
    \centering
       \includegraphics[width=0.75\linewidth]{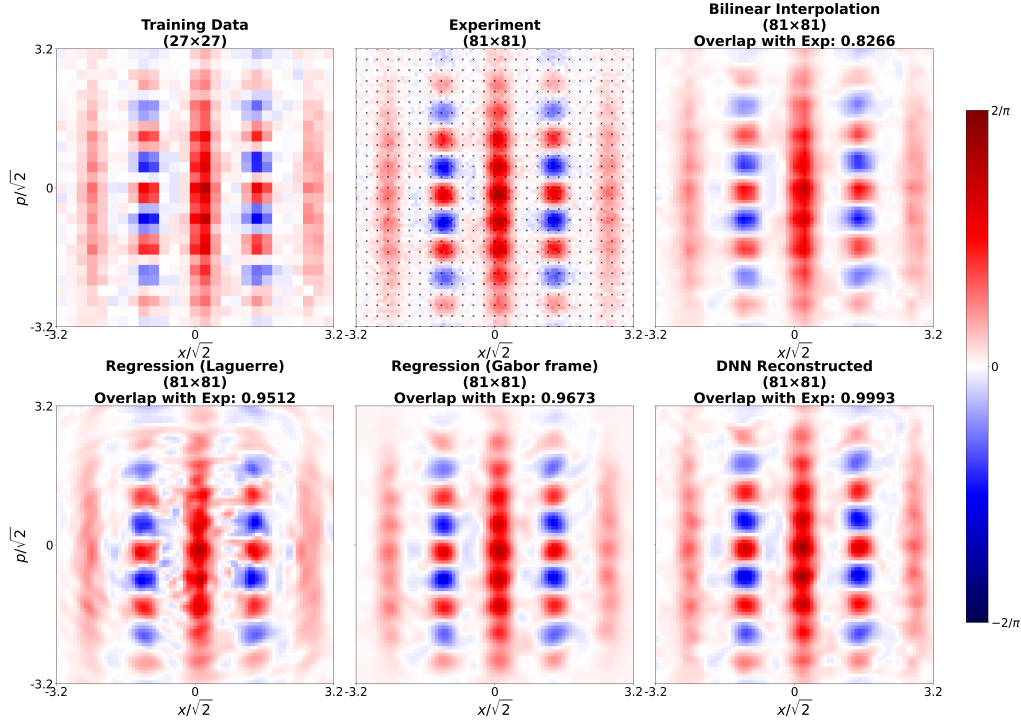}
    \caption{Reconstruction of measured Wigner functions for experimental GKP states. From upperleft to downright panels: the training dataset, downsampled to a $27 \times 27$ grid from the original experimental measurements; the full experimental data collected on an $81 \times 81$ grid; the Wigner function reconstructed via bilinear interpolation from the $27 \times 27$ grid; the Wigner function reconstructed by the regression models using Laguerre feature map, the Wigner function reconstructed by Gabor frame feature maps, and the Wigner function predicted by our DNN model. In the predicted Wigner function, each cross marks a phase-space point belonging to the original $27 \times 27$ measurement grid. }
    \label{fig:gkp_samples}
\end{figure*}

\begin{figure*}[htbp]
    % \centering
\includegraphics[width=0.8\linewidth]{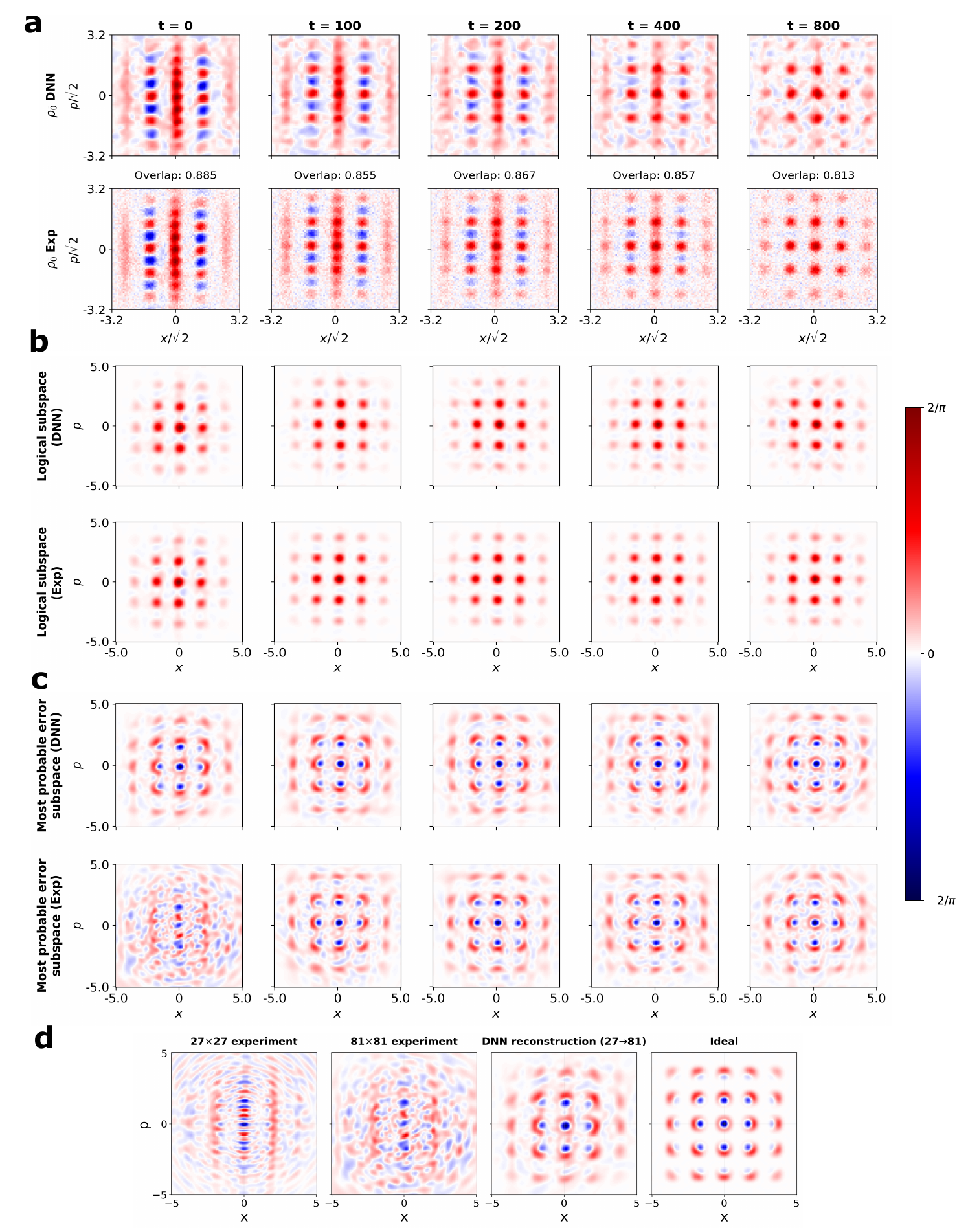}
\caption{
Analysis of experimental GKP states and the extracted logical and error subspaces.
In Subfig.~\textbf{a}, the Wigner functions of the representative logical state $\rho_{\bar{0}}$ are shown after $t=0$, $100$, $200$, $400$, and $800$ rounds of QEC. The first row shows the $81\times81$ Wigner functions reconstructed by the DNN from the corresponding experimental $27\times27$ measurements, while the second row shows the experimentally measured $81\times81$ Wigner functions. The overlap between each DNN reconstruction and its experimental counterpart is indicated between the two Wigner functions.
In Subfig.~\textbf{b}, the Wigner functions of the normalized projection onto the logical subspace,
$(\ket{\tilde{0}}\bra{\tilde{0}}+\ket{\tilde{1}}\bra{\tilde{1}})/2$,
extracted from the mixed state
$\rho_{\mathrm{mix}}=(\rho_{\bar{0}}+\rho_{\bar{1}})/2$
using Eq.~(\ref{specDec}), are shown. The upper row is obtained from the DNN reconstructed Wigner functions, and the lower row from the experimental data.
In Subfig.~\textbf{c}, the Wigner functions of the normalized projection onto the most probable error subspace,
$(\ket{E_{\tilde{0}}}\bra{E_{\tilde{0}}}+\ket{E_{\tilde{1}}}\bra{E_{\tilde{1}}})/2$,
are shown. The corresponding results obtained from the DNN reconstruction and the experimental data are shown in the upper and lower rows, respectively.
In Subfig.~\textbf{d}, the extracted most probable error subspace at $t=0$ is compared. The first three panels correspond to the experimental $27\times27$ data, the experimental $81\times81$ data, and the DNN reconstruction from the corresponding $27\times27$ measurements, respectively. The rightmost panel shows the Wigner function of the projection onto the error subspace with single-photon-addition error
$\hat{a}^{\dagger}(\ket{\bar{0}}+\ket{\bar{1}})
(\bra{\bar{0}}+\bra{\bar{1}})
\hat{a}$
(up to normalization). The overlaps between the first three reconstructed error-subspace projections and the ideal error subspace shown in the rightmost panel are $0.449$, $0.439$, and $0.693$, respectively.
}
    \label{fig:gkp_com}
\end{figure*}

\section{Demonstration on experimental data}
\label{sec:experiment}
Although the numerical experiments above, based on classically simulated data, demonstrates the effectiveness of our learning approaches, experimental data are inevitably affected by systematic and stochastic noises. This challenge is particularly pronounced in bosonic quantum error correction (QEC) experiments, where Wigner tomography is routinely performed over dense phase-space grids and repeated across multiple QEC cycles to characterize and validate the bosonic code states throughout the QEC process. Such procedures impose a substantial measurement overhead, which is expected to grow rapidly as future bosonic processors achieve larger effective Hilbert-space dimensions.

Therefore, to establish our DNN model as a practical tool for quantum state characterization, it is essential to demonstrate its robustness under realistic experimental conditions.
To this end, we apply our DNN model to experimentally prepared GKP states in the QEC experiment~\cite{sivak2023real} and evaluate its performance on real measurement data.
We demonstrate that our DNN model can not only reconstruct experimental GKP states from sparse noisy data, but also enables diagnosis of the dominant quantum error process and quantitative characterization of the logical code states during the QEC process.

Specifically, we consider noisy GKP code states prepared in the demonstration of beyond-break-even quantum error correction~\cite{sivak2023real}. The state is prepared in a circuit-QED architecture consisting of a cavity mode coupled to an auxiliary transmon qubit. The experimental dataset comprises Wigner-function measurements on an $81\times81$ phase-space grid, where each point is obtained from the empirical average of repeated measurement shots. To train the DNN model, we use only a sparse $27\times27$ subset of the measured grid. The remaining points in the original $81\times81$ dataset are kept hidden from training and are used only as the dense experimental reference for evaluation. After training, the model predicts the Wigner function over the entire $81\times81$ grid.

Figure~\ref{fig:gkp_samples} compares the reconstructed Wigner function of an experimentally prepared logical-zero GKP state $\rho_{\bar{0}}$ with the experimentally measured reference data.
By denoting the full measurement dataset on a phase-space grid $\Gamma$ as $\{(x,p), y_{(x,p)}\}_{(x,p)\in\Gamma}$, we use 
$\frac{\pi}{2} \Delta\sum_{(x,p) \in \Gamma} \widehat{W}(x,p) y_{(x,p)}$,
where $\Delta$ is the area of each cell in $\Gamma$,
to quantify the overlap between the predicted Wigner function and the full measurement data.
This figure of merit is a variant of Eq.~(\ref{eq:overlap}) by replacing true Wigner function values by their measurement results.
The overlap exceeds $0.99$, demonstrating excellent agreement. We further benchmark the DNN against bilinear interpolation and the sparse regression model. As shown by Fig.~\ref{fig:gkp_samples}, bilinear interpolation achieves significantly lower prediction accuracy than either learning-based approach, while the regression model is noticeably more sensitive to experimental noise than the DNN.

To further evaluate the practical performance of the DNN framework, we consider experimentally prepared GKP states subjected to multiple rounds of quantum error correction. Specifically, we study the code states $\rho_{\bar{0}}(t)$ and $\rho_{\bar{1}}(t)$ after $t=0,100,200,400$, and $800$ QEC cycles implemented in Ref.~\cite{sivak2023real}. 
%For each state, measurements from only a sparse $27\times27$ phase-space grid are used to reconstruct the Wigner function over the full $81\times81$ grid. 
The resulting reconstructions for $\rho_{\bar{0}}$, together with the experimentally measured references, are presented in Fig.~\ref{fig:gkp_com}\textbf{a}. The corresponding results for $\rho_{\bar{1}}$ are shown in Fig.~\ref{fig:wigner_nn_vs_exp} of the Appendix. % We also plot the marginal distributions of our reconstructed Wigner function in both position and momentum bases, namely the position and momentum probability distributions, together with those obtained from complete experimental data in Fig.~\ref{fig:wigner_nn_vs_exp}.

The reconstructed Wigner functions further enable the identification of the experimental logical subspace and its dominant error subspace in this bosonic QEC experiment. Since logical code states and error code states cannot be cleanly isolated from other eigenstates directly in phase space, this analysis requires reconstruction of the underlying density matrix. To this end, we reconstruct density matrices in a truncated Hilbert space of dimension $d=32$ from the Wigner-function data using maximum-likelihood estimation. We denote by $\rho_{\mathrm{DNN}}$ the density matrices reconstructed from the DNN-predicted Wigner functions and by $\rho_{\mathrm{EXP}}$ those reconstructed from the complete experimental dataset.

To retrieve the logical and error subspaces, we perform the spectral decomposition of the mixture
\begin{widetext}
\begin{equation}\label{specDec}
\rho_{\mathrm{mix}}(t)=\frac{\rho_{\bar{0}}(t)+\rho_{\bar{1}}(t)}{2},
\end{equation}
yielding
\begin{equation}
\rho_{\mathrm{mix}}
=\lambda_0 \ket{\tilde{0}}\bra{\tilde{0}}
+\lambda_1\ket{\tilde{1}}\bra{\tilde{1}}
+\lambda_2 \ket{E_{\tilde{0}}}\bra{E_{\tilde{0}}}
+\lambda_3 \ket{E_{\tilde{1}}}\bra{E_{\tilde{1}}}
+\cdots,
\end{equation}
\end{widetext}
where the eigenvalues satisfy $\lambda_0\ge \lambda_1\ge \lambda_2\ge \lambda_3$, while  QEC cycle number $t$ and those terms with eigenvalues smaller than $\lambda_3$ are omitted. 
The subspace spanned by ${\ket{\tilde{0}},\ket{\tilde{1}}}$ defines the experimental logical subspace, which deviates from the ideal logical subspace spanned by ${\ket{\bar{0}},\ket{\bar{1}}}$ due to experimental imperfections. Similarly, the subspace spanned by ${\ket{E_{\tilde{0}}},\ket{E_{\tilde{1}}}}$ corresponds to the dominant error subspace. 

Figure~\ref{fig:gkp_com}\textbf{b} and~\ref{fig:gkp_com}\textbf{c} visualize the Wigner functions associated with the reconstructed logical subspaces and dominant error subspaces for different QEC cycles. We compare results obtained from two different sources: full $81\times81$ experimental measurements and the DNN-predicted $81\times81$ Wigner functions reconstructed from sparse input data. A further comparison with the results obtained from sparse $27\times27$ experimental measurements is provided in Fig.~\ref{fig:wigner_subspaces} of the Appendix. The results show that the DNN reconstruction preserves not only the Wigner function itself, but also the logical and error-subspace structures inferred from the reconstructed density matrix.

To further identify the physical origin of the dominant error, we focus on the most probable error subspace at $t=0$, shown in Fig.~\ref{fig:gkp_com}\textbf{d}. Specifically, we plot the Wigner function of the normalized projector
$\frac{\ket{E_{\tilde{0}}}\bra{E_{\tilde{0}}}+\ket{E_{\tilde{1}}}\bra{E_{\tilde{1}}}}{2}$,
obtained from the spectral decomposition above. For comparison, we also plot the Wigner function of the normalized operator
$\hat{a}^\dagger (\ket{\bar{0}}+\ket{\bar{1}})
(\bra{\bar{0}}+\bra{\bar{1}})
\hat{a}$,
which corresponds to the ideal error subspace generated by photon-addition errors. Such photon-addition processes correspond to heating noise in the cavity mode and were experimentally observed on the same platform for $t=100,200,400$, and $800$ QEC cycles in Ref.~\cite{sivak2023real}. As shown in Fig.~\ref{fig:gkp_com}\textbf{c}, the error-subspace Wigner function inferred from the DNN reconstruction exhibits a much clearer resemblance to the photon-addition error pattern than reconstructions obtained directly from sparse experimental measurements or even from the full experimental dataset.
These contrasting results can be explained by the fact that the quantum error subspaces also exhibit structured, local patterns in phase space resembling images, making them easier to capture by our DNN model.

\begin{figure*}
    \centering
    \includegraphics[width=0.9\linewidth]{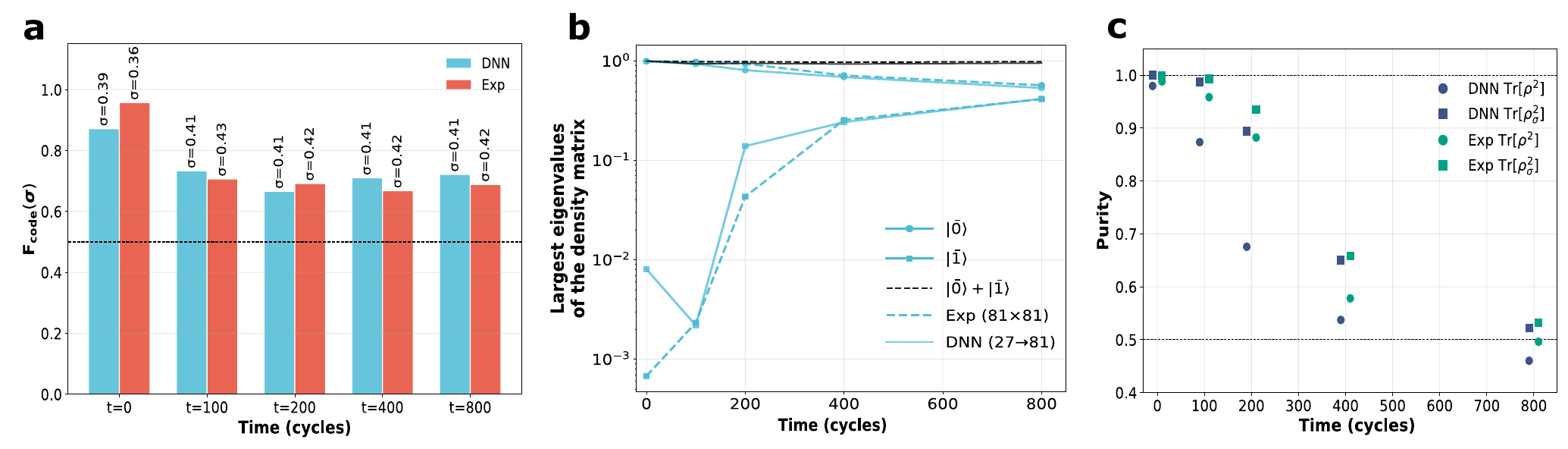}
    \caption{
Characterization of the reconstructed experimental state $\rho_{\bar{0}}(t)$ over QEC cycles, comparing the density matrices obtained from DNN predictions and full measurement data.
In Subfig.~\textbf{a}, the code-space fidelity in Eq.~(\ref{eq:codeSpaceFidelity}) of the reconstructed state $\rho_{\bar{0}}(t)$ is shown for $t = 0, 100, 200, 400$, and $800$ QEC cycles. For each time point, the fidelity is maximized over the squeezing parameter $\sigma \in [0.1, 0.5]$, with the optimal value annotated above each bar.
In Subfig.~\textbf{b}, the two dominant eigenvalues of $\rho_{\bar{0}}(t)$ from spectral decomposition, corresponding to the GKP logical codewords $\ket{\bar{0}}$ and $\ket{\bar{1}}$, together with their sum (the total code-space weight), are shown on a logarithmic scale for $t = 0, 100, 200, 400$, and $800$ QEC cycles. Solid and dashed lines correspond to DNN-based and full-data reconstructions, respectively.
In Subfig.~\textbf{c}, the full purity $\mathrm{Tr}[\rho^2]$ (circles) and projected purity (squares) of $\rho_{\bar{0}}(t)$ are shown for $t = 0, 100, 200, 400$, and $800$ QEC cycles. To confirm the preservation of logical information, the projected purity is defined as the purity of the state $\rho_{\sigma}=\Pi_{\sigma}\rho\Pi_{\sigma}/\mathrm{Tr}[\Pi_{\sigma}\rho\Pi_{\sigma}]$, where $\Pi_{\sigma}$ is the projection onto the GKP code space spanned by $\ket{\bar{0}}$ and $\ket{\bar{1}}$. Here, $\sigma$ is fixed at the value that maximizes the code-space fidelity at $t=200$ QEC cycles. Dashed horizontal lines indicate reference values of $0.5$ and $1.0$. Results from DNN-based and full-data reconstructions are compared.
}
    \label{fig:multiQECcycles}
\end{figure*}

Finally, we focus on benchmarking the quality of experimental GKP code state $\rho_{\bar{0}}$ using its reconstructed Wigner function.
We reconstruct the density matrix of the state $\rho_{\bar{0}}(t)$ from both the DNN-reconstructed and experimentally measured Wigner functions using MLE. The squeezing level of the experimental GKP state is then estimated by optimizing the parameter $\sigma$ to maximize the code-space fidelity
\begin{equation}
F_{\mathrm{code}}(\sigma)
=
\langle \bar{0}_{\mathrm{GKP}}(\sigma)
|\rho_{\bar{0}}(t)|
\bar{0}_{\mathrm{GKP}}(\sigma)\rangle 
+
\langle \bar{1}_{\mathrm{GKP}}(\sigma)
|\rho_{\bar{0}}(t)|
\bar{1}_{\mathrm{GKP}}(\sigma)\rangle,
\label{eq:codeSpaceFidelity}
\end{equation}
where $\ket{\bar{0}(\sigma)}_{\mathrm{GKP}}$ and $\ket{\bar{1}(\sigma)}_{\mathrm{GKP}}$ refer to the GKP code states with wavefunctions defined in Eqs.~(\ref{state:GKP}) and $s_{\text max}=12$, respectively.
Figure~\ref{fig:multiQECcycles}\textbf{a} compares the estimation from the DNN predictions and that from complete experimental data, indicating a close match between these two.  
To track the state transition within the code space,
Figure~\ref{fig:multiQECcycles}\textbf{b} illustrates the dominant pairs of eigenvalues in the spectrum decompositions of $\rho_{\bar{0},\mathrm{DNN}}$ and $\rho_{\bar{0},\mathrm{EXP}}$, together with their sums, for different numbers of QEC cycles. It can be seen that while the probability of the logical zero state decreases and that of the logical one state increases, their sum, namely the probability of occupying the code space keeps almost constant and close to unity. 
We show the purities of both the reconstructed states $\rho_{\bar{0},\mathrm{DNN}}$ and $\rho_{\bar{0},\mathrm{EXP}}$, together with the purity of their projections onto the associated code space in Fig.~\ref{fig:multiQECcycles}\textbf{c}. As demonstrated, the experimental GKP code state is initially close to a pure state at $t=0$, but over time, it approaches a maximally mixed state within the code space, with a purity of one-half. Collectively, these results demonstrate that by using only one-ninth of the experimental data, our DNN model can accurately benchmark the quality of the experimental GKP code state, yielding results closely matching those obtained from full measurement data.

\section{Discussions}
\label{sec:discussion}

In this paper, we have introduced two complementary machine learning algorithms for reconstructing Wigner functions from sparse pointwise measurements. 
The first is a sparse regression model, for which we theoretically establish upper bounds on the sample complexity required to learn Wigner functions of states that are sparsely supported on the Fock basis and those sparsely supported on the coherent state basis. The second is a DNN model that learns a neural implicit representation of the Wigner function from sparse data and can predict Wigner function values at arbitrary resolutions. This model is highly robust to practical noise and is applicable to experimental data.
Despite their different implementations, both methods formulate Wigner-function reconstruction as a continuous-function learning problem, aiming to minimize the number of training samples required to achieve a low expected prediction error.

The two approaches differ slightly in their learning formalisms. In the first, theoretically motivated formalism, given measurement data at a set of points $\{\alpha_i\}$ randomly sampled from an arbitrary distribution $\mathbb{D}$, the learning model aims to minimize the expected prediction error of the Wigner function values $W(\alpha)$ when $\alpha$ is drawn from the same distribution $\mathbb{D}$. This formalism casts Wigner function learning as a \textsf{PAC} learning problem, for which we prove our main results in Theorem~\ref{theorem:fock_ML} and Theorem~\ref{theorem:Gabor_ML}.
In the second, experimentally motivated formalism, given measurement data on a coarse-grained grid, the learning model targets predicting Wigner function values on a fine-grained grid at arbitrary resolution beyond those encountered during training. This setting is closer to practical Wigner tomography experiments, in which measurements are typically performed over a grid of points within a region of phase space around the origin.

 Below, we summarize and compare the performances of these two learning models on different types of states. 

\textit{Binomial code states.}
For these states, our regression model is provably efficient at reconstructing their Wigner functions. Numerical results in Fig.~\ref{fig:binomial} show no apparent increase in measurement complexity as the system size grows. In contrast, the DNN model cannot faithfully reconstruct the Wigner functions of binomial code states using the same amount of measurement data. The Wigner functions of CV states with sparse Fock support are typically highly oscillatory in phase space, and our DNN model struggles to learn these high-frequency patterns. This is due to two reasons. First, we use smoothed interpolated data as training inputs, and this data bias leads the DNN model to preferentially learn smooth patterns in phase space. Second, the spectral bias of DNN models~\cite{rahaman2019spectral} makes them preferentially learn low-frequency patterns over high-frequency oscillatory ones. 

\textit{Cat states.}
Our regression model is also provably efficient for cat states, with numerical results presented in Fig.~\ref{fig:cat}. The measurement complexity required by our DNN model, as shown in Appendix Fig.~\ref{fig:cat_nn}, is substantially higher than that of the regression model. With limited data, the DNN model struggles to learn the highly oscillatory interference fringes near the phase-space origin, again due to the data bias and spectral bias discussed above.

\textit{Practical GKP states.}
For these states, our DNN model requires lower measurement complexity than the regression model to achieve the same reconstruction accuracy (compare the performance in Fig.~\ref{fig:GKP} with that in Appendix Fig.~\ref{fig:regression_ml}).  On experimental data, the DNN model also provides much higher prediction accuracy relative to the experimental ground truth than the regression model.

\textit{Random CV states.}
The advantage of the DNN model over the regression model extends to general CV states prepared by random sequences of displacement and SNAP gates. Given the same amount of data, the DNN model achieves higher prediction accuracy in reconstructing their Wigner functions, as can be seen by comparing the results in Fig.~\ref{fig:layervsFidelity} with those of a regression model using Gabor feature maps in Appendix Fig.~\ref{fig:layervsFidelityGabor}.

In summary, the regression model exploits explicit sparsity in a suitable basis to enable efficient learning of Wigner functions for CV states with sparse structure. However, its performance is not guaranteed for states beyond the sparse-support regime. In contrast, the DNN model leverages the low-complexity structure of smooth and compressible Wigner functions, while also exhibiting strong robustness to measurement noise, enabling accurate reconstruction from realistic experimental data. The two learning models are therefore complementary. For CV states with sparse support in the Fock or coherent-state basis, the regression model provides an efficient and theoretically guaranteed approach. For more general CV states, particularly those arising in noisy experiments, the DNN model offers a more flexible and practical solution.

\section{Conclusions}
\label{sec:conclusion}
The key contributions of this work on characterizing CV quantum states can be summarized from three perspectives.
First, we introduce a novel framework for learning the phase-space representations of quantum states. Unlike conventional quantum state reconstruction, we develop supervised learning models that directly predict continuous Wigner functions over phase space from a discrete set of pointwise measurements. 
Second, our approaches substantially reduce the measurement complexity required to characterize a wide range of important CV states. We theoretically prove that for CV states with sparse Fock or coherent-state supports, our regression model can efficiently learn their Wigner functions from pointwise measurements. For CV states beyond these regimes, such as GKP states, numerical results indicate that our DNN model can accurately reconstruct the Wigner function using far fewer measurements than that required for information completeness.
Third, our DNN model provides a robust tool for reconstructing experimental CV states from sparse and noisy data, demonstrating strong noise resilience. We validate its performance using experimental data of GKP states. Remarkably, we find that our model can correctly identify the dominant quantum error process from only a few noisy data points, a task in which conventional estimation approaches fail.

In this paper, we have focused on reconstructing the Wigner functions of single-mode CV states. It is important to explore how the results presented here can be generalized to multi-mode Wigner function reconstruction. 
Beyond measuring displaced parity, measuring the photon number distribution followed by applying a displacement operation, and optimizing the set of phase-space points by maximizing the condition number can in principle, enhance the efficiency in reconstructing Wigner functions~\cite{PhysRevA.94.052327,PRXQuantum.6.010303}. How these improved approaches can be combined with our machine learning algorithms to further boost the efficiency in learning Wigner functions are beyond the scope of this paper and will be left for future investigation.  On the other hand, developing classical surrogate models, including both theoretical machine learning algorithms and DNN models, for predicting quantum properties of CV states beyond Wigner functions is another important direction to explore. 

\section*{Acknowledgments}
We thank Weizhou Cai and Pei Zeng for the stimulating discussions.
We appreciate Volodymyr Sivak for sharing both the experimental data and the code. T.X. acknowledge support from the National Key R\&D Program of China (No.\ 2025YFF0515504), the National Natural Science Foundation of China through Grant No.\ 62401359). 
X.T., Y.H.L.\ and Y.D.W.\ acknowledge funding from the National Natural Science Foundation of China through grants no.\ 12405022. Y.X.D. acknowledges the funding from A*STAR (H25-MRO3488). G.C.\ acknowledges support from the Hong Kong Research Grant Council through Grants No. 17307520, No. R7035-21F, and No. T45-406/23-R, the Ministry of Science and Technology through Grant No. 2023ZD0300600 and State Key Laboratory of Quantum Information Technologies and Materials.

\section{Appendix}

%{\color{red} [two regression models should be separated in two sections. For sparse coherent state supports, the order should be introduction of frame theory, Lemmas for Theorem 5, its corrolaries and finally Proof of Theorem 4.]}

The appendix is structured as follows. Sec.~\ref{app:relatedWork} provides the discussion on related works. Sec.~\ref{app:Fock_ML_proof} details the supervised learning setup of the regression model for learning sparsely supported states in the Fock basis and provides the proof of Theorem~\ref{theorem:fock_ML}. Sec.~\ref{app:Gabor_ML_proof} is dedicated to the sample complexity of learning states with sparse coherent state support. It is broken up into subsections that introduce various technical preliminaries, leading to the proof of Theorem~\ref{theorem:Gabor_ML}. 
Sec.~\ref{app:complexity_of_conventional_approaches} explains the sample complexity of two conventional approaches: dense sampling and interpolation. 
Sec.~\ref{app:architecture_training} elaborates the DNN models, including the architecture and the training. Finally, Sec.~\ref{app:details_of_numexp} provides the details of the numerical examples presented in the main text and provides additional experiments for the regression models and DNN model.

\subsection{Related works}
\label{app:relatedWork}
Interpolation-based approaches, such as Lagrange interpolation~\cite{PhysRevLett.120.090501}, have been successfully applied to Q-function tomography of small-scale CV states; however, they still face limitations in the sparse and noisy regime. Lagrange interpolation can exactly fit all data points, but this exactness also makes it highly sensitive to practical noise in Wigner tomography: even small estimation errors can lead to large deviations in the reconstructed Wigner function. Meanwhile, its computational cost also increases rapidly with the number of data points and with the dimension of the phase space. 

Most theoretical studies on learning CV quantum states have focused on Gaussian measurements, including homodyne and heterodyne measurements~\cite{becker2024,mele2024,oh2024,PhysRevResearch.6.033280,zhao2025complexity,liu2025quantum,chen2026towards}. Many of the theoretical efforts~\cite{ohliger2011,gandhari2024,mele2024,zhao2025complexity} target at reconstructing the density matrix of a CV state up to small trace distance, following the objective of quantum state tomography. Other efforts focus on efficiently predicting the expectation values of a set of observables, in the spirit of shadow tomography~\cite{becker2024,PhysRevResearch.6.033280}. Our theoretical setting differs from previous works in several respects.
First, the measurements we consider are restricted to displaced parity measurements, which directly provide pointwise estimates of the Wigner function. Second, the object to be learned is a continuous phase-space representation of the quantum state, but not a density matrix. Consequently, our learning objective is to predict the Wigner function over phase space, and the reconstruction error is quantified by the mean-squared prediction error of the Wigner function rather than by a distance measure between density operators. Third, our theoretical formalism falls within \textsf{PAC} learning, rather than shadow tomography. Here the central question is to minimize the size of the training dataset, namely, the number of phase-space measurement points. 
Besides CV state learning, there are also closely related works on CV state characterization, such as state verification~\cite{Aol15} and fidelity estimation~\cite{qd1c-1fk9}.

DNN models have been widely applied for predicting quantum properties. For CV quantum systems, DNN models have been applied for the detection of quantum entanglement~\cite{PhysRevLett.132.220202,gao2024classifying}, quantum non-Gaussianity~\cite{cimini2020}, and quantum similarity~\cite{wu2023}. In these works, training is performed in a supervised manner, with quantum properties serving as labels. However, in our DNN model, training is done in a self-supervised way. 

Within the large family of deep learning models, generative models have been widely employed for quantum-state representation by learning the underlying probability distributions of measurement outcomes~\cite{torlai2018,carrasquilla2019}. This paradigm has also been extended to CV quantum systems, where neural generative models have been used to represent Husimi-Q functions~\cite{dugan2023q}. However, unlike the Husimi-Q function, the Wigner function is generally not a probability distribution and can take negative values. Consequently, it cannot be directly modeled using standard generative-learning frameworks.
Besides, generative models have also been used for reconstructing density matrices of CV quantum states~\cite{tiunov2020,Ahmed2021PRL,Ahmed2021PRR}.
In contrast to neural quantum-state representations based on generative models, our approach formulates Wigner tomography as a regression problem. Rather than learning a probability distribution, the neural network is trained to approximate a continuous phase-space function from sparse pointwise measurements. The resulting model serves as a neural implicit representation of a CV quantum state in phase space, enabling the prediction of Wigner-function values at arbitrary phase-space coordinates.

Regarding different learning paradigms for representing and characterizing quantum systems, a comprehensive review is given in Ref.~\cite{du2025artificial}. A recent study~\cite{zhao2025rethink} compares the performance of theoretical machine learning models and that of deep learning models in predicting quantum properties of many-body ground states, and finds that theoretical machine learning models yield competitive and even better performance. In this work, we also compare the performance of both the regression model and the DNN model across different state reconstruction tasks, and find that their respective strengths are complementary.

\subsection{Proof of Theorem~\ref{theorem:fock_ML}}
\label{app:Fock_ML_proof}

In this section, we elaborate on the setup of the regression models in the supervised learning framework and provide the proof of Theorem~\ref{theorem:fock_ML}, namely the sample complexity of the regression model for learning states with sparse Fock-basis support.

Recall that the task is to predict the Wigner function of a single-mode state with sparse Fock support, given access to statistical estimates of different phase-space points. Mathematically, the density matrix is \(\rho = \sum_{n,m=0}^{d-1} \rho_{mn} |m\rangle\langle n|\), and its Wigner function is  
\begin{equation}
    W_{\rho}(\alpha) = \sum_{n,m=0}^{d-1} \rho_{mn} W_{|m\rangle \langle n|} (\alpha),
\end{equation}
where
\begin{widetext}
\begin{equation}
        W_{|m\rangle \langle n|} (\alpha) = \frac{2}{\pi} (-1)^m \sqrt{\frac{m!}{n!}} (2\alpha)^{n-m} L_m^{(n-m)}\left( 4|\alpha|^2\right) \exp\left(-2|\alpha|^2  \right),
    \label{eq:fock_wiger}
\end{equation}
and
\begin{equation}\label{LaguerreP}
L_m^{(n-m)}(x) = \frac{e^x x^{-(n-m)}}{m!} \frac{d^m}{dx^m}\bigl(e^{-x} x^{n}\bigr) \quad \mathrm{with} \quad x \in \mathbb{R}
\end{equation}
\end{widetext}
 is the generalized Laguerre polynomial, also known as the associated Laguerre polynomial, of degree \(m\). Given this explicit function expansion, a simple regression model can be employed to learn the coefficients $\rho_{mn}$.

We formalize the regression problem in the supervised learning framework.
Let the parameter space be a bounded region in the phase space $\mathcal{X} =  \{\alpha \in \mathbb{R}^2 : |\alpha|\leq R\} $, equivalently parametrized in polar coordinates by $ (r,\phi)\in[0,R]\times [0,2\pi)$. Let the output space be \(\mathcal{Y} = [-\frac{2}{\pi},\frac{2}{\pi}]\), since the value of the Wigner function is bounded between $\frac 2 \pi$ and $-\frac 2\pi$. To obtain real-valued basis functions from Eq~\eqref{eq:fock_wiger}, use the polar coordinate $\alpha = re^{i\phi}$ and note that
\begin{widetext}
\begin{align}
   W_{|m\rangle \langle n|} (r, \phi) = \frac{2}{\pi} (-1)^m \sqrt{\frac{m!}{n!}} (2r)^{n-m} L_m^{(n-m)}\left( 4r^2\right) \exp\left(-2r^2  \right) \exp(i (n-m)\phi).
\end{align}
Finally, we decompose the exponential into trigonometric functions to obtain a real basis for learning real-valued Wigner functions, namely, for all \(0\leq n,m \leq d-1\), the feature map is composed of basis functions of the form
\begin{align}
    \Phi_{m,n,1} (r,\phi) = \frac{2}{\pi} (-1)^m \sqrt{\frac{m!}{n!}} (2r)^{n-m} L_m^{(n-m)}\left( 4r^2\right) \exp\left(-2r^2  \right) \cos((n-m)\phi), \\
    \Phi_{m,n,-1} (r,\phi) = \frac{2}{\pi} (-1)^m \sqrt{\frac{m!}{n!}} (2r)^{n-m} L_m^{(n-m)}\left( 4r^2\right) \exp\left(-2r^2  \right) \sin((n-m)\phi).
\end{align}
\end{widetext}
Without loss of generality, each feature is indexed by \(\bomega \in \Omega := \{ (n,m,\pm1) : 0\leq n < m \leq d-1\} \cup \{ (n,n,1): 0\le n\le d-1\}\) since when $n=m$, $\sin((n-m)\phi)$ is identically zero and therefore \(d_f = |\Omega| = d^2-d+d = d^2\). The feature map is given explicitly as \(\mathbf{\Phi}: \mathcal{X} \to \mathcal{D} = [-\frac{2}{\pi},\frac{2}{\pi}]^{d_f}\). We refer to $\mathcal{D}$ as the feature space with dimension $d_f$. Note that this is a bold symbol corresponding to \(\alpha \mapsto \mathbf{\Phi}(\alpha) = [\Phi_{\bomega_1}(\alpha),\Phi_{\bomega_2}(\alpha), \dots, \Phi_{\bomega_{d_f}}(\alpha)]\). The hypothesis class is given as the linear span of these basis functions with a one-norm regularization on the parameters:
\begin{align}
    \mathcal{H} = \left\{ h(\alpha) = \sum_{\bomega \in \Omega} \eta_{\bomega} \Phi_{\bomega}(\alpha) : \|\bm{\eta}\|_1 \leq t \right\}.
    \label{eq:Fock_ML_hypothesis_class}
\end{align} 
Given the dataset \(\mathcal{T} = \dataset\) where each \(\alpha_i\) is sampled from the distribution \(\mathbb{D}\), the surrogate \(\widehat{W}\) is obtained by solving the following Lasso regression:
\begin{align}
    \min_{\bm{\eta}} \sum_{i=1}^{\mathcal N} \left| \langle \bm{\eta},\mathbf{\Phi} (\alpha_i)\rangle - y_i \right|^2 \; \text{, subject to } \; \|\bm{\eta}\|_1 \leq t,
    \label{eq:Lasso_Fock}
\end{align}
where $\langle \bm{\eta},\mathbf{\Phi} (\alpha_i)\rangle = \sum_{\bomega \in \Omega} \eta_{\bomega} \Phi_{\bomega}(\alpha_i)$. In other words, we are to minimize the empirical error on the dataset $\mathcal{T}$ given by
\begin{equation}
    \widehat{\mathsf{R}}_{\mathcal{T}}(\widehat{W}_{\bm \eta}) := \frac{1}{\mathcal N} \sum_{i=1}^{\mathcal N} |\widehat{W}_{\bm \eta}(\alpha_i) - y_i|^2.
\end{equation}

We are now ready to state the generalization error of the regression model, defined by
\begin{equation}
    \mathsf{R}(\widehat{W}) := \mathbb{E}_{\alpha \sim \mathbb{D}} \left| \widehat{W}(\alpha) - W(\alpha) \right|^2.
\end{equation}

\begin{theorem}[Restatement of Theorem~\ref{theorem:fock_ML}] 
     Suppose the target state $\rho$ satisfies $\sum_{n,m=0}^{d} \mathbbm{1}\{\rho_{mn} \neq 0\} \leq s^2$ and all training data in \(\mathcal{T}\) satisfies \(|W(\alpha_i) - y_i| \leq \epsilon_0\) with probability at least $1-\delta'$ with $\delta'=\delta/\mathcal{N}$ for some $0<\delta <1/2$ over the measurement randomness. Let $\widehat{W}$ be the surrogate obtained by solving Eq.~\eqref{eq:Lasso}, and assume the empirical optimization satisfies
    \begin{align}
    \widehat{\mathsf R}_{\mathcal{T}}(\widehat{W})
    \le
    \inf_{h\in\mathcal H}\widehat{\mathsf R}_{\mathcal{T}}(h)
    +
    \epsilon_{1}.
    \end{align}
    Then, with probability at least $1-2\delta$,
    \begin{align}
       \!\mathsf R (\widehat{W}) \leq  \epsilon_0^2 \!+\! \epsilon_1 + \!\frac{32s^2}{\pi^2\sqrt{{\mathcal N}}}\! \left(\!\! \sqrt{2\log(2d^2)} \!+\! \sqrt{\frac12 \log \frac 1 \delta} \right).
    \end{align}    
    In particular, if
    \(
    \epsilon_0^2 \le \frac{\epsilon}{4},
    \:
    \epsilon_1 \le \frac{\epsilon}{4},
    \)
    and
    \begin{align}
    {\mathcal N} \ge \frac{8192s^4}{\epsilon^2 \pi^4} \log \left( \frac{4d^4}{\sqrt{\delta}} \right),
    \end{align}
    then
    \(
    \mathsf R(\widehat{W}) \le \epsilon 
    \)
    with probability at least \(1-2\delta\).
    \label{theorem:fock_ML_formal}
\end{theorem}

We use a standard result from statistical learning theory, which bounds the generalization error of the Lasso regression as follows.

\begin{lemma}[Ref.~\cite{mohri_foundations_2018}, Theorem 11.16]
    Let $x\in \mathcal{X}$.
    Let the feature space be \(\mathcal{D} \subseteq \mathbb{R}^{d_f}\) and the hypothesis class be \(\mathcal{H} = \{ \mathbf{\Phi}(x) \in \mathcal{D} \mapsto h(x) = \langle \bm \eta,\mathbf{\Phi}(x) \rangle : \| \bm \eta \|_1 \leq t \}\). Assume that there exists \(r_\infty > 0\) such that for all \(\mathbf{\Phi}(x) \in \mathcal{D}\), \(\|\mathbf{\Phi}(x)\|_\infty \leq r_\infty\) and \(M > 0 \) such that \(|h(x) - y| \leq M\) for all \((\mathbf{\Phi}(x),y) \in \mathcal{D} \times \mathcal{Y}\). Then, for any \(\delta > 0\), with probability at least \(1-\delta\), the following inequality holds for all \(h \in \mathcal{H}\):
    \begin{align}
    {\mathsf R}(h) \leq \widehat{\mathsf R}_{\mathcal{T}}(h) + 2r_\infty t M \sqrt{\frac{2 \log (2d_f)}{{\mathcal N}}} + M^2\sqrt{\frac{\log \frac{1}{\delta}}{2{\mathcal N}}},
    \end{align}
    where $\widehat{\mathsf R}_{S}(h)$ is the training error of the hypothesis $h$ on the dataset $\mathcal{T}$ with size ${\mathcal N}$, and ${\mathsf R}(h)$ is the generalization error of $h$.
    \label{lemma:Lasso}
\end{lemma}

We will also employ the following lemma that bounds the empirical error on the dataset $\mathcal{T}$. It states that the training error can be bounded by the statistical error of the dataset.

\begin{lemma}[Empirical error from approximate Lasso optimization]
Following the assumptions in Theorem~\ref{theorem:fock_ML_formal} and assume that the parameter $t$ in Eq.~\eqref{eq:Fock_ML_hypothesis_class} is chosen such that the true Wigner function is in $\mathcal{H}$, let \(\widehat{W}(\alpha)\) be returned by solving Eq.~\eqref{eq:Lasso_Fock} to an additive optimization accuracy \(\epsilon_{1}\), i.e.
\begin{align}
\widehat{\mathsf R}_{\mathcal{T}}(\widehat{W})
\le
\inf_{\|\bm{\eta}\|_1\le t}
\frac1{\mathcal N}\sum_{i=1}^{\mathcal N}
\left|\langle \bm{\eta},\mathbf{\Phi}(\alpha_i) \rangle-y_i \right|^2
+ \epsilon_1.
\end{align}
If the data satisfy
\(
|W(\alpha_i)-y_i|\le \epsilon_0
\) for all $i \in [{\mathcal N}]$ with probability at least $1-\delta'$, then
\begin{align}
\widehat{\mathsf R}_{\mathcal{T}}(\widehat{W})
\le
\epsilon_0^2 + \epsilon_1
\end{align}
with probability at least $1-\delta$.
\label{lemma:empirical_error_bound}
\end{lemma}
\begin{proof}
We manipulate the definition as:
    \begin{align}
    \begin{aligned}
        &\inf_{\|\bm{\eta}\|_1\le t}
    \frac1{\mathcal N}\sum_{i=1}^{\mathcal N}
    \left|\langle \bm{\eta}, \mathbf{\Phi}(\alpha_i) \rangle-y_i \right|^2  \\
    =& \inf_{h \in \mathcal{H}} \frac1{\mathcal N}\sum_{i=1}^{\mathcal N} \left| h(\alpha_i)-y_i \right|^2 \\
    % \leq &\inf_{h \in \mathcal{H}} \frac1{\mathcal N}\sum_{i=1}^{\mathcal N} \left( \left| h(\alpha_i) -W(\alpha_i) \right| + \left| W(\alpha_i)-y_i \right| \right)^2 \\
    \leq & \frac1{\mathcal N}\sum_{i=1}^{\mathcal N}  \left| W(\alpha_i)-y_i \right|^2 \\
    \leq & \epsilon_0^2,
    \end{aligned}
    \end{align}
where the first inequality comes from the fact that the target Wigner function is included in the hypothesis class following the choice of $t$, i.e., $W\in\mathcal{H}$.
\end{proof}

With these two lemmas, we are now ready to prove Theorem~\ref{theorem:fock_ML_formal}.
\begin{proof}[Proof of Theorem~\ref{theorem:fock_ML_formal}]
    The proof reduces to determining the relevant parameters $t,r_\infty$, and $M$ in Lemma~\ref{lemma:Lasso} in our setting. Then, combining with the empirical error in Lemma~\ref{lemma:empirical_error_bound} yields the result.

    First, recall that the parameter $t$ determines the size of the function in the hypothesis class $\mathcal H$ by regulating $\|\bm \eta\|_1$. We should choose $t$ such that the target Wigner function \(W_{\rho}(\alpha) = \sum_{n,m=0}^{d} \rho_{mn} W_{|m\rangle \langle n|} (\alpha) \) that is $s^2$-sparse lies within the hypothesis class.
    % In other words, we require $ \sum_{n,m=0}^{d} |\rho_{mn}| \le t$ .
    Since we assume these coefficients satisfies $\sum_{n,m=0}^{d-1} \mathbbm{1}\{\rho_{mn} \neq 0\} \leq s^2$, by Cauchy-Schwarz,
    \begin{align}
    \begin{aligned}
        \sum_{n,m=0}^{d-1} |\rho_{mn}| &\leq \sqrt{s^2}\left( \sum_{n,m=0}^{d-1} |\rho_{mn}|^2\right)^{1/2} \\
        &= s \sqrt{\tr(\rho^2)} \\
        &\leq s
    \end{aligned}
    \label{eq:thm3_one_norm_bound}
    \end{align}
    where the equality follows from $\sum_{n,m=0}^{d-1}|\rho_{mn}|^2 = \tr(\rho^\dagger \rho) = \tr(\rho^2)$ and the last inequality is the consequence of $\tr(\rho^2) \le 1$ for any state. One subtlety remains. In Eq.~\eqref{eq:thm3_one_norm_bound}, the coefficients of the Wigner function $\rho_{mn}$ are complex. In practice, the Lasso is solved over the real feature map in Eq.~\eqref{eq:Fock_ML_hypothesis_class} using real coefficients. The two are equivalent up to a factor as $|\rho_{mn}| \le |\Re[\rho_{mn}]| + |\Im[\rho_{mn}]| \le \sqrt{2} |\rho_{mn}|$.
    Therefore, the value of $t = \sqrt{2}s$ ensures that the real hypothesis class encompasses the Wigner function of the $s^2$-sparse states under consideration.

    Next, we determine the suitable value for $r_\infty$. From the definition of the Wigner function as $W_{|m\rangle \langle n|}(\alpha) = \frac{2}{\pi}\tr(|m\rangle \langle n| D(\alpha)(-1)^{\hat n} D(-\alpha))$ where \((-1)^{\hat n}\) is the Parity operator, by Hölder's inequality, we have
    \begin{align}
    \begin{aligned}
        & \tr(|m\rangle \langle n| D(\alpha)(-1)^{\hat n} D(-\alpha) ) \\
        \le& \||n\rangle \langle m|\|_1\|D(\alpha) (-1)^{\hat n} D(-\alpha)\|_\infty  \\
        =& \||n\rangle \langle m|\|_1,
    \end{aligned}
    \end{align}
    since \(D(\alpha) (-1)^{\hat n} D(-\alpha)\) is unitary. With \(\||n\rangle\langle m|\|_1 = 1\), we conclude that \(\|\mathbf\Phi(\alpha)\|_\infty = \sup_\alpha|W_{|n\rangle\langle m|}(\alpha) |\leq \frac{2}{\pi}\). Thus, we can take $r_\infty = \frac{2}{\pi}$.

    Finally, we need a bound on $M$. For all $h \in \mathcal{H}$,
    \begin{align}
    \begin{aligned}
        |h(\alpha) - y| &\leq |h(\alpha)| + |y| \\
        &\leq  |\langle \bm{\eta}, \mathbf\Phi (\alpha)\rangle | + |W(\alpha) | + \epsilon_0 \\
        &\leq \| \bm{\eta} \|_1 \|\mathbf\Phi (\alpha)\|_\infty + \frac 2 \pi +\epsilon_0 \\
        &\leq \frac{2t}{\pi} + \frac 2 \pi + \epsilon_0 \\
        &\leq \frac{4t}{\pi} \\
        &\leq \frac{4\sqrt{2}s}{\pi}
    \end{aligned}
    \end{align}
    where the second inequality follows from the noisy estimate of the Wigner function and the third from Hölder's inequality. The fourth inequality employs the bound $\|\bm{\eta}\|_1\leq t \le \sqrt{2}s$ and $\|\mathbf \Phi (\alpha) \|_\infty \leq \frac 2\pi$, respectively, which are found above. 
    
    Thus, it is sufficient to take $t=\sqrt{2} s$, $r_\infty = \frac{2}{\pi}$, and $M = \frac{4\sqrt{2}s}{\pi}$. Note that $t$ and $M$ depend nicely on the sparsity parameter $s$ instead of the ambient dimension of the feature space $d_f$, which gives rise to the efficient sample complexity.

    Now, inserting these values and Lemma~\ref{lemma:empirical_error_bound} into Lemma~\ref{lemma:Lasso}, and using a union bound yields with probability at least $1-2\delta$,
    \begin{align}
       \mathsf R (\widehat{W}) \leq  \epsilon_0^2 + \epsilon_1 + \frac{32s^2}{\pi^2\sqrt{{\mathcal N}}} \left( \!\!\sqrt{2\log(2d^2)} + \sqrt{\frac{ \log \frac 1 \delta }{2} } \right).
    \end{align}

    To bound the generalization error by $\epsilon$, we require that the estimation error from the data satisfies $\epsilon_0^2 \leq \epsilon/4$ and the additive optimization accuracy $\epsilon_1\leq \epsilon/4$. This translates to an empirical error bounded by $\epsilon/2$. We separately bound the generalization term by requiring
    \begin{align}
        \frac{32s^2}{\pi^2\sqrt{\mathcal N}} \left( \sqrt{2\log(2d^2)} + \sqrt{\frac12 \log \frac 1 \delta} \right) \leq \frac \epsilon 2.
    \end{align}
    Rearranging the inequality and using $(\sqrt{A} + \sqrt{B})^2 \leq 2 (A+B)$, we get
    \begin{align}
    \begin{aligned}
        {\mathcal N} &\ge \frac{8192s^4}{\epsilon^2 \pi^4} \left( 2\log(2d^2) + \frac 1 2 \log \frac 1 \delta \right) \\
        &= \frac{8192s^4}{\epsilon^2 \pi^4} \log \left( \frac{4d^4}{\sqrt{\delta}} \right). 
    \end{aligned}
    \end{align}
\end{proof}

Now we determine the sample complexity by considering estimating the Wigner function at a point $\alpha$ with the displaced-parity measurement, i.e., measuring the observable $D(\alpha)(-1)^{\hat n} D(-\alpha)$.
With the random variable $X_{\alpha} \in \{+1,-1\}$, the Wigner function is given as $W(\alpha) = \frac 2 \pi \tr(\rho D(\alpha)(-1)^{\hat n} D(-\alpha)) = \frac 2 \pi \mathbb E [X_\alpha ]$. Therefore, the average of $M$ results of the random variable, denoted by $\hat X_\alpha$, can be used to approximate the Wigner function. Recall that we tolerate a pointwise estimation error no larger than $\sqrt{\epsilon/4}$. By Hoeffding's inequality,
\begin{align}
\operatorname{Pr} \left(\frac{2}{\pi} \left| \hat X_\alpha  - \mathbb E[\hat X_\alpha] \right| \ge \sqrt{\frac{\epsilon}{4}} \right) \le 2 e^{-\epsilon\pi^2M/{32}}.
\end{align}
Finally, using the union bound, the probability that any of the estimates at ${\mathcal N}$ different phase-space points is off by more than $\sqrt{\epsilon/4}$ is at most ${\mathcal N}\cdot 2e^{-\epsilon\pi^2M/32}$. Bounding this failure probability by $\delta$, we get
\begin{align}
M = O\left( \frac{1}{\epsilon} \log  \frac{{\mathcal N}}{\delta}\right)
\end{align}
for each pointwise estimation.
Together, the total number of samples needed to achieve $\mathsf R (\widehat{W}) \leq  \epsilon$ with probability at least $1-2\delta$ is
\begin{align}
M\cdot{\mathcal N} = O \left( \frac{s^4}{\epsilon^3} \log \frac{d^4}{\sqrt{\delta}} \left( \log \frac{s^4}{\epsilon^2\delta} + \log\log \frac{d^4}{\sqrt{\delta}} \right)\right)
\end{align}
for an $s^2$-sparse state on a truncated Hilbert space of dimension $d+1$.

\subsection{Proof of Theorem~\ref{theorem:Gabor_ML}}
\label{app:Gabor_ML_proof}
In this section, we provide the sample complexity bound for learning states with sparse coherent-state support. To do so, we first introduce frame theory (Sec.~\ref{sec:frame_theory}) and Gabor frames (Sec.~\ref{sec:Gabor_STFT}). These preliminaries, together with four lemmas in Sec.~\ref{sec:lemmas_for_Gabor_trunc}, culminate in the proof of Proposition~\ref{proposition:Gabor_truncation_bound} in Sec.~\ref{sec:proof_gabor_truncation}. Its corollaries (Sec.~\ref{sec:corollaries}) are pivotal in the proof of Theorem~\ref{theorem:Gabor_ML}, which is presented lastly in Sec.~\ref{sec:Gabor_ML_proof}, along with the formal supervised learning setup.

\subsubsection{Frame theory}
\label{sec:frame_theory}

Here, we present a minimal introduction to frame theory. See Refs.~\cite{heil_basis_2011,waldron_introduction_2018} for more on frame theory. Frames can be regarded as an overcomplete basis in a Hilbert space.
Some important frames in finite-dimensional quantum computing include mutually unbiased bases (MUBs) and symmetric-informationally complete POVM (SIC-POVMs)~\cite{waldron_introduction_2018}. 

We start by recalling how to represent any element in a Hilbert space, which leads to the concept of a basis.
For example, if $\{e_j\}_{j\in J}$ is an orthonormal basis for a Hilbert space $\mathcal H$, then every $f\in\mathcal H$ admits an unique expansion
\begin{align}
    f = \sum_{j\in J} \langle f,e_j\rangle e_j,
\end{align}
where $\langle \cdot, \cdot \rangle$ is the inner product. With the norm induced by the inner product, Parseval's identity gives
\begin{align}
    \|f\|^2
    =
    \sum_{j\in J} |\langle f,e_j\rangle|^2 .
\end{align}
Thus, an orthonormal basis provides three useful properties: every element can be represented, the representation is unique, and the coefficients are stable in the sense of Parseval's identity.

However, in many applications, it is useful to work with systems that carry redundancy. Such systems are typically not linearly independent, so the coefficients of expansion need not be unique. 
As a result, the exact Parseval identity for orthonormal bases must be replaced by a weaker stability condition: the coefficient energy should remain comparable to the norm of the vector. This leads to the definition of a frame.

\begin{definition}[Frame]
Let $\mathcal H$ be a Hilbert space. 
A countable collection $\{\varphi_j\}_{j\in J}\subset \mathcal H$ is called a frame for $\mathcal H$ if there exist constants
$0<A\le B<\infty$ such that, for every $f\in \mathcal H$,
\begin{align}
    A\|f\|^2
    \le
    \sum_{j\in J} |\langle f,\varphi_j\rangle|^2
    \le
    B\|f\|^2 .
    \label{eq:frame_condition}
\end{align}
The constants $A$ and $B$ are called frame bounds. If $A=B$, then $\{\varphi_j\}_{j\in J}$ is called a tight frame.
\end{definition}

Intuitively, the lower bound ensures that there are no missing directions in the space that cannot be detected with the frame via the inner product. If all coefficients are small, then the function itself must be small. 
The upper bound prevents the coefficient sequence from being unstable or excessively large compared to the norm of the function. 
Together, these conditions ensure that a frame provides a stable but possibly redundant system for the Hilbert space. Note that an orthonormal basis is a tight frame with unit frame bounds.

In contrast to an orthonormal basis, the general reconstruction of $f$ with a frame $\{\varphi_j\}_{j\in J}$ is not simply given by $\sum_{j\in J} \langle f,\varphi_j\rangle \varphi_j $
Instead, a second frame is required, called the dual frame.
\begin{definition}[Dual frame]
Let $\{\varphi_j\}_{j\in J}$ be a frame for $\mathcal{H}$.
Another frame $\{\psi_j\}_{j\in J}$ for $\mathcal H$ is called a dual frame of $\{\varphi_j\}_{j\in J}$ if, for every $f\in\mathcal H$,
\begin{align}
        f
    =
    \sum_{j\in J}
    \langle f,\psi_j\rangle \varphi_j = \sum_{j\in J} \langle f, \varphi_j\rangle \psi_j .
    \label{eq:frame_reconstruction}
\end{align}
\end{definition}
In general, there are multiple dual frames that satisfy Eq.~\eqref{eq:frame_reconstruction}. A canonical choice is obtained from the frame operator $S$, defined by 
\begin{align}
    S f
    =
    \sum_{j\in J}
    \langle f,\varphi_j\rangle \varphi_j .
    \label{eq:frame_operator}
\end{align}
The frame condition in Eq.~\eqref{eq:frame_condition} implies that \(S\) is bounded, positive, self-adjoint, and invertible, which we will not prove. One can think of the frame operator as the generalization of the resolution of the identity to a redundant system, which corresponds to $S = \sum_{j \in J} |\varphi_j\rangle\langle\varphi_j|$ in physics notation. Another characterization of the tight frame is such that $S$ is proportional to the identity.
Then observe that
\begin{align}
f = S^{-1} Sf = \sum_{j\in J}
    \langle f,\varphi_j\rangle S^{-1}\varphi_j.
\end{align}
Hence, by Eq.~\eqref{eq:frame_reconstruction}, the set
\begin{align}
    \{\widetilde{\varphi}_j\}_{j\in J} = \{S^{-1}\varphi_j\}_{j \in J}
\end{align}
forms a dual frame of $\{\varphi_j\}_{j\in J}$, called the canonical dual frame. From the definition, it can be shown that the canonical dual frame has frame bounds
$1/B$ and $1/A$. Among all dual frames $\{\psi_j\}_{j\in J}$, the canonical dual provides one with the minimal norm of $\sum_{j\in J}|\langle f, \psi_j\rangle|^2$.

\subsubsection{Gabor frame and STFT}
\label{sec:Gabor_STFT}
Before diving into the Gabor frame, it is expedient to introduce the central objects in time-frequency analysis: translation and modulation operators, and collect a few important properties of theirs along the way. An excellent reference to the Gabor frame and the short-time Fourier transform is given by Ref.~\cite{grochenig_foundations_2001}.

Formally, for $x,\omega\in \mathbb R^d$ and $f\in L^2(\mathbb R^d)$, the translation operator is defined as
\begin{align}
T_x f(t) = f(t-x),
\end{align}
and the modulation operator as
\begin{align}
M_\omega f(t) = e^{2\pi i \omega\cdot t} f(t).
\end{align}
The combinations $T_xM_\omega$ or $M_\omega T_x$ are called time-frequency shifts. They satisfy the commutation relation
\begin{align}
    T_xM_\omega = e^{-2\pi i\omega \cdot x} M_\omega T_x.
    \label{eq:time-frequency_comm}
\end{align}
For a window function $g \in L^2(\mathbb R^d)$, the Gabor system is generated by time-frequency shifts on the window function given by $\mathcal{G}(g,a,b) = \{M_{b n}T_{a k}g : \: n, k \in \mathbb{Z}^d\}$. Here, the parameters $a,b \in (0,\infty)$ determine the size of the lattice and $\mathbb Z^d$ denotes the $d$-fold cartesian product of $\mathbb Z$.
We work exclusively with the Gaussian window $g(x) = e^{-\pi |x|^2}$ where $|x|^2 = \langle x,x\rangle$. It is well known that if $ab < 1$, $\mathcal G(g,a,b)$ forms a frame for $L^2(\mathbb R^d)$ \cite{grochenig_foundations_2001}. It is commonly referred to as a Gabor frame or a lattice Gabor frame to emphasize that the time-frequency shifts are periodic. An element of a Gabor frame is called a Gabor atom.

For a Gabor frame $\mathcal{G}(g,a,b)$, the dual frame is another Gabor frame $\mathcal{G}(\gamma,a,b)$ with the dual window function $\gamma \in L^2(\mathbb R^d)$. For our purpose, we take $\mathcal{G}(\gamma,a,b)$ to be the canonical dual frame of $\mathcal{G} (g,a,b)$ with the Gaussian window $g$. In this case, the canonical dual window is related to the original window by $\gamma = S_{g,a,b}^{-1} g$ where $S_{g,a,b}$ is the frame operator for $\mathcal{G}(g,a,b)$.
Then, from Eq.~\eqref{eq:frame_reconstruction}, any $f \in L^2(\mathbb R^d)$ admits the Gabor expansion
\begin{align}
f = \sum_{n,k \in \mathbb Z^d} \langle f,M_{bn}  T_{ak}\gamma \rangle M_{bn}  T_{ak}g.
\end{align}

Next, to transform between the time and frequency domains, we recall the action of Fourier transforms.
\begin{definition}[Fourier transform]
    The Fourier transform of a function $f$ is defined as
    \begin{align}
    \mathcal{F} (f(x)) = \widehat f(\xi) = \int_{\mathbb R^d} f(x) e^{-2\pi x \cdot  \xi} dx .
    \end{align}
\end{definition}
We use $\mathcal{F}$ when we wish to emphasize that it is a linear and unitary operator. For time-frequency shifts, we have the following identity
\begin{align}
    \mathcal{F}(M_\omega T_x f) = T_\omega M_{-x} \widehat f.
    \label{eq:fourier_time-frequency_shifts}
\end{align}
That is, the translation in the time domain corresponds to the modulation in the frequency domain and vice versa.

The Fourier transform decomposes a function in the time domain into its frequency components. In doing so, however, we are completely ignorant of where these frequency components are localized in time. This limitation is important when the local oscillatory behavior of a function carries relevant information.

The short-time Fourier transform, which provides a joint time-frequency representation, addresses this by first localizing the function with a window and then applying the Fourier transform. 
More precisely, to probe the frequency content of \(f\) near a point \(x\), one multiplies \(f\) by a window function centered at \(x\) and computes the Fourier transform of the resulting localized function. 
Varying \(x\) yields a joint time-frequency description of \(f\).

\begin{definition}[Short-time Fourier transform]
    Fixing a window function $g$, then the short-time Fourier transform (STFT) of $f$ with respect to $g$ is
    \begin{align}
        \mathcal{V}_gf(x,\omega) = \int_{\mathbb R^d} f(t) \overline{g(t-x)} e^{-2\pi i t\cdot \omega} dt, \quad \text{for } x,\omega\in \mathbb R^d.
    \end{align}
\end{definition}
Notice that the STFT is a map from $L^2(\mathbb R^d)$ to $L^2(\mathbb R^{2d})$; the result is a function of both time and frequency.
It can be equivalently expressed in terms of the translation and modulation operators as
\begin{align}
    \mathcal{V}_gf(x,\omega) = \langle f, M_\omega T_xg\rangle = \langle\, \widehat f \, , T_\omega M_{-x} \widehat g \, \rangle.
    \label{eq:STFT_inner_product}
\end{align}
The first equality follows from the definition. For the second equality, we use the fact that the Fourier transform is a unitary operator, which implies
\begin{align}
\begin{aligned}
    \langle f, M_\omega T_xg\rangle &=  \langle f, \mathcal{F}^{-1} \mathcal{F}(M_\omega T_xg)\rangle \\
    &= \langle \mathcal{F} f, \mathcal{F}(M_\omega T_xg)\rangle \\
    &= \langle\, \widehat f \, , T_\omega M_{-x} \widehat g \, \rangle,
\end{aligned}
\end{align}
where the last equality follows from Eq.~\eqref{eq:fourier_time-frequency_shifts}.

As an interesting side note, observe that the Wigner function of a quantum state $\psi \in L^2(\mathbb R)$ is just the STFT of $\psi(t)$ with the window function $\check \psi = \psi(-t)$ in disguise, up to a constant factor. By definition,
\begin{align}
    \mathcal V_{\check \psi} \psi (2x,2p) = \int \psi(t) \overline{\psi(-(t-2x))} e^{-2\pi i t 2p} dt.
\end{align}
Using a change of variables $t = x + y$,
\begin{align}
\begin{aligned}
    \mathcal V_{\check \psi} \psi (2x,2p) &= \frac{1}{2} \int \psi\!\left(x+ y \right) \overline{\psi\!\left(x-y\right)} e^{-2\pi i (x+y) 2p} dy \\
    &= \frac 1 2 e^{-4\pi i xp}  \int \psi\! \left(x+ y \right) \psi^*\!\left(x - y \right) e^{-2\pi i p y} dy   \\
    &= \frac \pi 2 e^{-4\pi i xp} W(x,p).
\end{aligned}
\end{align}
Thus, the Wigner function is a joint distribution in position and momentum, obtained from the STFT of the wavefunction. However, in the following, we treat the Wigner function itself as the basic object. Consequently, the Gabor expansion in the following involves the STFT, or finding the localized patterns of the Wigner function, not the wavefunction.

\subsubsection{Lemmas for Proposition~\ref{proposition:Gabor_truncation_bound}}
\label{sec:lemmas_for_Gabor_trunc}
In this section, we provide four lemmas for the proof of Proposition~\ref{proposition:Gabor_truncation_bound}. It is advised to first follow the proof of Proposition~\ref{proposition:Gabor_truncation_bound} and reference it accordingly with the appropriate context.

The first concern the decay of the STFT. Since we are dealing with well-behaved window functions like the Gaussians, it turns out that we can provide quantitative statements on the rate of decay for a class of functions enjoying the same property. It is well known in the STFT literature that if $f,g \in \mathcal{S}(\mathbb R^d)$ the Schwartz space, then $\mathcal{V}_g f \in \mathcal{S}(\mathbb R^{2d})$ \cite{grochenig_foundations_2001}. The Schwartz class gives a superpolynomial type decay. For (sub-)exponential decay, we need the functions to lie within a subspace of the Schwartz space. One well-studied example in time-frequency analysis is the so-called Gelfand-Shilov space. 
 
Ref.~\cite{chung_characterizations_1996} proved an equivalent characterization of the Gelfand-Shilov space in terms of the decaying properties of a function in the real space and the frequency space, which we take as the definition.
\begin{definition}[Gelfand--Shilov space]
Let \(s,r\ge 1/2\). The Gelfand--Shilov space
\(S^s_r(\mathbb R^d)\) consists of all functions
\(f\in C^\infty(\mathbb R^d)\) such that there exist constants
\(C,h,k>0\) satisfying
\begin{align}
    |f(x)| \le C e^{-k|x|^{1/r}},
    \qquad
    |\widehat f(\xi)| \le C e^{-h|\xi|^{1/s}}
\end{align}
for all \(x,\xi\in\mathbb R^d\).
\label{def:Gelfand--Shilov}
\end{definition}
The proof in the following only involves the symmetric Gelfand--Shilov space where $s = r= 1$, denoted as $S^1_1$.

\begin{lemma}
    Let $g,g_0$ be Gaussian functions and $\mathcal{G}(g,a,b)$ be a lattice Gabor frame for $L^2(\mathbb R^d)$, then the canonical dual window $\gamma \in S^1_1(\mathbb R^d)$. Furthermore, $\mathcal{V}_\gamma g_0 \in S^1_1(\mathbb R^{2d})$. In particular,
    \begin{align}
    \mathcal{V}_\gamma g_0(z) \lesssim e^{-\nu |z|}
    \end{align}
    for some $\nu>0$.
    \label{lemma:dual_STFT_decay}
\end{lemma}
We use the notation $f\lesssim g$ to mean that there exists a constant $C>0$ such that $f \le C\cdot g$.
\begin{proof}
    This result follows from combining a few known results. 
    First, observe that for Gaussians, $g,g_0\in S^{1/2}_{1/2} \subset S^1_1$. If $\gamma \in S^1_1$, then it is a known result from Ref.~\cite{cordero_gabor_2015} that the STFT of $g_0$ with the window function $\gamma$ remains in the same space. In other words, $\mathcal{V}_\gamma g_0 \in S^1_1$ and the proof is finished.

    It remains to show that $\gamma \in S^1_1$. From Def.~\ref{def:Gelfand--Shilov}, we need to separately show
    $|\gamma(x)| \lesssim e^{-k|x|} $ and $|\widehat \gamma(\xi)| \lesssim e^{-h|\xi|} $ for some $k,h>0$. Writing the Frame operator for $\mathcal{G}(g,a,b)$ as $S_{g,a,b}$, we have $\gamma = S^{-1}_{g,a,b} g$. Notice that the Gaussian window satisfies $|g(x)| \lesssim e^{-\lambda |x|}$ for some $\lambda >0$ trivially. The result from Ref.~\cite{strohmer_approximation_2001} then states that the canonical dual window also satisfies $|\gamma(x)|\lesssim e^{-k |x|}$ for some $k > 0$.

    Next, we show that the Fourier transform $\widehat \gamma$ also has the same decay. Taking the Fourier transform of $\gamma$, we have
    \begin{align}
    \widehat \gamma = \mathcal{F}\gamma = \mathcal{F}S^{-1}_{g,a,b} g= \mathcal{F}S^{-1}_{g,a,b} \mathcal{F}^{-1} \widehat g = \left( \mathcal{F}S_{g,a,b} \mathcal{F}^{-1} \right)^{-1} \widehat g .
    \end{align}
    We study the operator $\mathcal{F}S_{g,a,b} \mathcal{F}^{-1} $ by inspecting its action on some function $f$ with $\mathcal{F}^{-1}f = h$. Then,
    \begin{align}
    \begin{aligned}
        \mathcal{F}S_{g,a,b} \mathcal{F}^{-1}  f &= \mathcal{F}S_{g,a,b} h \\
        &= \mathcal{F} \sum_{n,k} \langle h, M_{bn} T_{ak} g \rangle  M_{bn} T_{ak} g \\
        &= \sum_{n,k} \langle h, M_{bn} T_{ak} g \rangle \mathcal{F} ( M_{bn} T_{ak} g ) \\
        &= \sum_{n,k} \langle \widehat h,  T_{bn}  M_{-ak} \widehat g \rangle \mathcal{F} ( M_{bn} T_{ak} g ) \\
        &= \sum_{n,k} \langle f,  T_{bn}  M_{-ak} \widehat g \rangle  T_{bn} M_{-ak} \widehat g \\
        &= S_{\widehat g, b, a} f.
    \end{aligned}
    \end{align}
    The fourth equality employs the relation between STFT and the Fourier transform from Eq.~\eqref{eq:STFT_inner_product} and the fifth follows from the Fourier transform of the time-frequency shifts in Eq.~\eqref{eq:fourier_time-frequency_shifts}. The last equality uses the definition of the frame operator in Eq.~\eqref{eq:frame_operator}.
    Therefore, we show that $\mathcal{F}S_{g,a,b} \mathcal{F}^{-1} = S_{\widehat g, b, a}$ and consequently,
    \begin{align}
    \widehat \gamma = \left( \mathcal{F}S_{g,a,b} \mathcal{F}^{-1} \right)^{-1} \widehat g = S_{\widehat g, b, a}^{-1} \widehat g.
    \end{align}
    
    This establishes that $\widehat \gamma$ is the canonical dual window of $\widehat g$ where $g$ is our original Gaussian window. However, the Fourier transform of $g(x) = e^{-\pi |x|^2}$ is the same Gaussian in the frequency space, i.e., $\widehat g(\xi) = e^{-\pi |\xi|^2}$. Invoking the result of Ref.~\cite{strohmer_approximation_2001} again, there exists some $h>0 $ such that $|\widehat \gamma (\xi)| \lesssim e^{-h|\xi|}$. Finally, the characterization of the Gelfand-Shilov space in Def.~\ref{def:Gelfand--Shilov} concludes that $\gamma \in S^1_1$.
\end{proof}

Next, to bound the tail in the truncation involving infinite summation over exponential functions, we use the following lemma.

\begin{lemma}[Uniform exponential lattice tail]
Let $n \in \mathbb Z^2$ and $z\in \mathbb R^2$.
Let $a,\mu,R > 0$ with $aN >R$, then for any $|z|\le R$, there exist constants $C,C'>0$ such that
\begin{widetext}
\begin{align}
    \sum_{\|n\|_\infty>N} e^{-\mu |an-z|}
    \le
    C
    \bigl(1+aN-R\bigr)e^{-\mu(aN-R)} \le C'e^{-\mu'(aN-R)}
\end{align}
for some $0<\mu'<\mu$ where $\|n\|_\infty = \max\{|n_1|,|n_2|\}$ with $n = (n_1,n_2)$.
\end{widetext}
\label{lemma:exponential_lattice_tail}
\end{lemma}

\begin{proof}
We first clarify the notation. Overloading the notation, we use $|x|$ to represent the absolute value of $x$ when $x$ is a scalar and the standard Euclidean norm as $|x| = \sqrt{\langle x,x\rangle}$ when $x$ is a vector.
Given \(aN>R\) and \(|z|\le R\). If
\(\|n\|_\infty>N\), then \(|n|>N\), and by the reverse triangle inequality,
\begin{align}
    |an-z| \ge \left| a|n| - |z|\right| > aN - R.
\end{align}
Set $ L:=aN-R>0$ and decompose the omitted two-dimensional lattice points into shells with unit width
\begin{align}
    A_m^{(L)}(z) :=
    \{n\in\mathbb Z^2:L+m\le |an-z|<L+m+1\}.
\end{align}
Using $\#A^{(L)}_m(z)$ to denote the number of elements in $A_m^{(L)}(z)$, then
\begin{align}
\sum_{\|n\|_\infty > N} e^{-\mu |an-z|} \le \sum_{m=0}^\infty \#A^{(L)}_m(z) e^{-\mu (L+m)},
\end{align}
since, by definition, each $n \in A_m^{(L)}(z)$ satisfies $|an-z| \ge L+m$.

Now, by an area argument, each shell $A^{(L)}_m(z)$ has a size 
\begin{align}
\pi (L+m+1)^2 - \pi (L+m)^2 = 2\pi(L+m)+\pi,
\end{align}
while each lattice cell has a size $a^2$. Therefore, the number of points $\#A^{(L)}_m(z)$ is bounded as
\begin{align}
\#A^{(L)}_m(z) \le \frac{2\pi(L+m)+\pi}{a^2} \le C_a(L+m+1)
\end{align}
for some constant $C_a>0$ independent of $m$.

Therefore,
\begin{align}
    \sum_{\|n\|_\infty>N}e^{-\mu |an-z|}
    &\le
    C_a\sum_{m=0}^\infty (1+L+m)e^{-\mu(L+m)} \\
    &=
    C_a e^{-\mu L}
    \sum_{m=0}^\infty (1+L+m)e^{-\mu m} 
\end{align}
The summation involves a standard geometric series and its variant,
\begin{align}
\begin{aligned}
    \sum_{m=0}^\infty (1+L+m)e^{-\mu m}  &= (1+L)  \sum_{m=0}^\infty e^{-\mu m} + \sum_{m=0}^\infty me^{-\mu m} \\
    &= (1+L) \frac{1}{1-e^{-\mu}} + \frac{e^{-\mu}}{(1-e^{-\mu})^2}  \\
    &\le C_\mu(1+L)
\end{aligned}
\end{align}
for some $C_\mu >0$.
Together, we have
\begin{align}
    \sum_{\|n\|_\infty>N}e^{-\mu |an-z|}
    \le
     C_{\mu,a}(1+L)e^{-\mu L}
\end{align}
for some $C_{a,\mu} >0$.
We can further remove the linear term $(1+L)$ by choosing any $0<\mu'<\mu$. Then,
\begin{align}
    (1+L)e^{-\mu L} = (1+L)e^{-(\mu-\mu' )L} e^{-\mu' L}.
\end{align}
Now, since exponential decays faster than any linear growth,
\begin{align}
    \sup_{\ell\ge 0 } (1+\ell)e^{-(\mu-\mu' )\ell} := C_{\mu,\mu'}  < \infty.
\end{align}
Therefore, $(1+L)e^{-\mu L} \le C_{\mu\mu'} e^{-\mu'L}$ and the result follows.
\end{proof}

The following two lemmas connect the coefficient of the sparse coherent state explicitly to the sparsity and the separation between coherent states. We begin with a definition. 

\begin{definition}
    Define the coherent state Gramian matrix $G$ of size $s\times s$ associated with the state $|\psi\rangle = \sum_{i=1}^s a_i |\alpha_i\rangle$ by 
    \begin{align}
        G_{j\ell} := \langle \alpha_j|\alpha_\ell\rangle
    \end{align}
    for all $j,\ell \in [s]$.
\end{definition}
Note that $G$ is positive definite. $\lambda_{\min}(G)$ denotes the minimum eigenvalue of $G$.
\begin{lemma}
     Let $|\psi\rangle = \sum_{i=1}^s a_i |\alpha_i\rangle$. Denoting the coefficient as $a := (a_1,\dots,a_s)$, then 
     \begin{align}
     \|a\|_1^2 \le \frac{s}{\lambda_{\min}(G)}.
     \end{align}
     \label{lemma:a_one_norm_bound}
\end{lemma}
\begin{proof}
    Since the state is normalized, we have the following condition
    \begin{align}
    \langle \psi|\psi\rangle = \sum_{j=1}^s\sum_{\ell = 1}^s a_j^* a_\ell  \langle\alpha_j|  \alpha_\ell \rangle = \sum_{j,\ell=1}^s a_j^*a_\ell G_{j\ell} = 1.
    \end{align}
    Equivalently, in matrix notation,
    \begin{align}
    a^\dagger Ga = 1.
    \end{align}
    From the Rayleigh quotient bound,
    \begin{align}
    a^\dagger G a \ge \lambda_{\min}(G) \|a\|_2^2.
    \end{align}
    Taken together, $\|a\|_2^2 \le 1/\lambda_{\min}(G)$. 
    Finally, by Cauchy-Schwarz,
    \begin{align}
    \begin{aligned}
        \|a\|_1^2  &= \left( \sum_{j = 1}^s |a_j| \right)^2\\
        &\le \left(\sum_{j=1}^s 1 \right) \left(\sum_{j=1}^s |a_j|^2 \right) \\
        &= s \|a\|_2^2 \\
        &\le \frac{s}{\lambda_{\min}(G)}. 
    \end{aligned}
    \end{align}
\end{proof}

The final lemma relates the separation between coherent states in $|\psi\rangle$ to the smallest eigenvalue of the Gram matrix $G$.
\begin{lemma}
    Let $|\psi\rangle = \sum_{i=1}^s a_i |\alpha_i\rangle$ and $\Delta_\alpha = \min_{j\neq\ell}|\alpha_j - \alpha_\ell|$ for all $j,\ell \in [s]$. For any fixed $\kappa \in (0,1)$, if
    \begin{align}
    \Delta_\alpha \ge \sqrt{2\log \frac{s-1}{1-\kappa}},
    \end{align}
    then $\lambda_{\min}(G) \ge \kappa$.
    \label{lemma:Gram_min_bound}
\end{lemma}
\begin{proof}
    The coherent state Gramian matrix $G$ captures all overlaps between the coherent states that constitute the state $\ket{\psi}$. Recall that for coherent states,
    \begin{align}
    \begin{aligned}
        \langle \alpha_j|\alpha_\ell\rangle &=
        \exp\left(-\frac{1}{2}\left(|\alpha_j|^2+|\alpha_\ell|^2 -  2\alpha^*_j\alpha_\ell\right)\right) \\
        &= \exp\left(-\frac{1}{2}\left(|\alpha_j-\alpha_\ell|^2 -  2i \Im[\alpha^*_j\alpha_\ell]\right)\right).
    \end{aligned}
    \end{align}
    Then, for all $j,\ell \in [s]$ with $j\neq \ell$,
    \begin{align}
        |G_{j\ell}| = |\langle \alpha_j|\alpha_\ell\rangle| = \exp\left(-\frac{1}{2}|\alpha_j-\alpha_\ell|^2 \right) \leq e^{-\frac 1 2 \Delta_\alpha^2}.
        \label{eq:coherent_state_overlap_bound}
    \end{align}
    On the diagonal with $j\in[s]$, $|G_{jj}| = \langle \alpha_j |\alpha_j\rangle = 1$ for normalized coherent states.

    Next, by Gershgorin circle theorem, every eigenvalue of $G$ lies in one of the $s$ discs at center $G_{jj}$ with radius $\sum_{\ell \neq j} |G_{j\ell}|$. That is, if $\lambda$ is an eigenvalue of $G$, then
    \begin{align}
    |\lambda -  G_{jj}| \le \sum_{\ell \neq j} |G_{j\ell}|
    \end{align}
    for at least one $j\in[s]$.
    In particular 
    \begin{align}
    \lambda_{\min}(G) \geq 1 - \sum_{\ell \neq j} |G_{j\ell}|.
    \end{align}
    From Eq.~(\ref{eq:coherent_state_overlap_bound}), $|G_{j\ell}|$ is bounded for every pair $j\neq \ell$, then 
    \begin{align}
    \lambda_{\min} (G)\geq 1 - \sum_{\ell \neq j}  e^{-\frac 1 2 \Delta_\alpha^2}  \ge 1 - (s-1)e^{-\frac 1 2 \Delta_\alpha^2}.
    \end{align}
    We set the right-hand side to be larger than $\kappa$ to solve for $\lambda_{\min}(G) \ge \kappa$. Note that this is nontrivial only for $\kappa\in(0,1)$. Rearranging the terms, we see that if 
    \begin{align}
    \Delta_\alpha \ge 
    \sqrt{
        2\log\frac{s-1}{1-\kappa}
    },
    \end{align}
    then 
    \begin{align}
    e^{-\frac 1 2\Delta_\alpha^2}
    \le
    \frac{1-\kappa}{s-1}.
    \end{align}
    Therefore, $1- (s-1) e^{-\frac 1 2 \Delta_\alpha^2} \ge 1-(1-\kappa) \ge \kappa$, and hence,
    \begin{equation}
    \lambda_{\min}(G) \ge \kappa. 
    \end{equation}
\end{proof}

\subsubsection{Proposition~\ref{proposition:Gabor_truncation_bound}: Truncation error with Gabor frame}
\label{sec:proof_gabor_truncation}

Having introduced the relevant concepts of frame theory for Gabor frames and STFT, we are now in a position to bound the approximation error using a finite number of Gabor atoms. The lemmas needed can be referenced in Sec.~\ref{sec:lemmas_for_Gabor_trunc}.

\begin{proposition}[Truncation error bound]
    Let $W_\psi \in L^2(\mathbb R^2)$ be the Wigner function of the state $|\psi\rangle = \sum_{i=1}^s a_i |\alpha_i\rangle$. Furthermore, assume that the magnitude of the coherent states lies in a bounded region in phase space, i.e., $|\alpha_j| \le R$ for all $j \in [s]$ and the separation between any two coherent states is at least $\Delta_\alpha := \min_{j\neq\ell}|\alpha_j - \alpha_\ell| $ such that
    \begin{align}
    \Delta_\alpha \ge \sqrt{2\log \frac{s-1}{1-\kappa}}
    \end{align}
    for some $\kappa \in (0,1)$. Let $\mathcal{G}(g,a,b)$ be a Gabor frame with a two-dimensional Gaussian window function $g(x) = e^{-\pi |x|^2} $ and the canonical dual frame be $\mathcal{G}(\gamma,a,b)$. We define the 
    truncation lattice as
    \begin{align}
    \Lambda_{N,K} =
    \left\{
        (n,k)\in\mathbb Z^2\times\mathbb Z^2:
        \|n\|_\infty\le N,\ \|k\|_\infty\le K
    \right\}
    \end{align}
    and the associated canonical approximation of $W_\psi$ by
    \begin{align}
    W_{N,K} = \sum_{(n,k) \in \Lambda_{N,K} } \langle W_\psi, M_{bn}T_{ak} \gamma \rangle M_{bn}T_{ak} g.
    \label{eq:gabor_trunc_rep}
    \end{align}
    Then, there exist constants $C,p_z,p_\omega >0$ such that 
    \begin{align}
    \|W_\psi-W_{N,K}\|_\infty \le C \frac{s}{\kappa}  \left[
    e^{-p_z(aK-\sqrt{2}R)}
    +
    e^{-p_\omega(bN-\frac{\sqrt{2}R}{\pi})}
    \right].
    \end{align}
    \label{proposition:Gabor_truncation_bound}
\end{proposition}

\begin{proof}
Recall that the Gabor frame is given by $\mathcal G(g,a,b) = \{M_{bn}T_{ak}g:n,k\in\mathbb Z^2\}$ with the canonical dual frame $\mathcal G(\gamma,a,b) = \{M_{bn}T_{ak}\gamma:n,k\in\mathbb Z^2\}$. This provides the canonical Gabor expansion for any $W_\psi\in L^2(\mathbb R^2)$
\begin{align}
W_\psi = \sum_{n,k \in \mathbb Z^2} c_{n,k} M_{bn}T_{ak} g,
\end{align}
where $c_{n,k}$ is calculated as the overlap between $W$ and the dual frame
\begin{align}
c_{n,k} = \langle W_\psi,M_{bn}T_{ak}\gamma\rangle.
\end{align}
In general, the summation carries infinitely many terms. However, our approximation considers a finite lattice centered at the origin, given by $\Lambda_{N,K}$.
Then, the truncated approximation is
\begin{align}
W_{N,K} = \sum_{(n,k) \in \Lambda_{N,K}} c_{n,k} M_{bn}T_{ak} g.
\end{align}
The goal is to bound the approximation error uniformly. By definition,
\begin{align}
    \|W-W_{N,K} \|_\infty &= \left\|\sum_{(n,k) \notin \Lambda_{N,K}} c_{n,k} M_{bn}T_{ak} g \right\|_\infty \\
    &\le  \sum_{(n,k) \notin \Lambda_{N,K}} |c_{n,k}| \| M_{bn}T_{ak} g\|_\infty \\
    &= \sum_{(n,k) \notin \Lambda_{N,K}} |c_{n,k}| \|g\|_\infty \\
    &= \sum_{(n,k) \notin \Lambda_{N,K}} |c_{n,k}| \label{eq:coefficient_bound},
\end{align}
where we use the triangle inequality on the second line. The second equality follows by noting that $M_{bn}T_{ak}$ is unitary and the unitary invariance of the infinity norm, and the last equality from $\|g\|_\infty = 1$. 

In order to bound the Gabor coefficients $c_{n,k}$, we now invoke our assumption that the target state is a superposition of coherent states. 

Before continuing, we first set up some notations for working in real phase-space coordinates $ r=(x,p)\in\mathbb R^2$. Specifically, let
\(
    |\psi\rangle=\sum_{j=1}^s a_j|\alpha_j\rangle
\)
be an \(s\)-sparse coherent-state superposition. We denote by $r_j = (x_j,p_j)\in\mathbb R^2$ the real phase-space center associated with \(|\alpha_j\rangle\), given by $\alpha_j = (x_j+ip_j) /\sqrt{2}.$
Then the Wigner function of $|\psi\rangle$ is
\begin{align}
W_{\psi} = \sum_{j,\ell=1}^s a_j a_\ell^* W_{j\ell},
\end{align}
where we define $W_{j\ell} = W_{|\alpha_j\rangle\langle\alpha_\ell|}$. From Eq.~(\ref{eq:coherent_operator_Wigner_main}), we have the explicit form of the Wigner function of the operator $|\alpha_j\rangle\langle\alpha_\ell|$. With the notation here, we define
\begin{align}
z_{j\ell} := \frac{r_j + r_\ell}{2},
\end{align}
and the symplectic form $\sigma(u,v):= u_qv_p - u_pv_q$ for $u = (u_q,u_p)$ and $v=(v_q,v_p)$. Then, it can be shown that
\begin{equation}\label{WignerFuncCrossTerm}
W_{j\ell} (r)= \frac{2}{\pi} e^{-|r-z_{j\ell}|^2} \exp \left(- i J(r_\ell-r_j) \cdot r + \frac i 2  \sigma(r_j,r_\ell) \right).
\end{equation}
where the symplectic matrix $J$ is 
\begin{align}
    J :=
    \begin{pmatrix}
        0 & 1\\
        -1 & 0
    \end{pmatrix}.
\end{align}
Therefore by defining $\theta_{j\ell} := \frac{1}{2}\sigma(r_\ell,r_j)$ and
\begin{align}
\omega_{j\ell} := \frac{-1}{2\pi} J(r_\ell - r_j),
\end{align}
we can express $W_{j\ell}$ (\ref{WignerFuncCrossTerm}) in terms of the translation and modulation operator as
\begin{align}
    W_{j\ell}(r)
    =
    e^{i\theta_{j\ell}}
    M_{\omega_{j\ell}}T_{z_{j\ell}}g_0(r),
\end{align}
where
\begin{align}
    g_0(r)=\frac2\pi e^{-|r|^2}
\end{align}
denotes the Gaussian Wigner function of a vacuum state.

Therefore, using the fact that $\ket{\psi}\bra{\psi}=\sum_{j,\ell=1}^s a_j a^*_\ell \ket{\alpha_j}\bra{\alpha_\ell}$, we can write the canonical Gabor coefficients of \(W_\psi\) explicitly as
\begin{align}
    c_{n,k}
    =
    \left\langle
        W_\psi,
        M_{bn}T_{ak}\gamma
    \right\rangle  =\sum_{j,\ell=1}^s
    a_j^* a_\ell \left\langle W_{j\ell}, M_{bn}T_{ak}\gamma
    \right\rangle 
\end{align}
For each coherent state operator $|\alpha_j\rangle \langle \alpha_\ell|$, define
\begin{align}
    d^{j\ell}_{n,k}
    =
    \left\langle
        W_{j\ell},
        M_{bn}T_{ak}\gamma
    \right\rangle .
\end{align}
By linearity,
\begin{align}
    c_{n,k}
    =
    \sum_{j,\ell=1}^s
    a_j^* a_\ell d^{j\ell}_{n,k}.
\end{align}
Thus, Eq.~\eqref{eq:coefficient_bound} reduces to
\begin{align}
    \sum_{(n,k) \notin \Lambda_{N,K}} |c_{n,k}| &\le  \sum_{j,\ell=1}^s
    |a_j| |a_\ell|  \sum_{(n,k) \notin \Lambda_{N,K}} |d^{j\ell}_{n,k}| 
    \label{eq:c_one_norm}
\end{align}
This provides us with a concrete form of the Gabor coefficients in relation to the target state $|\psi\rangle$.

To proceed further, we recall the commutation relation between the time-frequency shifts in Eq.~\eqref{eq:time-frequency_comm},
\begin{align}
\begin{aligned}
    |d^{jl}_{n,k}| &= |\langle 
    M_{\omega_{j\ell}}T_{z_{j\ell}}g_0, M_{bn}T_{ak}\gamma\rangle | \\
    &= |\langle g_0, T_{-z_{j\ell}}M_{-\omega_{j\ell}}M_{bn}T_{ak}\gamma\rangle| \\
    &= |\langle g_0, T_{-z_{j\ell}}M_{bn-\omega_{j\ell}}T_{ak}\gamma\rangle| \\
    &= |\langle g_0, M_{bn-\omega_{j\ell}}T_{ak-z_{j\ell}}\gamma\rangle|.
\end{aligned}
\end{align}
Compared with Eq.~\eqref{eq:STFT_inner_product}, we observe that this is exactly the short-time Fourier transform of $g_0$ with the window function $\gamma$, evaluated at the point $(ak-z_{j\ell},bn-\omega_{j\ell})$, i.e.
\begin{align}
    |d^{j\ell}_{n,k}|
    =
    |\mathcal{V}_\gamma g_0(ak-z_{j\ell},bn-\omega_{j\ell})|.
\end{align}

The dual window cannot be found analytically as it involves the inversion of the frame operator, which is an infinite-dimensional matrix. However, it turns out that when $\gamma$ and $g_0$ are both sufficiently smooth and decay exponentially to infinity, the STFT also decays exponentially, although the exact constant cannot be obtained.
Formally, we invoke the decaying properties of the STFT given by Lemma~\ref{lemma:dual_STFT_decay}, which states that for ``nice" functions $\gamma$ and $g_0$, there exist some $C, p_z,p_\omega >0$ such that
\begin{align}
|\mathcal{V}_\gamma g_0 (ak-z_{j\ell},bn-\omega_{j\ell})| \leq C e^{-p_z|ak-z_{j\ell}|} e^{-p_\omega |bn-\omega_{j\ell}|}.
\end{align}

Continuing from Eq.~(\ref{eq:c_one_norm}) we first bound the summation involving the truncated lattice,
\begin{align}
\sum_{(n,k) \notin \Lambda_{N,K}} |d^{j\ell}_{n,k}|  \le C\sum_{(n,k) \notin \Lambda_{N,K}} e^{-p_z|ak-z_{j\ell}|} e^{-p_\omega |bn-\omega_{j\ell}|}.
\end{align}
Recalling the definition of the truncated lattice and using the union bound, we separate the sum as
\begin{widetext}
    \begin{align}
    \begin{aligned}
    &\sum_{(n,k) \notin \Lambda_{N,K}} \hspace{-1em }e^{-p_z|ak-z_{j\ell}|} e^{-p_\omega |bn-\omega_{j\ell}|} \\
    &\le \sum_{n \in \mathbb Z^2} \sum_{\|k\|_\infty >K } e^{-p_z|ak-z_{j\ell}|} e^{-p_\omega |bn-\omega_{j\ell}|}
    +
    \sum_{\|n\|_\infty>N}\sum_{k\in\mathbb Z^2}
    e^{-p_z |ak-z_{j\ell}|}
    e^{-p_\omega |bn-\omega_{j\ell}|}. \\
    &= \left( \sum_{n \in \mathbb Z^2} e^{-p_\omega |bn-\omega_{j\ell}|}\right) \hspace{-0.5em} \left(\sum_{\|k\|_\infty >K } e^{-p_z|ak-z_{j\ell}|} \right) 
    \!+\!
    \left( \sum_{\|n\|_\infty>N} e^{-p_\omega |bn-\omega_{j\ell}|} \right) \hspace{-0.5em}\left( \sum_{k\in\mathbb Z^2} 
    e^{-p_z |ak-z_{j\ell}|}
    \right).  
    \end{aligned}
    \end{align}
\end{widetext}

The exponential tail sum can be bounded by Lemma~\ref{lemma:exponential_lattice_tail}, which asserts that there exist $D>0$ and $0<p_z'<p_z$ such that
\begin{align}
\sum_{\|k\|_\infty >K } e^{-p_z|ak-z_{j\ell}|} \le  D e^{-p_z'(aK - |z_{j\ell}|)}
\end{align}
This trivially translates to the infinite sum over $k \in \mathbb Z^2$ as some finite constant $D'$, that is,
\begin{align}
\sum_{k\in\mathbb Z^2} 
    e^{-p_z |ak-z_{j\ell}|} = D' <\infty.
\end{align}
The same Lemma applies to the sum over $\|n\|_\infty >N$ and therefore there exists some $C',p_z',p_\omega'>0$ such that
\begin{align}
    \sum_{(n,k) \notin \Lambda_{N,K}} |d^{j\ell}_{n,k}| \le C'(e^{-p_z'(aK- |z_{j\ell}|)}  + e^{-p_\omega'(bN- |\omega_{j\ell}|)})
\end{align}

From our assumption that $|\alpha_j| \le R$ for all $j \in [s]$, we have $|r_j| = \sqrt2 |\alpha_j| \le R $. Then,
\begin{align}
    \max_{j,\ell}|z_{j\ell}| = \max_{j,\ell} \left|\frac{r_j+r_\ell}{2} \right| \le\sqrt{2} R,
\end{align}
and
\begin{align}
    \max_{j,\ell}|\omega_{j\ell}| = \max_{j,\ell} \frac{1}{2\pi} |r_j - r_\ell | \le \frac{\sqrt{2}R}{\pi}. 
\end{align}
If
\begin{align}
    aK > \sqrt{2}R,
    \qquad
    bN > \frac{\sqrt{2}R}{\pi},
\end{align}
Then the tail bound simplifies to
\begin{align}
    \sum_{(n,k) \notin \Lambda_{N,K}} |d^{j\ell}_{n,k}| \le C\left[
        e^{-p_z'(aK-\sqrt{2}R)}
        +
        e^{-p_\omega'(bN-\frac{\sqrt{2}R)}{\pi}}
    \right].
\end{align}
Therefore, we have
\begin{widetext}
    \begin{align}
    \begin{aligned}
        \sum_{(n,k)\notin\Lambda_{N,K}} |c_{n,k}|
        &\le C'   \sum_{j,\ell=1}^s
        |a_j| |a_\ell| \left[
        e^{-p_z'(aK-\sqrt{2}R)}
        +
        e^{-p_\omega'(bN-\frac{\sqrt{2}R)}{\pi})}
        \right]  \\
        &\le
        C'_{p_z',p_\omega',a,b}
        \|a\|_1^2
        \left[
        e^{-p_z'(aK-\sqrt{2}R)}
        +
        e^{-p_\omega'(bN-\frac{\sqrt{2}R)}{\pi})}
        \right].
    \end{aligned}
    \end{align}
\end{widetext}

From Lemma~\ref{lemma:a_one_norm_bound}, we have $\|a\|_1^2 \le s/{\lambda_{\min}(G)}$ where $G$ is the coherent state Gram matrix corresponding to $|\psi\rangle$. 
Following our assumption on the separation of the coherent states, namely,
\begin{align}
\Delta_\alpha \ge \sqrt{2\log \frac{s-1}{1-\kappa}},
\end{align}
where $\kappa \in (0,1)$, Lemma~\ref{lemma:Gram_min_bound} guarantees that $\lambda_{\min} (G) \ge \kappa$. Therefore, $\|a\|_1^2 \le \frac{s}{\kappa}$.

Finally, putting everything together and relabeling the constants, there exists some $C,p_z,p_\omega > 0$ such that the uniform error of the canonical Gabor expansion with the lattice $\Lambda_{N,K}$ satisfies
\begin{align}
    \|W-W_{N,K}\|_\infty \le C \frac{s}{\kappa}  \left[
        e^{-p_z(aK-\sqrt{2}R)}
        +
        e^{-p_\omega(bN-\frac{\sqrt{2}R)}{\pi}}
    \right]. 
    \label{eq:gabor_truncation_bound}
\end{align}
\end{proof}

\subsubsection{Corollaries of Proposition~\ref{proposition:Gabor_truncation_bound}}
\label{sec:corollaries}
Two implications resulting from Proposition~\ref{proposition:Gabor_truncation_bound} are crucial to the generalization bound in Theorem~\ref{theorem:Gabor_ML_app}. Specifically, they determine the value of $t$ and the truncation parameter for $N$ and $K$ such that the empirical error can be made arbitrarily small.

\begin{corollary}
    Define the hypothesis class $\mathcal{H}_{N,K}$ as 
    \begin{align}
        \mathcal H_{N,K}
    =
    \left\{
        h(r)
        = \hspace{-1em}
        \sum_{(n,k)\in\Lambda_{N,K}}
        \eta_{n,k}M_{bn}T_{ak}g(r)
        :
        \|\bm\eta\|_1\le t'
    \right\}.
    \end{align}
    Then, with the choice $t' = \mathcal O(\frac{s}{\kappa})$, $W_{N,K} \in \mathcal{H}_{N,K}$ where $W_{N,K}$ is given by Eq.~\eqref{eq:gabor_trunc_rep}.
    \label{corollary:t_bound}
\end{corollary}
\begin{proof}
    For $W_{N,K}$ to be in the hypothesis class, we require $t$ to be greater than the norm of the coefficient of $W_{N,K}$ \. From Eq.~\eqref{eq:c_one_norm}
    \begin{align}
    \begin{aligned}
        \|c_{N,K}\|_1 &= \sum_{(n,k)\in \Lambda_{N,K}} |c_{n,k}|  \\
        &\le \sum_{j ,\ell=1}^s |a_j||a_\ell| \sum_{(n,k)\in \Lambda_{N,K}} |d^{j\ell}_{n,k}|.
    \end{aligned}
    \end{align}
    Let $ B = \sum_{(n,k) \in \Lambda_{N,K}} |d^{j\ell}_{n,k}|$, which is a constant by Proposition~\ref{proposition:Gabor_truncation_bound}. Then,
    \begin{align}
        \|c_{N,K}\|_1 \le B\|a\|_1^2  \le B\frac{s}{\kappa}.
    \end{align}
    Therefore $t' = \mathcal O(\frac{s}{\kappa})$ ensures that the canonical expansion lies within the hypothesis class.
\end{proof}

\begin{corollary}
    To obtain an error of
    \begin{align}
    \|W_\psi-W_{N,K}\|_\infty \le \sqrt \frac{ \varepsilon}{8},
    \end{align}
    it is sufficient for the lattice truncation parameters to satisfy
    \begin{align}
        K
        =
        O\left(
            R
            +
            \frac1{p_z}
            \log\left(
                \frac{s}{\kappa\sqrt{\varepsilon}}
            \right)
        \right),
    \end{align}
    and
    \begin{align}
        N
        =
        O\left(
            R
            +
            \frac1{p_\omega}
            \log\left(
                \frac{s}{\kappa\sqrt{\varepsilon}}
            \right)
        \right).
    \end{align}
    \label{corollary:NK_bound}
\end{corollary}
\begin{proof}
    From Eq.~(\ref{eq:gabor_truncation_bound}), we separately require each term to be at most half of the target, 
\begin{align}
    C
    \frac{s}{\kappa}
    e^{-p_z(aK-\sqrt{2}R)}
    \le
    \frac{1}{4}\sqrt{\frac{\varepsilon}{2}} ,
\end{align}
and
\begin{align}
    C
    \frac{s}{\kappa}
    e^{-p_\omega(bN-\frac1 \pi \sqrt 2R)}
    \le
    \frac{1}{4}\sqrt{\frac{\varepsilon}{2}}.
\end{align}
Equivalently, it suffices that
\begin{align}
    aK-\sqrt{2} R
    \ge
    \frac1{p_z}
    \log\left(
        \frac{4\sqrt{2}Cs}{\kappa\sqrt{\varepsilon}}
    \right),
\end{align}
and
\begin{align}
    bN- \frac{2}{\pi}\sqrt{2} R
    \ge
    \frac1{p_\omega}
    \log\left(
        \frac{4\sqrt{2}Cs}{\kappa\sqrt{\varepsilon}}
    \right).
\end{align}
Finally, we have
\begin{align}
    K
    =
    O\left(
        R
        +
        \frac1{p_z}
        \log\left(
            \frac{s}{\kappa\sqrt{\varepsilon}}
        \right)
    \right),
\end{align}
and
\begin{align}
    N
    =
    O\left(
        R
        +
        \frac1{p_\omega}
        \log\left(
            \frac{s}{\kappa\sqrt{\varepsilon}}
        \right)
    \right). 
\end{align}
\end{proof}

\subsubsection{Proof of Theorem~\ref{theorem:Gabor_ML}}
\label{sec:Gabor_ML_proof}
In this section, we provide a rigorous guarantee for learning states with sparse coherent state support using Gabor frame. We first motivate the choice of this feature map by inspecting the Wigner function state. The setup and the main theorem are then introduced. Finally, we prove Theorem~\ref{theorem:Gabor_ML} with the help of Corollary~\ref{corollary:NK_bound} and Corollary~\ref{corollary:t_bound}.

Recall that coherent states form an overcomplete basis for a CV state space with the resolution of the identity $\frac{1}{\pi}\int|\alpha\rangle \langle \alpha|d^2\alpha = \mathbbm{1}$. In other words, coherent states form a continuous tight frame since the frame operator is proportional to the identity. This provides another convenient way to perform expansion:
\begin{align}
\rho = \frac{1}{\pi^2}\int|\gamma\rangle \langle \gamma| \rho |\beta\rangle \langle \beta | d^2\gamma d^2\beta =\frac{1}{\pi^2}\int \rho_{\gamma\beta} |\gamma\rangle \langle \beta | d^2\gamma d^2\beta
\end{align}
where $\rho_{\gamma\beta} = \langle \gamma| \rho |\beta\rangle$. Then the Wigner function can be expressed as
\begin{align}
W_\rho(\alpha) = \frac{1}{\pi^2}\int \rho_{\gamma\beta} W_{|\gamma\rangle \langle \beta | }d^2\gamma d^2\beta.
\end{align}
Note that since $\langle\gamma|\beta\rangle\neq 0$, this decomposition, i.e. the coefficients $\rho_{\gamma \beta}$ is not unique. Using the definition of the Wigner function $W_A(\alpha) = \frac{2}{\pi} \tr(D(\alpha) (-1)^{\hat n} D(-\alpha)A)$, we have
\begin{align}
\begin{aligned}
    \quad W_{|\gamma\rangle \langle\beta|} &= \tr\left( |\gamma\rangle \langle \beta| D(\alpha) (-1)^{\hat n} D(-\alpha) \right) \\
    &= \tr\left( |0\rangle \langle 0| D(-\beta) D(\alpha) (-1)^{\hat n} D(-\alpha) D(\gamma) \right) \\
    &= C_1\tr\left( |0\rangle \langle 0| D(\alpha -\beta) (-1)^{\hat n} D(\gamma-\alpha)  \right) \\
    &= C_1 \tr \left( |0\rangle \langle 0| D(\alpha -\beta) D(\alpha-\gamma)  (-1)^{\hat n} \right) \\
    &= C_1C_2\tr \left( |0\rangle \langle 0| D(2\alpha -(\beta+\gamma))  (-1)^{\hat n} \right) \\
    &= C_1C_2 \langle -2\alpha +(\beta+\gamma)| 0\rangle \\
    &= C_1C_2 \exp(-|2\alpha - (\gamma +\beta)|^2/2)
\end{aligned}
\end{align}
where $C_1=e^{(\beta^*\alpha-\beta\alpha^*)/2} e^{(\alpha^*\gamma-\alpha\gamma^*)/2} $ and $C_2 = \exp\{ \frac{1}{2}((\alpha-\beta)(\alpha^*-\gamma^*)-(\alpha^*-\beta^*)(\alpha-\gamma))\}$. In the derivation, we use the relations $D(\alpha)D(\beta) = e^{(\alpha\beta^*-\alpha^*\beta)/2}D(\alpha+\beta)$.
The fourth equality uses the commutation $(-1)^{\hat n} D(\alpha) = D(-\alpha) (-1)^{\hat n}$ and the sixth follows from $\langle 0|D(\alpha) = \langle -\alpha|$ $(-1)^{\hat n}|0\rangle = |0\rangle$. The final equality comes from $\langle \alpha|0\rangle = e^{-|\alpha|^2/2}$.
Simplifying the phase factors $C_1C_2$, we arrive at
\begin{widetext}
    \begin{align}
        W_{|\gamma\rangle \langle \beta|} (\alpha) = \frac{2}{\pi} \exp\left(  - 2\left|  \alpha - \frac{\gamma+\beta}{2} \right|^2 \right) \exp\left( 2i \Im \left[  (\beta^* - \gamma^*)\alpha  + \frac{1}{2}\beta\gamma^* \right]\right).
        \label{eq:coherent_operator_Wigner}
    \end{align}
\end{widetext}

Introducing the center $z = \frac{\gamma+\beta}{2}$ and the separation $\omega = \frac{1}{\pi}(\beta - \gamma)$, the Wigner function becomes
\begin{align}
    W_{|\gamma\rangle \langle \beta|} (\alpha) = \frac{2}{\pi} e^{-2|\alpha-z|^2} \exp \left(  2\pi i \Im \left[\omega^*(\alpha - \frac z2) \right] \right)
    \label{Wigner_coherent_states}
\end{align}
where
\begin{align}
\gamma = z - \frac{\pi\omega}{2}, \quad \beta = z+\frac{\pi\omega}{2}.
\end{align}

Eq.~\eqref{Wigner_coherent_states} is composed of a Gaussian function centered at $z$, modulated by a phase $e^{2\pi i \Im[\omega^* \alpha]}$.
However, it does not furnish a suitable feature map for the regression model as the parameters $\gamma$ and $\beta$ are continuous complex values. This implies that infinitely many Gaussians and phase factors are needed. Naturally, we ask whether a discrete set of Gaussian functions modulated by a Fourier series can still represent any state. 

To this end, we consider the appropriate discrete counterpart by covering the phase space with Gaussians situated on a lattice modulated by a Fourier series. 
This motivates the choice of Fourier series with Gaussian envelopes as the appropriate basis functions for learning states with sparse coherent state support, which is nothing but the Gabor frame introduced in Sec.~\ref{sec:Gabor_STFT}.

% Formally, we use the real parameters $a,b >0$ to determine the size of the lattice given by $ \Lambda = a\mathbb Z^d  \times b \mathbb Z^d$ where $\mathbb Z^d$ is the $d$-fold cartesian product of $\mathbb Z$.
% For a window function $g\in  L^2(\mathbb{R}^d)$ and $n,k \in \mathbb Z^d$,
% we define the translation and modulation operators as \((T_{ak}g)(x) = g(x-ka)\) and \((M_{bn} g)(x) = \exp(2\pi i \langle bn, x\rangle ) g(x)\), respectively. 
% Then, the Gabor system is given by $\mathcal{G}(g,a,b) = \{M_{bn}T_{ak}g : \: n,k \in \mathbb{Z}^d\}$ with $g \in L^2(\mathbb{R}^d)$. For our purposes regarding the Wigner function of a single mode CV state, $d=2$ suffices and we work exclusively with the Gaussian window function, i.e. $g(x) = e^{-\pi |x|^2}$.

In fact, von Neumann considered the problem for $\psi \in L^2(\mathbb{R})$, using a square lattice in an attempt to find the most accurate simultaneous measurement of the position and momentum operator. He conjectured, which turned out to be correct, that $\mathcal{G}(g,1,1)$ is a complete system in $L^2(\mathbb{R})$, which is only a Riesz basis, but not a frame \cite{grochenig_foundations_2001,heil_basis_2011,perelomov_generalized_1986}.

% Instead, when $ab <1$, $\mathcal{G}(g,a,b)$ forms a frame for $L^2(\mathbb{R}^d)$, known as Gabor frame. See Sec.~\ref{sec:frame_theory} for a review of frame theory, which is the mathematical language for dealing with an overcomplete basis. Sec.~\ref{sec:Gabor_STFT} provides more properties of the Gabor frame and the Short-time Fourier transform.

Returning to our setting, the Gabor frame allows the expansion of the Wigner function (or more generally, any $L^2$-integrable function) as 
\begin{align}
W(r) = \sum_{n,k \in \mathbb{Z}^2} c_{nk} M_{bn}T_{ak}g(r),
\end{align}
where $r = (x,p)$. In general, for a frame, this decomposition is not unique and the coefficients $c_{nk} \neq \langle W, \Phi_{n,k}\rangle$. Intuitively, the Gabor frame covers the space with the lattice $\Lambda = a\mathbb Z^2 \times b \mathbb Z^2$. When the lattice is sufficiently fine, it is able to capture all information of the function.

In terms of our supervised learning setting, we have $\mathcal{X} =  \{\alpha \in \mathbb{R}^2 : |\alpha|\leq R\}$, equivalently parametrized in polar coordinates by $(r,\phi) \in [0,R]\times [0,2\pi)$ and the output space \(\mathcal{Y} = [-\frac{2}{\pi},\frac{2}{\pi}]\). For our learning problem, we truncate to a lattice of finite size \begin{align} 
\Lambda_{N,K}
=
\left\{
    (n,k)\in\mathbb Z^2\times\mathbb Z^2:
    \|n\|_\infty\le N,\ \|k\|_\infty\le K
\right\}.
\end{align}
This defines the feature map $\bm \Phi$. Explicitly, let $\bomega \in \Omega := \left\{ (n,k,\pm 1) : -N \le n_1,n_2 \le N, -K \le k_1,k_2\le K\right\}$ where
\begin{widetext}
    \begin{align}
        \Phi_{n,k,1} (x,p) = \exp\left( -\pi [(x-ak_1)^2+(p-ak_2)^2] \right) \cos( 2\pi ( bn_1 x + bn_2 p)) \\
        \Phi_{n,k,-1} (x,p) = \exp\left( -\pi [(x-ak_1)^2+ (p-ak_2)^2] \right) \sin( 2\pi ( bn_1 x + bn_2 p)),
    \end{align}
\end{widetext}
where we use real basis functions for the regression to learn real-valued Wigner functions. The feature dimension is $d_f = |\Omega| = (2N+1)^2(2K+1)^2$ where the linearly dependent features due to the parity of the trigonometric functions are removed, e.g. $\sin (n_1x+n_2p) = -\sin((-n_1)x + (-n_2)p)$.
The linear span of these functions with a one-norm restriction constitutes the hypothesis class
\begin{align}
    \mathcal H
    =
    \left\{
        h(r)
        =
        \sum_{\bomega \in \Omega}
        \eta_{\bomega }\Phi_{\bomega}
        :
        \|\bm{\eta}\|_1\le t
    \right\}.
\end{align}

In the proof below, however, it is easier to work with complex-valued functions of the form
\begin{align}
    \mathcal H
    =
    \left\{
        h(r)
        = \!\!
        \sum_{(n,k)\in \Lambda_{N,K}}
        \eta_{n,k}M_{bn}T_{ak} g(r)
        :
        \|\bm{\eta}\|_1\le t'
    \right\}.
\end{align}
The two coefficients are related by $|\eta_{n,k}| \le |\Re[\eta_{n,k}]|+|\Im[\eta_{n,k}]| \le \sqrt{2} |\eta_{n,k}|$. Thus, it is safe to take $t$ the same order as $t'$.

Finally, given the training data $\dataset$, the regression model $\widehat{W}$ is obtained by solving the following:
\begin{align}
    \min_{\bm{\eta}} \sum_{i=1}^{\mathcal N} \left| \langle \bm{\eta},\mathbf{\Phi} (\alpha_i)\rangle - y_i \right|^2 \; \text{, subject to } \; \|\bm{\eta}\|_1 \leq t.
    \label{eq:Lasso_Gabor}
\end{align}

\begin{theorem}[Restatement of Theorem~\ref{theorem:Gabor_ML}]
    Let the target state be $|\psi\rangle = \sum_{i=1}^s a_i |\alpha_i\rangle$ with $|\alpha_j|\le R$ and $|\alpha_j - \alpha_\ell|\ge \Delta_\alpha$ for all $j,\ell \in [s]$ and $j \neq \ell$. Let $\kappa = 1 - (s-1)e^{-\frac 1 2 \Delta_\alpha^2}$. Let $\widehat{W}$ be the surrogate obtained by solving Eq.~\eqref{eq:Lasso_Gabor}. 
    For any $\epsilon > 0$ and $1/2 >\delta > 0$, assume that all training data in \(\mathcal{T}\) satisfies \(|W(\alpha_i) - y_i| \leq \sqrt{\epsilon/4}\) with probability at least $1-\delta'$ over the measurement randomness where $\delta'=\delta/\mathcal{N}$. Then a training data of size
    \begin{align}
    \mathcal N = O\left( \frac{s^4}{\epsilon^2 \kappa^4} \log \frac{4(R+ \log (s/\kappa\sqrt \epsilon))^8}{\sqrt\delta} \right)
    \end{align}
    suffices to achieve $\mathsf R(\widehat{W}) \le \epsilon $ for an arbitrary distribution $\mathbb D$ with probability at least $1-2\delta$.
    \label{theorem:Gabor_ML_app}
\end{theorem}

\begin{proof}
    Again, we employ Lemma~\ref{lemma:Lasso} to derive the generalization error. 
    Thus, we begin by determining the suitable values for $t,r_\infty$ and $M$.
    From Corollary~\ref{corollary:t_bound}, we take $t = \mathcal O(s/\kappa)$ to ensure $W_{N,K} \in \mathcal{H}$ where $W_{N,K}$ is given by Eq.~\eqref{eq:gabor_trunc_rep}. Next, since
    \begin{align}
    \|\bm{\Phi}(\alpha)\|_\infty =  \| (M_{bn}T_{ak}g) (\alpha)\|_\infty =  \|g\|_\infty = 1,
    \end{align}
    where the second equality follows from the unitary invariance of the norm, we take $r_\infty = 1$. Lastly,  for all $h \in \mathcal{H}$,
    \begin{align}
    \begin{aligned}
        |h(\alpha) - y| &\leq |h(\alpha)| + |y| \\
        &\leq  |\langle \bm{\eta}, \mathbf\Phi (\alpha)\rangle | + |W(\alpha) | + \sqrt \frac{\epsilon}{4} \\
        &\leq \| \bm{\eta} \|_1 \|\mathbf\Phi (\alpha)\|_\infty + \frac 2\pi +\sqrt \frac{\epsilon}{4} \\
        &\leq t + \frac 2 \pi + \sqrt \frac{\epsilon}{4}.
    \end{aligned}
    \end{align}
    Therefore, we can take $M = \mathcal O(t) = \mathcal O(s/\kappa$). Combining everything, we have
    \begin{align}
        \mathsf R(\widehat{W}) \le  \widehat{\mathsf{R}}_{\mathcal T} (\widehat{W}) +  O\! \left( \frac{t^2}{\sqrt{\mathcal N}} \left( \sqrt{2\log (2d_f)} + \sqrt{\frac{1}{2} \log \frac{1}{\delta}} \right) \!\right).
        \label{eq:Gabor_generalization_eq}
    \end{align}

    Next, we need to determine the truncation values of $N,K$ such that the empirical error can be bounded as 
    \begin{align}
    \widehat{\mathsf{R}}_{\mathcal T} (\widehat{W})
    \le
    \frac{3\epsilon}{4}.
    \end{align}
    By definition, the model returned by the Lasso algorithm satisfies 
    \begin{align}
    \widehat{W} = \arg\min_{h\in\mathcal H}
    \frac1{\mathcal N}\sum_{i=1}^{\mathcal N} |h(\alpha_i)-y_i|^2.
    \end{align}
    Since the choice of $t = \mathcal O(s/\kappa)$ ensures that the canonical truncation $W_{N,K}$ lies within the hypothesis class (Corollary~\ref{corollary:t_bound}), we can use it to upper bound the training error of $\widehat W$ as
    \begin{align}
    \begin{aligned}
        \frac1{\mathcal N}\sum_{i=1}^{\mathcal N}
        |\widehat W(\alpha_i)-y_i|^2
        &\le
        \frac1{\mathcal N}\sum_{i=1}^{\mathcal N}
        |W_{N,K}(\alpha_i)-y_i|^2.
    \end{aligned}
    \end{align}
    Then,
    \begin{align}
        \begin{aligned}
        &\frac1{\mathcal N}\sum_{i=1}^{\mathcal N}
        |W_{N,K}(\alpha_i)-y_i|^2 \\
        = &
        \frac1{\mathcal N}\sum_{i=1}^{\mathcal N}
        |W_{N,K}(\alpha_i)-W(\alpha_i) + W(\alpha_i)-y_i|^2\\
        \le & \frac1{\mathcal N}\sum_{i=1}^{\mathcal N} \left( 2|W_{N,K} (\alpha_i)  - W(\alpha_i)|^2 + 2 |W(\alpha_i) - y_i|^2 \right) \\
        \le & 2 \|W-W_{N,K}\|_\infty^2 + \frac{\epsilon}{2}.
        \end{aligned}
    \end{align} 
    where the first inequality follows from $|a+b|^2 \le 2|a|^2 + 2|b|^2$. The second inequality follows from the definition of the infinity norm and the assumption that all training data satisfy $|W(\alpha_i) - y_i| \leq \sqrt{\epsilon/4}$.
    From Corollary~\ref{corollary:NK_bound}, by choosing 
    \begin{align}
    N,K = O \left( R + \log\frac{s}{\kappa \sqrt{\epsilon}} \right),
    \end{align}
    we have
    \begin{align}
    \|W-W_{N,K}\|_\infty \leq \sqrt{\frac{\epsilon}{8}}.
    \end{align}
    Then, with probability at least $1-\delta$, the empirical error can be bounded as 
    \begin{align}
    \widehat{\mathsf{R}}_{\mathcal T} (\widehat{W}) 
    \le
    \frac{3\epsilon}{4}. 
    \end{align}
    In this case, the feature dimension 
    \begin{align}
    d_f = \mathcal O(N^2K^2) = O\left( \left((R+\log \frac{s}{\kappa \sqrt{\epsilon}}\right)^4 \right).
    \end{align}
    By requiring 
    \begin{align}
    O \left( \frac{t^2}{\sqrt{\mathcal N}} \left( \sqrt{2\log (2d_f)} + \sqrt{\frac{1}{2} \log \frac{1}{\delta}} \right) \right) \le \frac{\epsilon}{4}
    \label{eq:training_data_gabor}
    \end{align}
    and using a union bound, Eq.~\eqref{eq:Gabor_generalization_eq} reduces to $ \mathsf R(\widehat{W}) \le \epsilon$.
    Finally, using $(\sqrt{A} + \sqrt{B})^2 \leq 2 (A+B)$, Eq.~\eqref{eq:training_data_gabor} translates to a sample complexity of 
    \begin{align}
    \mathcal N = O\left( \frac{s^4}{\epsilon^2 \kappa^4} \log \frac{4 (R+ \log (s/\kappa \sqrt{\epsilon}))^8}{\sqrt{\delta}}\right). 
    \end{align}
    for ensuring $\mathsf R(\widehat{W}) \le \epsilon $ with probability at least $1-2\delta$.
\end{proof}

\subsection{Complexity of conventional approaches}
\label{app:complexity_of_conventional_approaches}
For completeness, we discuss the sample complexity of two conventional approaches for estimating Wigner functions. The first approach discretizes the phase space into a dense grid and estimates the Wigner function value at any given point using nearest-neighbor estimation. The second approach performs estimation on a discrete set of points and then employs interpolation over this dataset to obtain an estimate at an arbitrary point. The results presented here follow the discussion in Ref.~\cite{4rf7-9tfx}.

\subsubsection{Dense sampling}

A straightforward way to approximately reconstruct the Wigner function of an unknown state over a finite region $\Omega$ over the phase space is to first discretize the region $\Omega$, for example, into a grid of squares $Q_j$ of equal size $h^2$. Then, at each center of the square, one can estimate the Wigner function value $W_\rho(c_j)$ by measuring displaced parity. When the grid is sufficiently fine-grained, we can use the measured Wigner function values to approximate the continuous Wigner function. 

Specifically, partition $\Omega$ into $\mathcal N$ square cells of side length $h$ and area $\Delta = h^2$. %Then we have $N = \Theta\!\left(\frac{A}{\Delta}\right)$.
Let $c_j$ denote the center of cell $Q_j:=[(c_j)_1-h/2, (c_j)_1+h/2]\times [(c_j)_2-h/2, (c_j)_2+h/2]$.
Then we use the piecewise-constant approximation
$W_{\mathrm{pc}}(\alpha) := W(c_j)$ for $\alpha \in Q_j$ as the estimation of the Wigner function in the entire region $\Omega$.  
For each $c_j$,  one can perform the measurement of displaced parity to get an unbiased estimator $\widehat{W}(c_j)$.
The following proposition tells us how dense this partition on $\Omega$ should be to guarantee that the overall estimation error is upper-bounded.

\begin{proposition}[Sample complexity of dense sampling Wigner function with bounded Hessian]
Let $W$ be the Wigner function of a single-mode quantum state mainly supported on a bounded phase-space region $\Omega\in \mathbb R^2$ of area $A$, beyond which the Wigner function values are negligible. Suppose $\sup_{\alpha \in \Omega} \|H W(\alpha)\| \le \Lambda$,  where $H W(\alpha)$ is the $2\times 2$ Hessian matrix of $W(\alpha)$.
Then the piecewise-constant estimator $W_{\mathrm{pc}}$ of the Wigner function satisfies
\begin{align}
\sum_j\left|\int_{Q_j} (W(\alpha) - W_{\mathrm{pc}}(\alpha)) \, d\alpha\right| \le \varepsilon
\end{align}
with probability at least $1-\delta$, using $\mathcal N= \Theta\!\left(\frac{\Lambda A^2}{\varepsilon}\right)$ different point-wise measurements.
\label{proposition:sample_complexity_single_mode}
\end{proposition}

\begin{proof}
% To show there exists such a reconstruction strategy satisfying the sample complexity, 
By Taylor's theorem with Lagrange form of the remainder, for all $u\in [-h/2,h/2]\times [-h/2, h/2]$, there exists $\xi_u\in Q_j$, such that
\begin{equation}\label{TaylorExp}
W(c_j + u)
=
W(c_j) + \nabla W(c_j)\cdot u + \frac{1}{2} u^\top H W(\xi_u) u,
\end{equation}
$\nabla W(c_j)$ is the gradient of $W$ at $c_j$.
By symmetry,
\begin{align}
\begin{aligned}
&\int_{Q_j} \nabla W(c_j)\cdot u \, du \\ 
=&\nabla W(c_j) \int_{[-h/2,h/2]\times [-h/2, h/2]} u \, du\\ \label{zeroInt}
=& 0.
\end{aligned}
\end{align}
Thus, taking the integral of both sides of Eq.~(\ref{TaylorExp}) over $Q_j$, we have
\begin{align}
\begin{aligned}
&\left|
\int_{Q_j} W(\alpha)\,d\alpha - W(c_j)\Delta
\right|\\ 
=&\left|\int_{Q_j} \nabla W(c_j)\cdot u \, du+\frac{1}{2}\int_{Q_j}  u^\top H W(\xi_u) u\, du\right|\\ \label{cpError}
\le& \frac{\Lambda}{2} \int_{Q_j} \|u\|^2 \, du.
\end{aligned}
\end{align}
where we have used Eq.~(\ref{zeroInt}) and the upper bound $\sup_{\alpha \in \Omega} \|H W(\alpha)\| \le \Lambda$.
For each square cell, we have
\begin{align}
\begin{aligned}
\int_{Q_j} \|u\|^2 du =& \int_{[-h/2,h/2]\times [-h/2,h/2]} (u_1^2+u_2^2) du_1 du_2 \\ 
=&\int_{[-h/2,h/2]\times [-h/2,h/2]} u_1^2 du_1 du_2\\ 
&+ \int_{[-h/2,h/2]\times [-h/2,h/2]} u_2^2 du_1 du_2 \\ 
=&h\cdot \frac{h^3}{12}+ h\cdot \frac{h^3}{12}\\ \label{uInt}
%&= \frac{h^4}{6} \\
=& \frac{\Delta^2}{6}.
\end{aligned}
\end{align}

Hence, by substituting Eq.~(\ref{uInt}) into (\ref{cpError}), we have
\begin{align}
\left|
\int_{Q_j} W(\alpha)\,d\alpha - W(c_j)\Delta
\right|
\le
\frac{\Lambda \Delta^2}{12}.
\end{align}
Summing over all cells, we have
\begin{align}
\sum_j\left|\int_{Q_j} (W(\alpha) - W_{\text{pc}}(\alpha)) \, d\alpha\right| \le \mathcal N \cdot \frac{ \Lambda \Delta^2}{12}
= \frac{\Lambda A \Delta}{12}.
\end{align}
Thus, by choosing 
$\Delta = \Theta\!\left(\frac{\varepsilon}{\Lambda A}\right)$
and 
$\mathcal N =\frac{A}{\Delta}= \Theta\!\left(\frac{\Lambda A^2}{\varepsilon}\right)$, we can ensure that
$\sum_j\left|\int_{Q_j} (W(\alpha) - W_{\text{pc}}(\alpha)) \, d\alpha\right| \le \varepsilon$.

\end{proof}

Since in phase space $A\sim d$ (note that $\hat{n}=\frac{\hat{x}^2+\hat{p}^2-1}{2}$), the required number of phase-space points for reliable reconstruction of Wigner functions scales as $\mathcal O(d^2)$. 

\subsubsection{Interpolation}

Here we briefly explain how the number of phase-space points required to reconstruct Wigner functions through polynomial interpolation scales at least faster than linearly with the truncation dimension $d$. 
\begin{proposition}
Suppose $W(\alpha)$ is the Wigner function of a CV state supported within the region $\Omega\subset \mathbb R^2$ such that it is differential at least $n+1$ times. Denote the maximum $n+1$th derivative of $W$ within the region $\Omega$ as $\Sigma_n:=\max_{\alpha\in \Omega, \|\eta\|=1} D_\beta^{n+1} W(\alpha)$, $D_\beta W$ denotes the derivative of $W$ along the direction of unit vector $\beta\in \mathbb R^2$.  Given the Wigner function values at $\mathcal N$  phase-space points $(\mathcal N\ge n)$ within $\Omega$,  we can use a $n$-degree polynomial $P_n(\alpha)$ to approximate $W(\alpha)$ using polynomial interpolation. To ensure the estimation error $\int_\Omega |W(\alpha)-P_n(\alpha)|d\alpha \le \varepsilon$, the required number of phase-space points scales at least as $\Omega\left(\frac{A^{1+2/n} \Sigma^{2/n}}{n^2}\right)$.
\end{proposition}
\begin{proof}
Suppose the Wigner function $W_\rho(\alpha)$ of a CV quantum state $\rho$ is approximated through an $n$-degree polynomial $P_n(\alpha)$ over a grid of phase-space points with grid width $h$ with maximal amplitude $\alpha_{\text{max}}$. The approximation error at any phase-space point is bounded by the Lagrange remainder formula, i.e.,
\begin{equation}
    |W_\rho(\alpha)-P_n(\alpha)| \le \frac{\Sigma_n h^n}{(n+1)!},\, \forall \alpha \in \Omega,
\end{equation}

To make the above estimation error bounded by $\epsilon$, we  only need to make sure
\begin{equation}
    h\le \frac{\left((n+1)! \epsilon \right)^{1/n} }{\Sigma_n^{1/n}}.
\end{equation}
Without loss of generality, suppose the region $\Omega$ can be partitioned into multiple square cells with each size $h^2$.
Then the number of measured point-wise Wigner function values is
\begin{equation}
  N=\frac{A}{h^2} \ge  \frac{ A \Sigma_n^{2/n}}{\left((n+1)! \epsilon \right)^{2/n} } 
  \in   \Omega\left( \frac{A \Sigma_n^{2/n}}{n^2 \epsilon^{2/n} } \right).
\end{equation}
To ensure $\int_\Omega |W(\alpha)-P_n(\alpha)|d\alpha \le \varepsilon$, we replace $\epsilon$ by $\varepsilon/A$ in the above inequality, and get
the minimum of phase-space points is $\mathcal N=\Omega\left( \frac{A^{1+2/n} \Sigma_n^{2/n}}{n^2 \varepsilon^{2/n} } \right)$.
\end{proof}

Using the relation $\hat{n} = \frac{\hat{x}^2 + \hat{p}^2 - 1}{2}$, we know that $d \sim A$. Consequently, for a fixed $n$-degree polynomial interpolation, the number of phase-space points at which the Wigner function is measured scales at least as $\Omega(d^{1+2/n})$.

\subsection{Network Architecture and Training Details}
\label{app:architecture_training}

This appendix provides additional details on the neural architecture and training procedure used for the single-mode super-resolution task. The model is formulated as a coordinate-conditioned implicit representation of the Wigner function: given a low-resolution input $W_{\mathrm{LR}}$, the network predicts the value of the corresponding high-resolution phase-space function at an arbitrary query coordinate $\mathbf{x}$. The central object learned by the model is therefore a continuous phase-space function, rather than a discrete array tied to a fixed output grid.

The network consists of an encoder that extracts latent features from the low-resolution Wigner function and a decoder that maps the sampled latent feature, together with the query coordinate and the associated cell size, to the Wigner-function value at that point. The encoder is a four-layer two-dimensional residual convolutional network. Each layer consists of a $3\times 3$ convolution, batch normalization, and a ReLU activation, together with a $1\times 1$ projection on the residual branch for channel matching. Starting from a single-channel input, the feature width increases to $64$, $128$, and $256$, with a final latent dimension of $256$. At each query location, the decoder samples the encoder feature map and combines the sampled feature with the two-dimensional coordinate and the corresponding cell size through a multilayer perceptron. An auxiliary shallow branch is used to encode the geometric information carried by the coordinate and cell-size inputs before feature fusion.

The training data consist of precomputed Wigner-function arrays. For each original sample,  both the low-resolution input and the supervision target are generated by resampling the available measurement array to different grid sizes. In the experimental-data setting, this available array is the retained sparse measurement grid, while the dense experimental measurement is used only as an evaluation reference. The model is therefore trained across multiple resolution pairs rather than a single fixed scaling factor. Resampling is performed using spline interpolation. Query coordinates are normalized to the unit interval along each dimension, and the cell size is given by the inverse grid spacing of the target resolution. 

The model is trained with the mean-squared error loss and optimized with Adam. We use a cosine learning-rate schedule with warm-up, and model selection is performed using the validation loss. In addition to the regression loss, we monitor the state fidelity during validation and testing as a physically meaningful measure of reconstruction quality. See Algorithm~\ref{alg:NIRWigner} and Fig.~\ref{fig:diagram_nn} for a summary of the learning procedure of our DNN model.

\begin{algorithm*}
\caption{Self-supervised training of the neural implicit Wigner representation.}
\label{alg:NIRWigner}

\KwData{
Measurement dataset $\mathcal T$ sampled on an $m\times m$ phase-space grid $\Gamma$; 
low-resolution grid sizes $\mathcal M_{\rm low}=\{m-10,\ldots,m-2\}$; 
medium-resolution grid sizes $\mathcal M_{\rm mid}=\{m-1,m\}$; 
maximum number of epochs $E$; 
learning rate $\delta$.
}

\KwResult{
Super-resolved Wigner reconstruction $\widehat W_{\bm\eta}$.
}

Generate paired low- and medium-resolution grids 
$\{(\Gamma_{\rm low}^{(j)},\Gamma_{\rm mid}^{(j)})\}_{j=1}^{J}$ 
from the original measurement grid $\Gamma$ by interpolation-based resampling, such as spline or bilinear interpolation\;

\tcp{The above grid-size choices are illustrative. In general, each target grid $\Gamma_{\rm mid}^{(j)}$ should have higher resolution than its corresponding input grid $\Gamma_{\rm low}^{(j)}$.}

Construct paired datasets
\[
\mathcal T_{\rm low}^{(j)}
=
\{(\alpha,y_\alpha)\}_{\alpha\in\Gamma_{\rm low}^{(j)}},
\qquad
\mathcal T_{\rm mid}^{(j)}
=
\{(\alpha,y_\alpha)\}_{\alpha\in\Gamma_{\rm mid}^{(j)}} .
\]

Initialize the encoder $E_{\bm\theta}$, decoder $D_{\bm\phi}$, and neural surrogate $\widehat W_{\bm\eta}$ with $\bm\eta=(\bm\theta,\bm\phi)$\;

$e=0$\;

\While{$e<E$}{

    \For{$j=1$ \KwTo $J$}{

        Sample a paired low- and medium-resolution dataset
        $(\mathcal T_{\rm low}^{(j)},\mathcal T_{\rm mid}^{(j)})$\;

        Form the low-resolution Wigner array
        \[
        W_{\rm low}^{(j)}
        =
        (y_\alpha)_{\alpha\in\Gamma_{\rm low}^{(j)}} .
        \]

        Compute the latent representation
        \[
        F_{\bm\theta}^{(j)}
        =
        E_{\bm\theta}(W_{\rm low}^{(j)}) .
        \]

        \ForEach{$\alpha\in\Gamma_{\rm mid}^{(j)}$}{

            Compute the local feature
            \[
            \bm f_{\bm\theta}^{(j)}(\alpha)
            =
            \mathcal I(F_{\bm\theta}^{(j)},\alpha) .
            \]

            Predict the Wigner-function value
            \[
            \widehat W_{\bm\eta}(\alpha,W_{\rm low}^{(j)})
            =
            D_{\bm\phi}
            \left(
            \bm f_{\bm\theta}^{(j)}(\alpha),
            \alpha,
            \bm c^{(j)}
            \right),
            \]
            where $\bm c^{(j)}$ denotes the normalized cell size of the target grid $\Gamma_{\rm mid}^{(j)}$\;
        }

        Compute the loss
        \[
        \widehat{\mathsf R}_{\mathcal T_{\rm mid}^{(j)}}(\widehat W_{\bm\eta})
        =
        \frac{1}{|\Gamma_{\rm mid}^{(j)}|}
        \sum_{\alpha\in\Gamma_{\rm mid}^{(j)}}
        \left|
        \widehat W_{\bm\eta}(\alpha,W_{\rm low}^{(j)})
        -
        y_\alpha
        \right|^2 .
        \]

        Update $\bm\eta$ by gradient-based optimization,
        \[
        \bm\eta
        \leftarrow
        \bm\eta
        -
        \delta
        \nabla_{\bm\eta}
        \widehat{\mathsf R}_{\mathcal T_{\rm mid}^{(j)}}(\widehat W_{\bm\eta}) .
        \]
    }

    $e=e+1$\;
}

Evaluate $\widehat W_{\bm\eta}(\alpha)$ on a target grid $\Gamma_{\rm target}$ of the desired resolution\;

\Return{$\{\widehat W_{\bm\eta}(\alpha)\}_{\alpha\in\Gamma_{\rm target}}$}\;

\end{algorithm*}

\subsection{Details of numerical experiments}
\label{app:details_of_numexp}

In this subsection, we provide details about the numerical experiments on both the regression model and the DNN model presented in the main text, along with additional numerical results for comparison between the two learning models. This subsection is divided into three parts: learning simulated states with a regression model; learning simulated states with a DNN model; and learning experimental states with a DNN model.

\subsubsection{Learning simulated states with regression models}
\label{app:regression_numexp}

\begin{figure*}                                                    
    \centering                                               
 
    \includegraphics[width=0.8\linewidth]{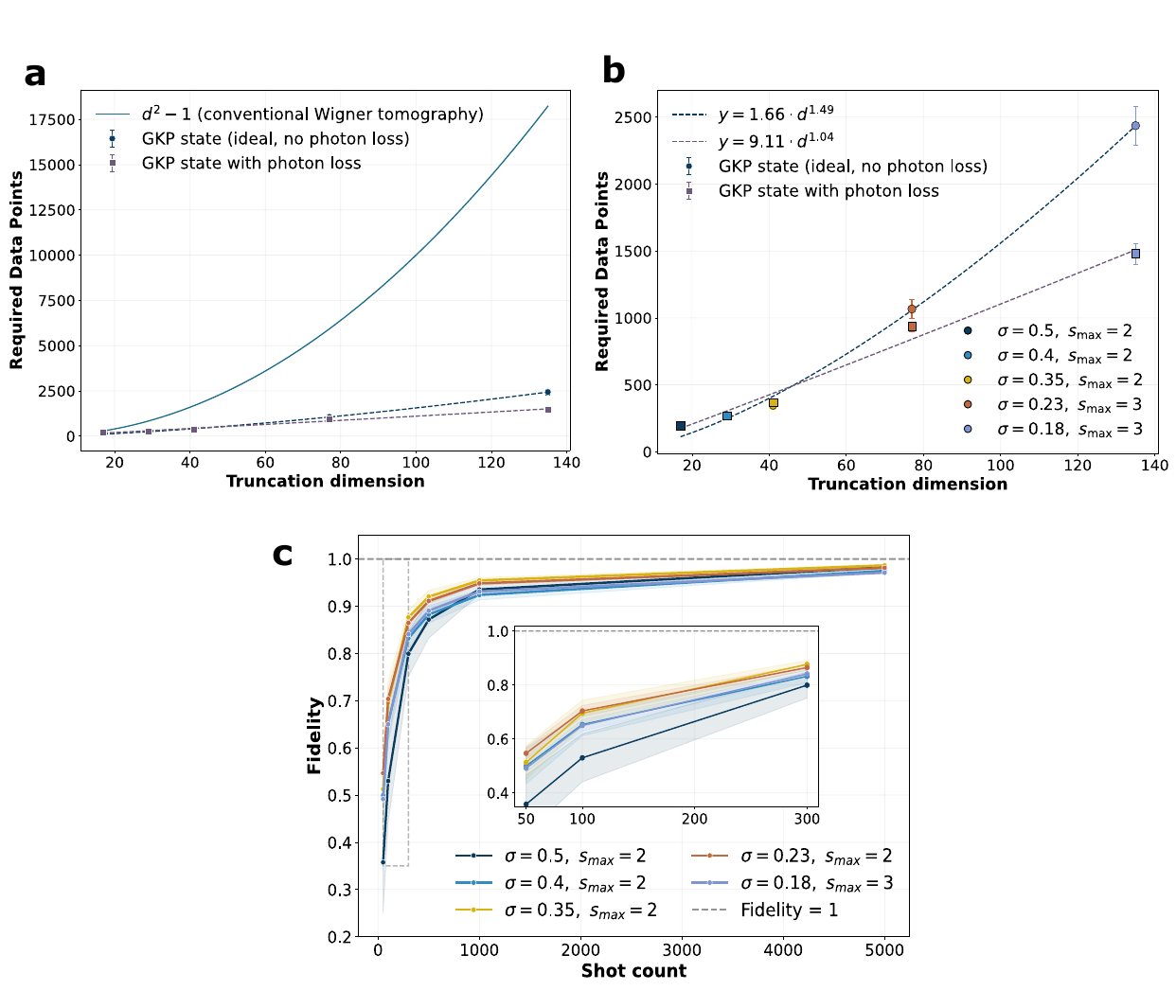} 
    \caption{
Reconstruction of Wigner functions for simulated GKP states using a regression model.
In Subfig.~\textbf{a}, the number of phase-space points required to achieve a fidelity of $0.99$ is compared for a noiseless GKP state and a GKP state subject to photon loss as a function of the truncated Hilbert-space dimension. The scaling required for information completeness in conventional Wigner tomography ($d^2-1$) is shown as a reference.
In Subfig.~\textbf{b}, a detailed view of the data requirements presented in Subfig.~\textbf{a} is shown, illustrating the number of phase-space points needed for reconstructing ideal and lossy GKP states as a function of truncation dimension. Each point represents the mean over five independent random selections of phase-space sampling points, and the error bars indicate one standard deviation.
In Subfig.~\textbf{c}, the fidelity of the reconstructed Wigner function is shown as a function of the number of measurement shots per phase-space point for GKP states with varying squeezing parameters and numbers of Gaussian peaks. Each point represents the mean fidelity over 10 independent experiments, and the shaded region indicates one standard deviation.
}
    \label{fig:regression_ml}
\end{figure*}  
  
\begin{figure*}
    \centering
        \includegraphics[width=0.8\linewidth]{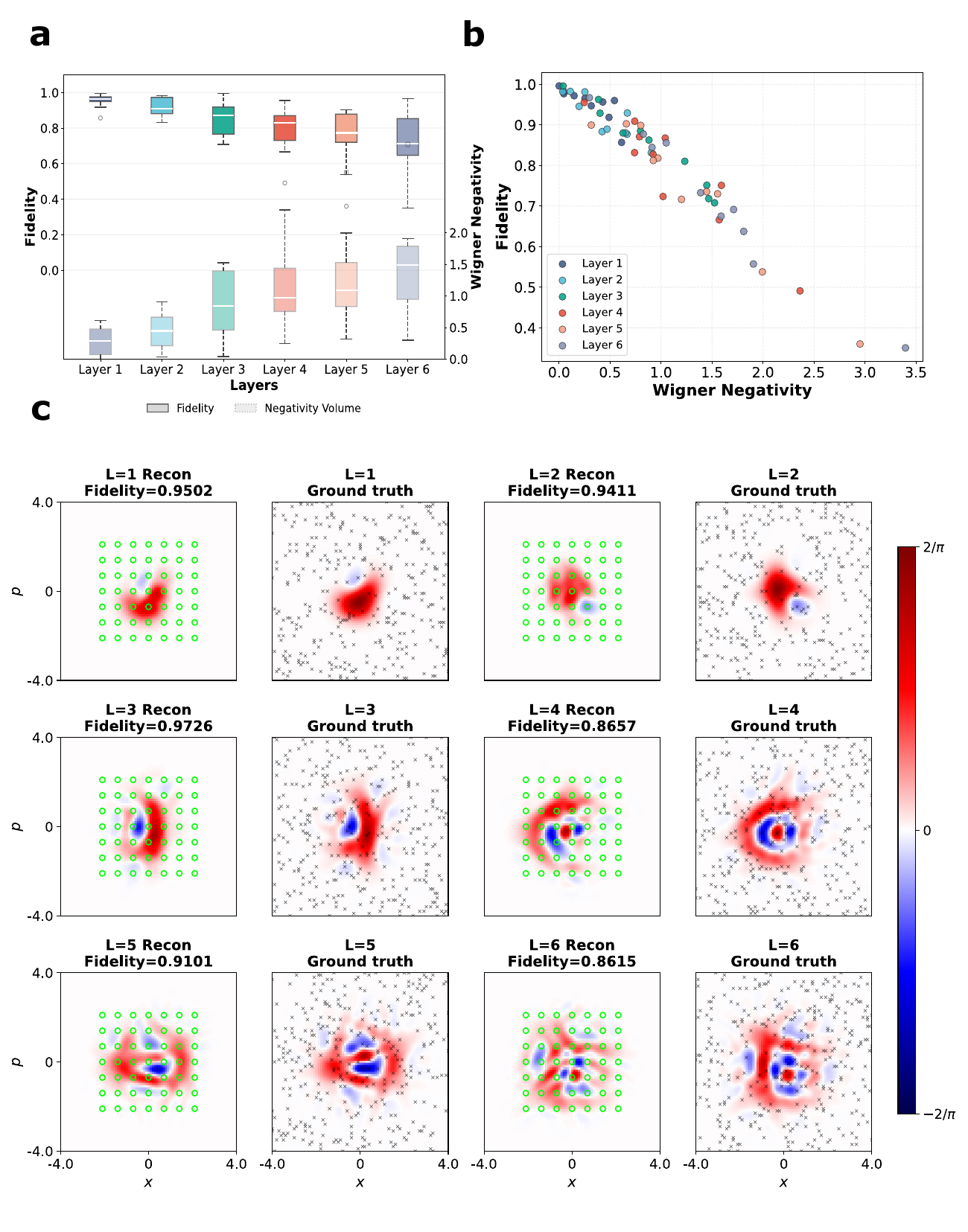}

    \caption{
Performance of the Gabor-frame regression model for output states prepared using random displacement and SNAP circuits.
In Subfig.~\textbf{a}, the reconstruction fidelity is shown as a function of the number of circuit layers. The box plots summarize 10 independently generated output states, with the center line indicating the median and the box indicating the interquartile range.
In Subfig.~\textbf{b}, the reconstruction fidelity is shown as a function of the Wigner negativity of the corresponding output states.
In Subfig.~\textbf{c}, examples of Wigner function reconstructions are shown for states with circuit depths $L$ ranging from one to six using the Gabor-frame regression model. Each pair of images shows the reconstructed Wigner function alongside the ground truth. The measurement data grid is also overlaid on the ground truth plots. The green circles on the left panels correspond to the centers of the Gaussians in the Gabor frame.
}
    \label{fig:layervsFidelityGabor}
\end{figure*}

The numerical study of the regression model is organized around a simple experimental question: how many phase-space sample points are needed for precise reconstruction, and how sensitive is the fit to finite shot noise at each sample point? We answer this question in two stages. We first test the feature map constructed from Wigner functions of Fock basis elements, equivalently, the generalized Laguerre feature map, on states whose structure is sparse on the Fock basis. We then turn to the Gabor-frame feature map and ask how well the same sparse-regression principle works for states as superpositions of a few coherent states.

The first test focuses on binomial code states. These states give a controlled way to separate the Hilbert-space cutoff from the actual Fock support: fixing $N=3$ while increasing $S$ enlarges the truncation dimension $d=(N+1)(S+1)$ without increasing the number of occupied Fock components, whereas fixing $S=3$ while increasing $N$ makes the support itself grow. For each state, we draw training coordinates uniformly at random from $[-6,6]^2$, evaluate the ideal Wigner function using QuTiP, and construct a finite Laguerre design matrix whose rows correspond to sampled phase-space coordinates and whose columns correspond to the retained Laguerre basis features evaluated at those coordinates. We then fit the corresponding coefficients using \texttt{sklearn.linear\_model.Lasso}.

Eq.~\eqref{eq:Lasso} uses the constrained Lasso formulation, where the radius $t$ bounds the $\ell_1$ norm of the coefficient vector and therefore specifies the sparse hypothesis class used in the analysis. In the numerical experiments, the standard scikit-learn implementation instead solves the optimization:
\begin{equation}
    \min_{\bm{\eta}}  \widehat{\mathsf{R}}_{\mathcal{T}}(\widehat{W}_{\bm \eta}) + \lambda \|\bm \eta \|_1,
    \label{eq:Lasso_sklearn}
\end{equation}
which writes the same Lasso tradeoff in penalized form. The training error is minimized together with a $\ell_1$ penalty, and the regularization hyperparameter $\lambda$ determines how strongly sparsity is promoted. While there is no analytical relation between $\lambda$ and $t$ for the two Lasso formulations, the regularization strength $\lambda$ plays the numerical role of choosing the constraint radius $t$: stronger $\lambda$ corresponds to a smaller effective $\ell_1$ budget and therefore a sparser fitted coefficient vector with smaller $t$, and vice versa.

The fidelity $\mathcal F$ is evaluated on a fixed $128\times128$ grid of ideal Wigner values, and we binary search over the number of training points to find the smallest set that reaches $\mathcal F\geq0.99$. Figures~\ref{fig:binomial}\textbf{a} and~\ref{fig:binomial}\textbf{b} show that the required number of points is nearly flat when the Fock support is fixed and increases when the support grows, while remaining far below the $d^2-1$ reference for conventional tomography. Representative reconstructions are shown in Figs.~\ref{fig:binomial}\textbf{c} and~\ref{fig:binomial}\textbf{d}.

We next apply the same Laguerre-feature regression to simulated GKP states. This is a more stringent test because GKP states are not exactly sparse in the Fock basis, even though their Wigner functions are highly structured in phase space. The training coordinates are drawn uniformly at random from $[-3,3]^2$, and the truncation dimension is chosen as the smallest $d$ that satisfies $\mathrm{Tr}(\Pi_d\rho)\geq 1-10^{-4}$, where $\Pi_d=\sum_{n=0}^{d-1}\ket n\bra n$. The same fitting and fidelity-evaluation protocol is then used. Figures~\ref{fig:regression_ml}\textbf{a} and~\ref{fig:regression_ml}\textbf{b} show that this regression model still reconstructs both ideal and lossy GKP Wigner functions from substantially fewer phase-space points than conventional Wigner tomography would require.

The second stage uses the Gabor-frame feature map. We begin with even cat states, which are the canonical sparse coherent-state example. The simulated states are normalized superpositions $\ket{\mathcal C_\alpha^+}\propto\ket{\alpha}+\ket{-\alpha}$ with phase fixed to zero, using the four amplitudes $\alpha=0.5854+0.5854\text{i}$, $0.8684+1.5041\text{i}$, $3.6722\text{i}$, and $5.5673+3.2143\text{i}$. For each amplitude, the phase-space window is chosen to contain the non-negligible support of the Wigner function, and training points are sampled randomly from a circular region inside that window. At these sampled coordinates, we evaluate the retained real-valued Gabor features to construct a finite design matrix, where each row corresponds to a sampled phase-space coordinate, and each column corresponds to a retained Gabor feature evaluated at that coordinate. We then fit the corresponding expansion coefficients using the same penalized Lasso solver from scikit-learn described above.
Fig.~\ref{fig:cat} shows the resulting sample-complexity scaling and representative reconstructions, demonstrating that the Gabor-frame feature map can efficiently capture both the coherent peaks and the interference fringes of the cat states.

Finally, we use the output states of multi-layered random displacement--SNAP circuits as a stress test for the coherent-state feature map. Each circuit layer consists of a displacement gate $D(\alpha)$ followed by a SNAP gate $\text{SNAP}(\bm\theta)=\sum_{n=0}^{d-1}e^{i\theta_n}\ket n\bra n$. We truncate the Hilbert space to $d=30$, choose layers $L\in\{1,2,3,4,5,6\}$, sample displacements with $|\alpha|\leq1.5$, and draw the SNAP phases independently from $[0,2\pi)$. Unlike cat states, these random circuits do not have an explicit small coherent-state support, and increasing the circuit depth produces increasingly structured Wigner functions. Figure~\ref{fig:layervsFidelityGabor} reports the resulting reconstruction fidelity as a function of circuit depth and Wigner negativity, and shows representative reconstructions. The degradation with depth highlights the limitation of using a fixed sparse coherent-state feature map outside its natural regime.

To study finite shot noise, we rerun the regression at several fixed shot counts. For each count, the training target at every sampled phase-space point is the empirical mean of simulated parity measurements after a displacement, as in the main text, so the labels carry statistical noise. The fidelity $\mathcal{F}$ is still evaluated on the fixed ideal Wigner grid using noiseless Wigner values as the reference. Figures~\ref{fig:binomial}\textbf{e},~\ref{fig:binomial}\textbf{f} and \ref{fig:regression_ml}\textbf{c} plot $\mathcal{F}$ versus the number of shots per sampled point for binomial code states and GKP states, respectively.

\subsubsection{Learning simulated states with DNN models}

We first test the DNN model on the same family of simulated even cat states described above. For each value of the coherent amplitude $\alpha$, we train the implicit neural representation from sparse phase-space data and record the minimum number of training points required to achieve $\mathcal{F}\geq 0.99$. As shown in Fig.~\ref{fig:cat_nn}, the required number of data points follows an empirical scaling close to quadratic in the truncation dimension $d$. This behavior contrasts with the regression model designed for CV states sparse in coherent-state support, and reflects that the DNN model does not explicitly use the two-component coherent-state structure of cat states.

\begin{figure}[htbp]
    \centering
    \includegraphics[width=0.5\linewidth]{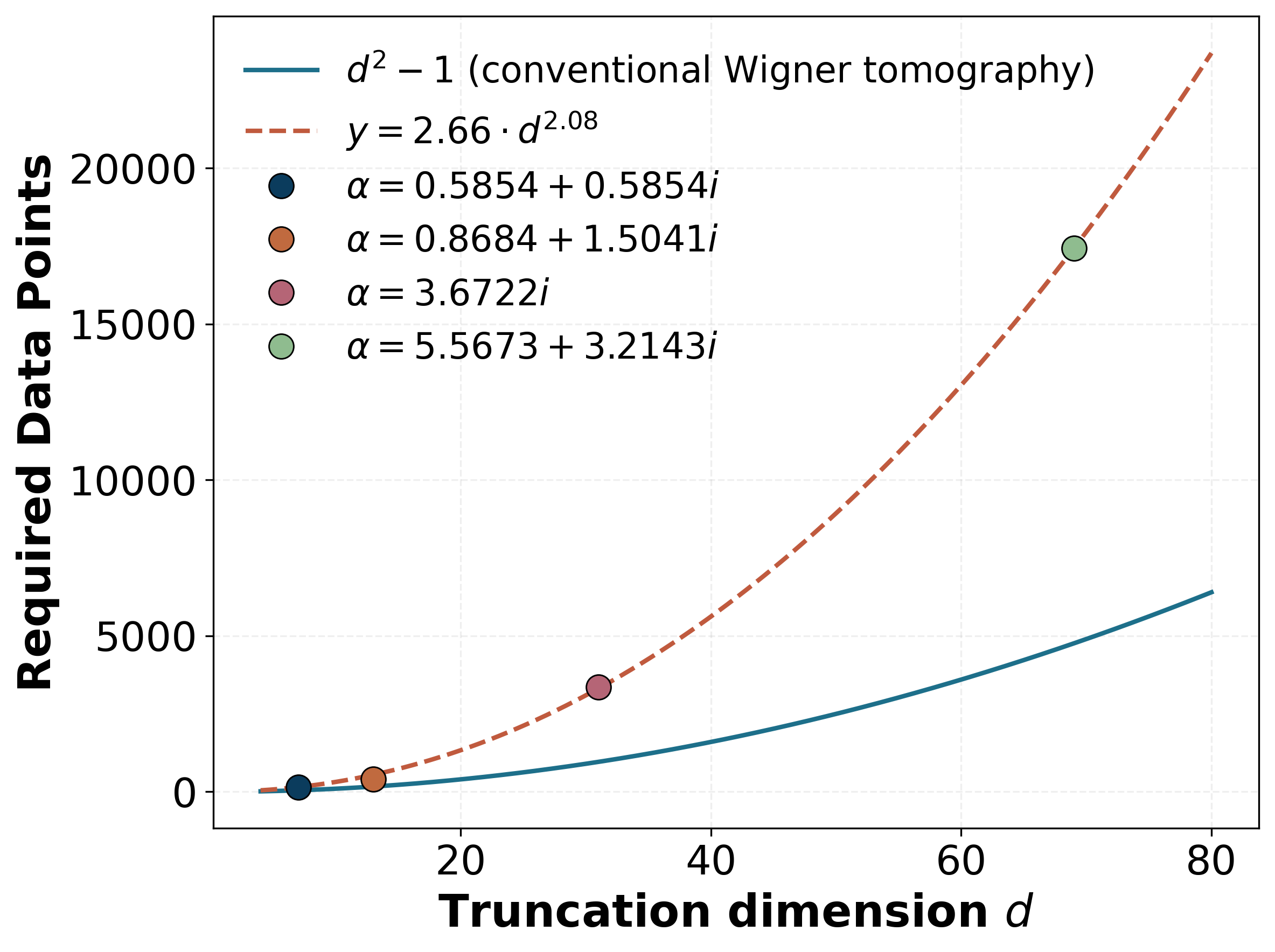}
    \caption{Required number of phase-space training points for the DNN model to reconstruct simulated cat states with fidelity $\mathcal{F}\geq 0.99$. The conventional Wigner tomography scaling $d^2-1$ is shown for comparison.}
    \label{fig:cat_nn}
\end{figure}

For the remaining simulated benchmarks, the neural network experiments reuse the simulated GKP preparation, the same $\mathcal{F}$ evaluation on the fixed $128\times128$ grid of ideal Wigner values from QuTiP introduced in Appendix~\ref{app:regression_numexp}, and the finite shot protocol described there, but replace the regression surrogate with the self-supervised implicit network training of Sec.~\ref{sec:nn_cv_states}. 

For $\ket{\bar{0}}$ states in Eq.~(\ref{state:GKP0}),  we generate low- and mid-resolution training maps by resampling a prescribed simulated measurement grid, train the decoder described in Sec.~\ref{sec:nn_cv_states}, and binary search the training resolution for the smallest grid that achieves  $\mathcal{F}\geq 0.99$. The fidelity is evaluated on the same $128\times128$  grid of ideal Wigner values used in Appendix~\ref{app:regression_numexp}. Note that the $128\times 128$ grid serves only as a dense evaluation mesh, and the trained implicit map can still be queried at arbitrary coordinates in phase space. We report sample complexity versus $d$ relative to $d^2-1$ for noiseless data and for data with photon loss, where the lossy state is obtained by solving the Lindblad master equation in QuTiP with zero Hamiltonian and a single photon-loss collapse operator. In this simulation, the loss rate is set to $0.001$, and the state is evolved until time $20.0$.  We also include a representative MLE runtime for comparison. Figure~\ref{fig:GKP} shows the results obtained in these numerical settings.

To probe states without a fixed algebraic description, we use the same layered displacement--SNAP circuit ensemble and parameter choices as in Appendix~\ref{app:regression_numexp}. For each $L$ we draw ten random instances. The Wigner function of each instance is evaluated on an $18\times18$ grid over $[-5,5]^2$; this resolution is both the measurement footprint used for self-supervised training and the scale at which coarse inputs are formed. After training, we query the model on an $81\times81$ grid and report the state fidelity $\mathcal{F}$ and the Wigner negativity $\int_{\mathbb{R}^2} \text{d}x \text{d}p \, |W_\rho(x, p)|$~\cite{mattia2021}, which is approximated by the corresponding sum over the $81\times81$ predicted Wigner grid. Figure~\ref{fig:layervsFidelity} shows how deeper circuits, which produce more complex phase-space patterns, affect these metrics.

\subsubsection{Learning experimental states with DNN model}

\begin{figure*}[htbp]
    \centering
    \begin{subfigure}[t]{0.9\textwidth}
        \includegraphics[width=\linewidth]{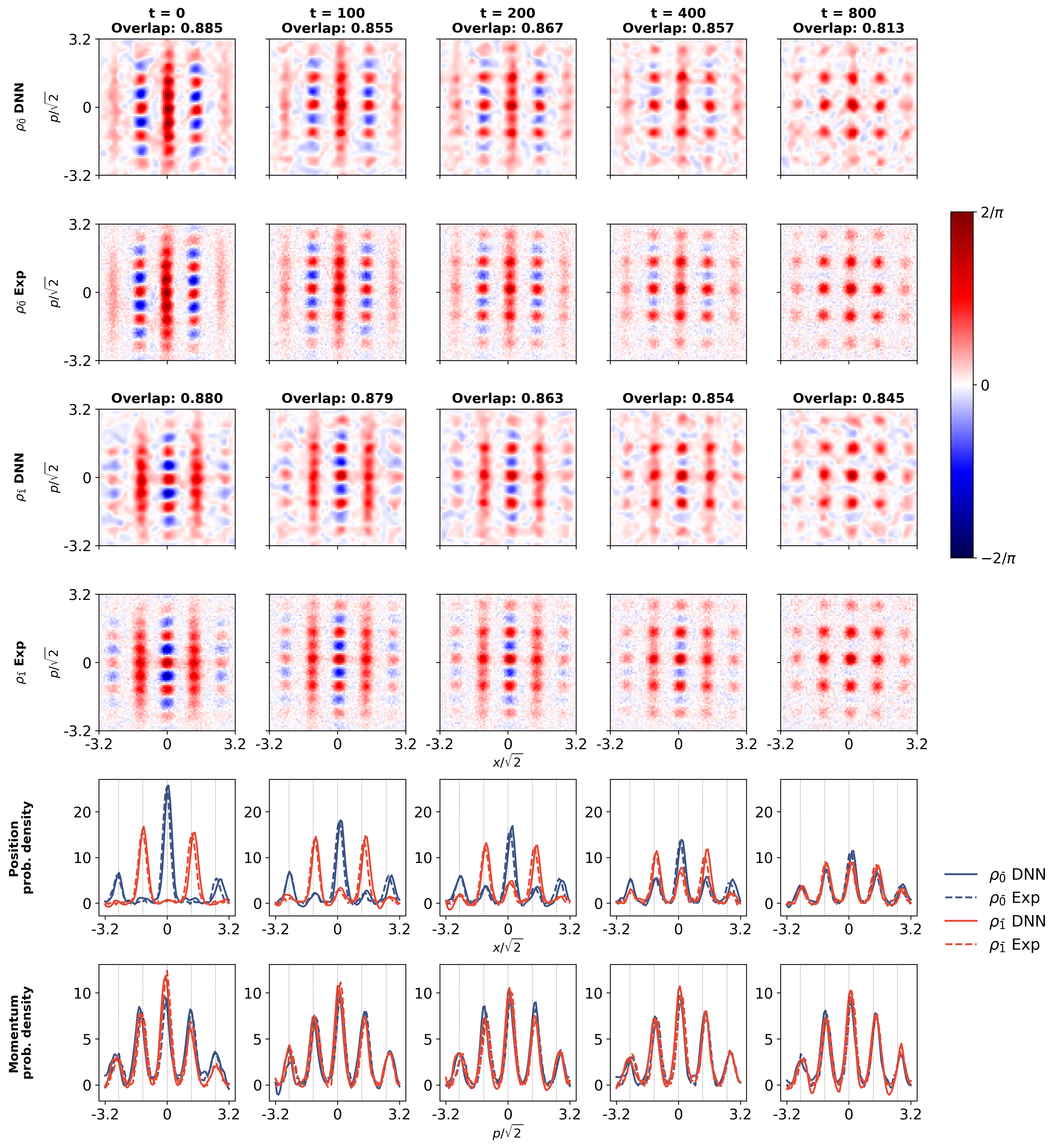}
    \end{subfigure}
    \caption{ The first four rows display the Wigner functions in phase space at QEC cycles $t = 0, 100, 200, 400$, and $800$. 
    Rows 1 and 3 show the DNN reconstructions for the $\rho_{\bar{0}}$ and $\rho_{\bar{1}}$ states, respectively, while Rows 2 and 4 display the corresponding experimental (Exp) data. 
    The overlap between the DNN predictions and the experimental data is annotated above each column. 
    Marginals of the Wigner functions along the momentum (position) quadrature, which yield the probability densities of the oscillator position (momentum), are shown in the fifth (sixth) row. 
    Blue curves correspond to the $\rho_{\bar{0}}$ state, and orange curves represent the $\rho_{\bar{1}}$ state. 
    Solid lines denote the DNN output, whereas dashed lines represent the experimental data. 
    The probability densities are not normalized.}
    \label{fig:wigner_nn_vs_exp}
\end{figure*}

\begin{figure*}[htbp]
    % \centering
\includegraphics[width=0.9\linewidth]{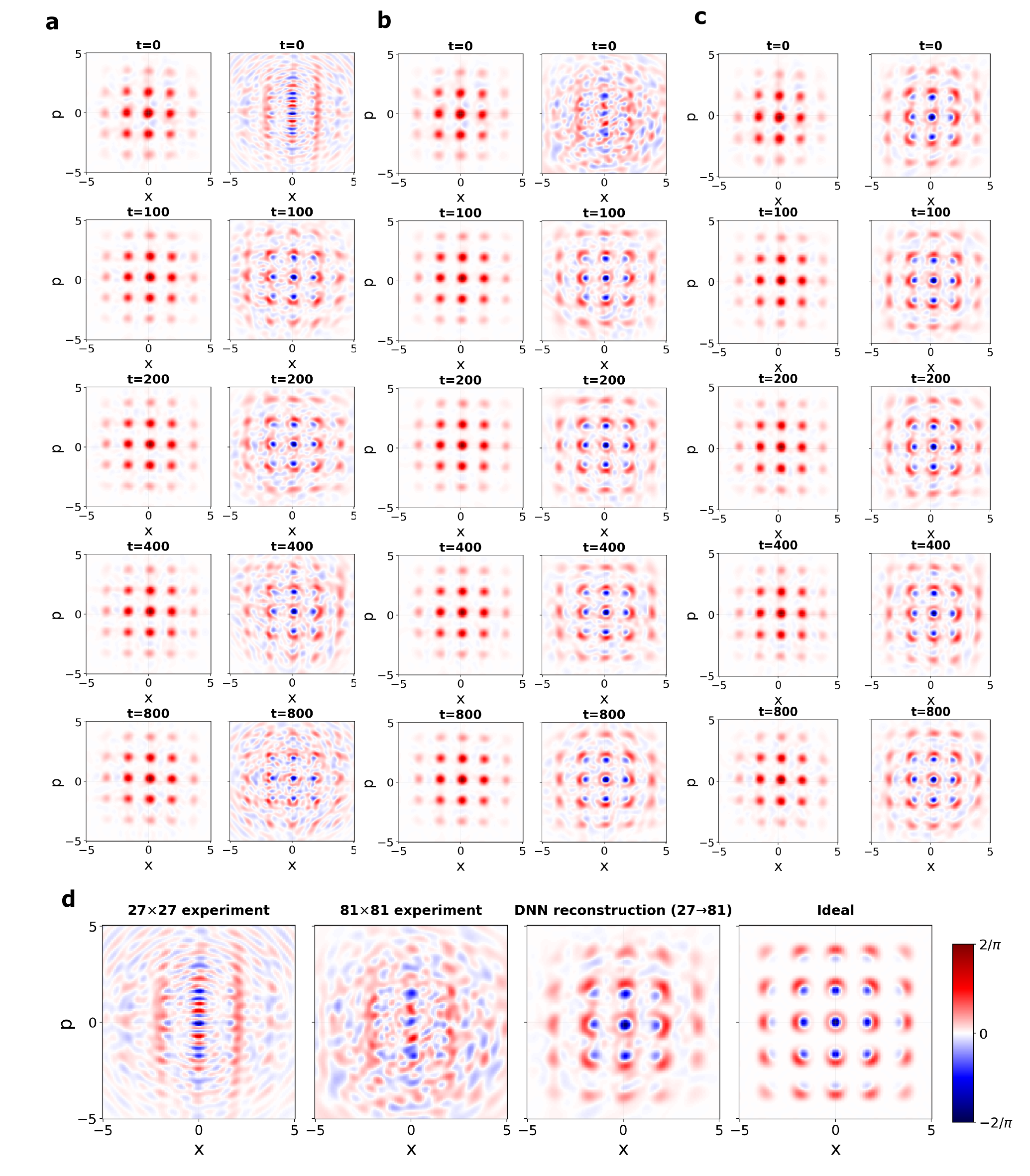}
    
\caption{
Wigner functions of the logical and most probable error subspaces at QEC rounds $t=0,100,200,400,800$.
In Subfig.~\textbf{a}, the projected states are computed from $\rho_{\mathrm{mix}}$ using experimental data on a $27\times27$ grid. The left side shows the Wigner function of the normalized projection onto the experimental logical subspace,
$(\ket{\tilde{0}}\bra{\tilde{0}}+\ket{\tilde{1}}\bra{\tilde{1}})/2$,
and the right side shows the normalized projection onto the most probable error subspace,
$(\ket{E_{\tilde{0}}}\bra{E_{\tilde{0}}}+\ket{E_{\tilde{1}}}\bra{E_{\tilde{1}}})/2$,
obtained using Eq.~(\ref{specDec}).
In Subfig.~\textbf{b}, the same analysis is shown using experimental data on an $81\times81$ grid.
In Subfig.~\textbf{c}, the same analysis is shown using DNN-predicted $81\times81$ Wigner functions reconstructed from the corresponding experimental $27\times27$ measurements.
In Subfig.~\textbf{d}, the extracted most probable error subspace from Subfigs.~\textbf{a}--\textbf{c} is validated at $t=0$. The first three panels show the Wigner functions of
$(\ket{E_{\tilde{0}}}\bra{E_{\tilde{0}}}+\ket{E_{\tilde{1}}}\bra{E_{\tilde{1}}})/2$
obtained from experimental $27\times27$ data, experimental $81\times81$ data, and DNN-predicted $81\times81$ Wigner functions, respectively. The rightmost panel shows the corresponding theoretical single-photon-addition error subspace
$\hat{a}^\dagger(\ket{\bar{0}}+\ket{\bar{1}})(\bra{\bar{0}}+\bra{\bar{1}})\hat{a}$
(up to normalization).
}
    \label{fig:wigner_subspaces}
\end{figure*}

Having explored the DNN model on simulated Wigner functions in the preceding subsection, we now test it on experimental states. We use the same pipeline as in the main text: the implicit architecture and self-supervised super-resolution training of Sec.~\ref{sec:nn_cv_states}, with a sparse measured grid as input and the full experimental data used only as a dense reference for evaluation. The discussion below follows the same order as Sec.~\ref{sec:nn_cv_states}: reconstruction of an experimental GKP state, analysis of the noisy GKP states after multiple rounds of quantum error correction, and reconstruction of a non-Gaussian state prepared with a SNAP gate on a coherent state. In the test of GKP states, the experimental data come from the circuit-QED implementation of Ref.~\cite{sivak2023real}. In every experimental test below, we use the same rule to construct the training data: although the recorded measurements are available on a dense $81\times81$ grid, we retain only every third sample along each axis to obtain a $27 \times 27$ sparse measurement grid for training. The remaining $81 \times 81$ data are kept hidden from the model and used only as the dense experimental reference for evaluation.

\textit{Reconstruction of an experimental GKP state.}
Figure~\ref{fig:gkp_samples} benchmarks the performance of our DNN model on reconstructing experimental GKP states by comparing it against bilinear interpolation, and two regression models in Sec.~\ref{sec:regression_model}. The overlap with the full experimental data is $0.9993$ for the DNN model, indicating that the DNN model preserves the fine lattice structure most faithfully. The density matrix reconstructed from the DNN-predicted $81\times81$ Wigner function via MLE achieves a fidelity of $0.9011$ with the density matrix reconstructed from the full experimental dataset. Training follows the multi-scale recipe: from the measured $27\times27$ map, we form low- and mid-resolution pairs, with input sizes ranging from $20\times20$ to $25\times25$ and target sizes $26\times26$ or $27\times27$; after training, the DNN model predicts the $81\times81$ Wigner function used in the comparison. All subsequent experiments on real measurement data are trained using the same procedure, and we do not repeat these details below.

\textit{Analysis of noisy GKP states after multiple rounds of QEC.}
We first analyze the experimental GKP  states $\rho_{\bar{0}}(t)$ and $\rho_{\bar{1}}(t)$ after QEC cycle counts $t\in\{0,100,200,400,800\}$. Figure~\ref{fig:wigner_nn_vs_exp} compares the DNN-predicted $81\times81$ Wigner functions with the full experimental data at each cycle count. We also plot the marginal distributions of our reconstructed Wigner function in both position and momentum bases, namely the position and momentum probability distributions, together with those obtained from complete experimental data in Fig.~\ref{fig:wigner_nn_vs_exp}. The corresponding MLE-reconstructed density matrices achieve fidelities of $0.8887$, $0.8860$, $0.8403$, $0.8795$, and $0.8779$ for $\rho_{\bar{0}}(t)$, and $0.8935$, $0.8687$, $0.8682$, $0.8531$, and $0.8770$ for $\rho_{\bar{1}}(t)$ at $t=0$, $100$, $200$, $400$, and $800$, respectively, relative to the density matrices reconstructed from the full experimental datasets.

We then combine the $\ket{\bar{0}}$ and $\ket{\bar{1}}$ branches to analyze the logical and most probable error subspaces. This mixture removes the dependence on a particular logical state and allows the leading eigenspaces of the density matrix to identify the experimentally realized logical subspace and the dominant error subspace. For each QEC cycle count, we form
\begin{align}
\rho_{\mathrm{mix}}(t)=\frac{\rho_{\bar{0}}(t)+\rho_{\bar{1}}(t)}{2}.
\end{align}
Following the spectral decomposition in Eq.~\eqref{specDec}, the two leading eigenvectors define the experimental logical subspace, while the next two define the most probable error subspace. 

To examine these inferred subspaces directly, Fig.~\ref{fig:wigner_subspaces}\textbf{a--c} visualize their Wigner functions in phase space. For each QEC cycle, it plots the Wigner functions of the normalized projection onto the experimental logical subspace and the normalized projection onto the most probable error subspace, computed from three inputs: the sparse $27\times27$ experimental grid, the full $81\times81$ experimental data, and the DNN-predicted $81\times81$ Wigner function. This comparison shows how the DNN prediction propagates into the inferred logical and error subspaces, not merely into pointwise Wigner-function overlap.

The most probable error subspace is examined more closely in Fig.~\ref{fig:wigner_subspaces}\textbf{d}. There we focus on the projection $(\ket{E_{\tilde{0}}}\bra{E_{\tilde{0}}}+\ket{E_{\tilde{1}}}\bra{E_{\tilde{1}}})/2$ and compare the Wigner functions obtained from the sparse $27\times27$ data, the full $81\times81$ data, and the DNN prediction. The rightmost panel shows the ideal photon-addition error pattern, proportional to $\hat a^\dagger(\ket{\bar{0}}+\ket{\bar{1}})(\bra{\bar{0}}+\bra{\bar{1}})\hat a$, as a reference for heating noise in the cavity mode. Consistent with the discussion in the main text, the DNN prediction recovers an error-subspace Wigner pattern closer to this photon-addition reference than the direct reconstructions from the sparse or full experimental grids.

Figure~\ref{fig:multiQECcycles} presents the analysis of the reconstructed logical state $\rho_{\bar{0}}(t)$ and compares results obtained from the DNN-predicted Wigner functions with those extracted from the full $81\times81$ experimental measurements.
Figure~\ref{fig:multiQECcycles}\textbf{a} evaluates how well the reconstructed state remains within the GKP logical space. For each reconstructed density matrix, maximum-likelihood estimation is first performed to obtain a physical density operator, after which the squeezing parameter $\sigma$ of the theoretical GKP codewords is optimized to maximize the total overlap with the logical basis states. Rather than considering only the target logical state $\ket{\bar{0}}$, we maximize the sum of the fidelities with the logical basis states $\ket{\bar{0}}$ and $\ket{\bar{1}}$, thereby quantifying the total code-space population irrespective of logical bit flips. The optimal squeezing parameters extracted from the DNN reconstructions closely match those obtained from the full-data reconstructions at all measured times. Moreover, the resulting maximal fidelities show excellent agreement, demonstrating that the DNN prediction preserves the physically relevant code-space structure despite using only one-ninth of the original phase-space samples.
To further characterize the reconstructed states, figure~\ref{fig:multiQECcycles}\textbf{b} shows the two dominant eigenvalues obtained from the spectral decomposition of $\rho_{\bar{0}}(t)$. These leading eigencomponents correspond to the logical GKP basis states $\ket{\bar{0}}$ and $\ket{\bar{1}}$, while their sum provides an estimate of the total probability that the state remains within the logical code space. Although the relative populations of the two logical components evolve throughout the QEC sequence, their combined weight remains above $0.95$ throughout the experiment and exhibits nearly identical trends for the DNN-based and full-data reconstructions. This agreement indicates that the dominant logical subspace is faithfully recovered even from substantially reduced measurement data.
Figure~\ref{fig:multiQECcycles}\textbf{c} compares the full purity $\mathrm{Tr}[\rho_{\bar{0}}(t)^2]$ with the projected purity of the normalized code-space component
\begin{equation}
\rho_{\sigma}
=
\frac{\Pi_{\sigma}\rho_{\bar{0}}(t)\Pi_{\sigma}}
{\mathrm{Tr}\!\left[\Pi_{\sigma}\rho_{\bar{0}}(t)\right]},
\end{equation}
where $\Pi_{\sigma}$ denotes the projector onto the GKP code space spanned by $\ket{\bar{0}}$ and $\ket{\bar{1}}$. As the number of QEC cycles increases, the full purity decreases steadily, reflecting the accumulation of logical mixing and experimental imperfections. In contrast, the projected purity remains substantially higher, indicating that the dominant component of the state remains confined to a nearly pure logical subspace. This behavior further reveals the overall evolution of the encoded states: the experimental GKP code states are initially close to a pure state at $t=0$, while gradually approaching a maximally mixed state within the logical subspace, with a limiting purity of approximately one half.

\textit{Reconstruction of a state prepared with a SNAP gate on a coherent state.}
Finally, we consider an experimental CV state prepared by applying a SNAP gate to a coherent state. Here, a coherent state $\ket{\alpha}$ is obtained by applying a displacement operation $D(\alpha)$ to a vacuum state $\ket{0}$, i.e.\ $\ket{\alpha}=D(\alpha)\ket{0}$.
Specifically, the experimental state under consideration is prepared by applying a $\pi$ phase shift to both the Fock components $\ket{0}$ and $\ket{1}$ of a coherent state $\ket{\alpha}$ with amplitude $\alpha=1$. Note that this state is a special case of states considered in Fig.~\ref{fig:layervsFidelity}.

As Fig.~\ref{fig:experimental_data} shows, using only $27\times 27$ measurement points, the reconstruction overlap of our DNN model is above $0.98$.
This overlap is much higher than that reconstructed by our regression model, demonstrating the stronger robustness of the DNN model against experimental noise beyond GKP states.
 Remarkably, our method requires only approximately 500 measurement points ($27\times 27$), which is significantly fewer than the adversarial generative model~\cite{Ahmed2021PRL}. Furthermore, the training of this DNN model is much more stable.

\begin{figure*}[htbp]
    \centering
\includegraphics[width=0.9\linewidth]{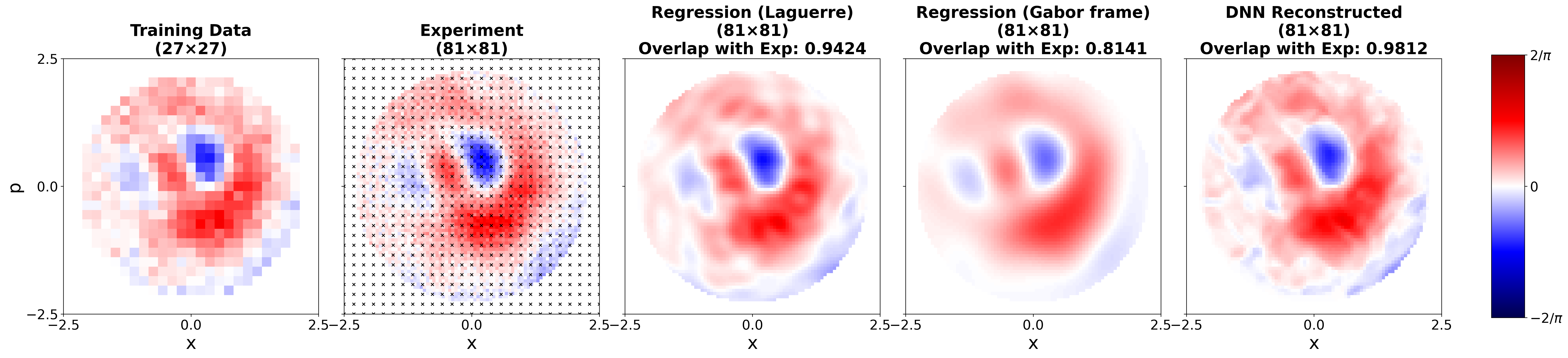}
\caption{Reconstruction of the measured Wigner function of an experimental state prepared by SNAP gate. From left to right: the training dataset downsampled to a $27\times 27$   grid from the original experimental dataset; the full experimental data collected on a $81\times 81$ grid; the Wigner function predicted by the regression models using Laguerre feature map, the Wigner function predicted by Gabor frame feature maps; and the Wigner function predicted by our DNN model. }
    \label{fig:experimental_data}
\end{figure*}

\bibliography{refs}
\end{document}